\documentclass[%
aps,
twocolumn,
superscriptaddress,
amsmath,amssymb,
pra,
floatfix,
]{revtex4-2}
\usepackage{url}
\usepackage{graphicx}% Include figure files
\usepackage{placeins}
\usepackage{dcolumn}% Align table columns on decimal point
\usepackage{bm}% bold math
\usepackage[colorlinks, linkcolor=red, anchorcolor=green, citecolor=blue]{hyperref}
\usepackage{physics}
\usepackage{dsfont}
\usepackage{tikz,pgfplots}
\usepackage{enumerate}
\usepackage{subfigure}
\usepackage{bbold}
\usepackage{makecell}
\usepackage{array}

\makeatletter
\newcommand{\thickhline}{%
    \noalign {\ifnum 0=`}\fi \hrule height 1pt
    \futurelet \reserved@a \@xhline
}
\newcolumntype{"}{@{\hskip\tabcolsep\vrule width 1pt\hskip\tabcolsep}}
\makeatother

\usepackage{multirow}
\usepackage{tabu}
\usepackage{textcomp,booktabs}
\usepackage{colortbl}
\usepackage[T1]{fontenc}
\definecolor{mygray}{gray}{0.9}
\definecolor{mypink}{rgb}{0.99,0.91,0.95}
\definecolor{mycyan}{cmyk}{0.3,0,0,0}

\newcommand{\id}{\text{id}}

\newcommand{\mB}{\mathcal{B}}
\newcommand{\mT}{\mathcal{T}}
\newcommand{\mE}{\mathcal{E}}
\newcommand{\mF}{\mathcal{F}}

\newcommand{\mD}{\mathcal{D}}

\newcommand{\mL}{\mathcal{L}}

\newcommand{\mO}{\mathcal{O}}

\newcommand{\mH}{\mathcal{H}}

\newcommand{\mU}{\mathcal{U}}
\newcommand{\T}{\mathbf{T}}

\newcommand{\1}{\mathbb{1}}

\DeclareMathOperator*{\argmin}{arg\,min}

\makeatletter
\newcommand*{\rom}[1]{\expandafter\@slowromancap\romannumeral #1@}
\makeatother

\makeatletter
\newcommand*{\Relbarfill@}{\arrowfill@\Relbar\Relbar\Relbar}
\newcommand*{\xeq}[2][]{\ext@arrow 0055\Relbarfill@{#1}{#2}}
\makeatother

\usepackage{amsthm}
\newtheorem*{thm*}{Theorem}
\newtheorem{thm}{Theorem}
\newtheorem{lem}{Lemma}
\newtheorem{cor}{Corollary}

\usepackage{tcolorbox}
\tcbuselibrary{theorems}

\newtcbtheorem[number within=section, use counter*=thm]{mythm}{Theorem}%
{colback=red!5,colframe=red!35!red,fonttitle=\bfseries}{thm}

\usepackage{tcolorbox}
\tcbuselibrary{theorems}

\newtcbtheorem[number within=section, use counter*=thm]{mylem}{Lemma}%
{colback=green!5,colframe=green!35!black,fonttitle=\bfseries}{lem}

\usepackage{tcolorbox}
\tcbuselibrary{theorems}

\newtcbtheorem[number within=section, use counter*=thm]{mydef}{Definition}%
{colback=yellow!5,colframe=yellow!70!black,fonttitle=\bfseries}{def}

\usepackage{tcolorbox}
\tcbuselibrary{theorems}

\newtcbtheorem[number within=section, use counter*=thm]{mycor}{Corollary}%
{colback=blue!5,colframe=blue!70!black,fonttitle=\bfseries}{cor}

\usepackage{tcolorbox}
\tcbuselibrary{theorems}

\newtcbtheorem[number within=section, use counter*=thm]{myprop}{Proposition}%
{colback=pink!5,colframe=pink!70!black,fonttitle=\bfseries}{prop}

\pgfplotsset{compat=1.17}

\begin{document}

%\preprint{APS/123-QED}

%%%%%%%%%%%%%%%%%%%%%%%%%%%%%%%%%%%%%%%%%%%%%%%%%%%%%%%%%%%%%%%%%%%%%%%%%%%%%%%%%%%%%%%%%%%%%%%%%%%%%%%%%%%%%%%%%%%%%%

\title{Practical Quantum Broadcasting}% Force line breaks with \\
% \thanks{A footnote to the article title}%

%%%%%%%%%%%%%%%%%%%%%%%%%%%%%%%%%%%%%%%%%%%%%%%%%%%%%%%%%%%%%%%%%%%%%%%%%%%%%%%%%%%%%%%%%%%%%%%%%%%%%%%%%%%%%%%%%%%%%%

\author{Ximing Wang}
\email{canoming.sktt@gmail.com}
\affiliation{Nanyang Quantum Hub, School of Physical and Mathematical Sciences, Nanyang Technological University, Singapore 639798, Singapore}

\author{Yunlong Xiao}
\email{mathxiao123@gmail.com}
\affiliation{Institute of High Performance Computing (IHPC), Agency for Science, Technology and Research (A*STAR), 1 Fusionopolis Way, \#16-16 Connexis, Singapore 138632, Republic of Singapore}
%\affiliation{Quantum Innovation Centre (Q.InC), Agency for Science Technology and Research (A*STAR), 2 Fusionopolis Way, Innovis \#08-03, Singapore 138634, Republic of Singapore}

%%%%%%%%%%%%%%%%%%%%%%%%%%%%%%%%%%%%%%%%%%%%%%%%%%%%%%%%%%%%%%%%%%%%%%%%%%%%%%%%%%%%%%%%%%%%%%%%%%%%%%%%%%%%%%%%%%%%%%

\begin{abstract}
Incorporating sample efficiency, by requiring the number of states consumed by broadcasting does not exceed that of a naive prepare-and-distribute strategy, gives rise to the no practical quantum broadcasting theorem.
To navigate this limitation, we introduce approximate and probabilistic virtual broadcasting and derive analytic expressions for their optimal sample complexity overheads.
Allowing deviations at the receivers restores sample efficiency even in the 1-to-2 approximate setting, whereas probabilistic protocols obey a stronger no-go theorem that excludes all sample efficient 1-to-2 implementations for arbitrary dimension and success probability.
Rather counterintuitive, this obstruction does not persist at larger receiver numbers: for qubit systems, practical 1-to-6 virtual broadcasting becomes attainable.
These results elevate sample complexity from a technical constraint to a defining operational principle, opening an unexplored route to the efficient distribution of quantum information.
%These results establish sample complexity as a defining operational principle and open an unexplored direction for efficient quantum information distribution.
\end{abstract}

\maketitle

%%%%%%%%%%%%%%%%%%%%%%%%%%%%%%%%%%%%%%%%%%%%%%%%%%%%%%%%%%%%%%%%%%%%%%%%%%%%%%%%%%%%%%%%%%%%%%%%%%%%%%%%%%%%%%%%%%%%%%

\noindent \textbf{Introduction}--Distributing information to multiple recipients is a primitive that underlies the scalability of any communication architecture. 
Quantum theory, however, imposes a fundamental restriction~\cite{Wootters1982clone,DIEKS1982clone,RevModPhys.77.1225}: 
transformations whose marginals reproduce noiseless quantum channels (see Fig.~\ref{fig:BC}(a)) are ruled out by linearity of Quantum mechanics~\cite{PhysRevLett.76.2818}, and therefore lie beyond the standard operational framework. 
This limitation has recently been sharpened into a stronger no-go principle, the no practical quantum broadcasting theorem~\cite{z2pr-zbwl,8g6j-w7ld}, which extends the restriction from quantum channels to arbitrary linear maps. 
In particular, no such map can simultaneously satisfy sample efficiency, unitary covariance, permutation invariance, and classical consistency.

What distinguishes practical broadcasting from earlier formulations is the elevation of sample complexity to a defining constraint. 
A protocol that consumes more samples than the naive strategy of preparing independent copies and distributing them to the receivers offers no operational advantage, and is therefore not meaningful in practice. 
This establishes sample efficiency as the natural benchmark against which any virtual implementation must be assessed~\cite{PhysRevLett.132.110203}. 
From this perspective, the incompatibility of complexity with symmetry and consistency requirements prompts an important question: 
can feasibility be restored once some conditions are relaxed? 
At the same time, the idealized regime, where receiver marginals exactly reproduce noiseless quantum channels and protocols succeed deterministically, is seldom accessible in realistic settings.
This motivates extending the framework to approximate (see Fig.~\ref{fig:BC}(b)) and probabilistic (see Fig.~\ref{fig:BC}(c)) implementations, in which small deviations from the ideal marginals or non-unit success probabilities are permitted.

\begin{figure}[ht]
\centering   
\includegraphics[width=0.48\textwidth]{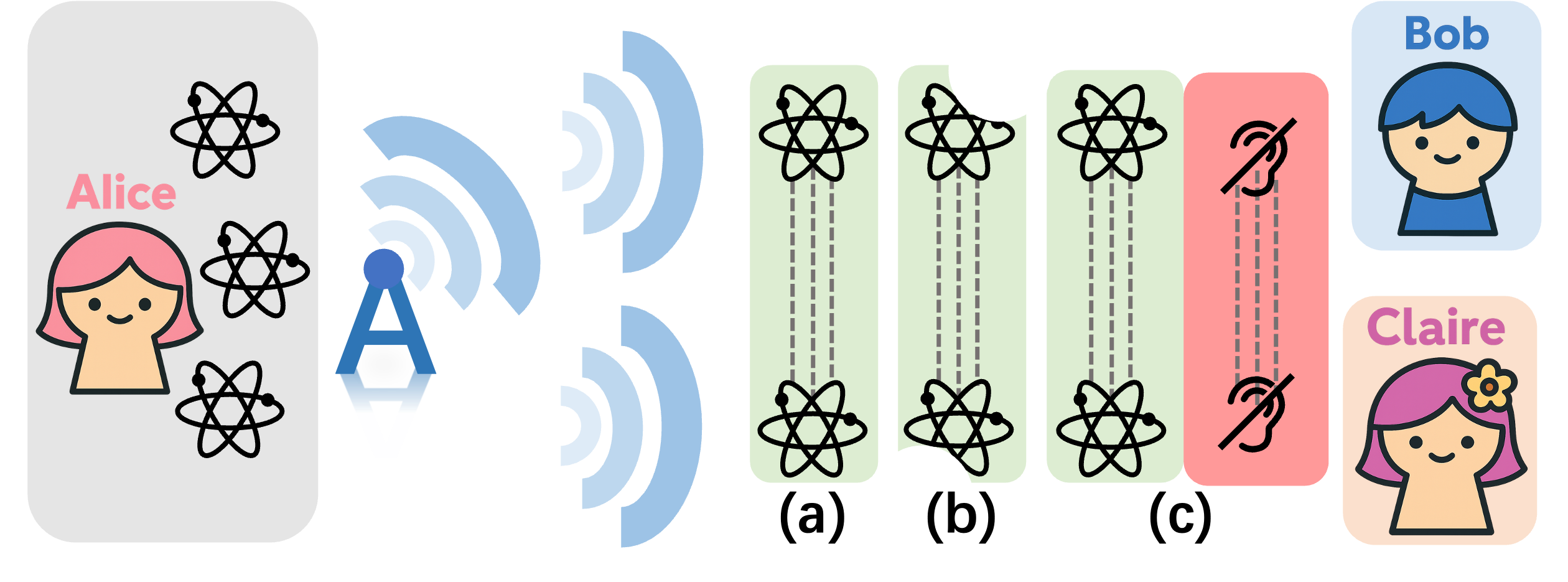}
\caption{(Color online) \textbf{Broadcasting Protocols}.  
(a) Broadcasting: Alice prepares independent copies of the input state and distributes them to the receivers, Bob and Claire. The corresponding process produces outputs whose marginals exactly reproduce the input state.
(b) Approximate Broadcasting: The process generates outputs whose marginals approximate the input state up to a prescribed error.
(c) Probabilistic Broadcasting: Broadcasting succeeds with a generally non-unit probability.
}
\label{fig:BC}
\end{figure}

In this work, we establish frameworks for both approximate and probabilistic virtual broadcasting, and delineate the regimes in which the sample efficiency requirement can be met. 
For approximate virtual broadcasting, we derive an analytic solution for the minimal sample complexity overhead and show that, even in the 1-to-2 setting, error tolerances at the receivers suffice to restore feasibility. 
For probabilistic virtual broadcasting, we obtain an analytic characterization of the sample complexity for general 1-to-$N$ transformations, which in turn yields a strengthened no-go statement: for arbitrary system dimension and success probability, no 1-to-2 virtual protocol can satisfy sample efficiency.
At the same time, we establish that practical 1-to-6 virtual broadcasting is attainable for qubit systems, revealing a qualitative departure from the conventional paradigm, where the impossibility of 
1-to-2 broadcasting excludes all large-scale realizations.
Together, these results demonstrate that once sample complexity is taken into account, the operational landscape of quantum information processing is fundamentally reshaped.

%%%%%%%%%%%%%%%%%%%%%%%%%%%%%%%%%%%%%%%%%%%%%%%%%%%%%%%%%%%%%%%%%%%%%%%%%%%%%%%%%%%%%%%%%%%%%%%%%%%%%%%%%%%%%%%%%%%%%%

\noindent \textbf{Practical Broadcasting}--Consider a sender, Alice, who attempts to broadcast an unknown quantum state to two receivers, Bob and Claire, through $\mE_{A\to BC}$. 
Operationally, broadcasting demands that each receiver obtains a state identical to the original input when examined locally. 
This requirement can be formalized through the following broadcasting condition (BC): $\Tr_{B}\circ\,\mE=\id_{A\to C}$, and $\Tr_{C}\circ\,\mE=\id_{A\to B}$.
Here, $\Tr_{B}$ and $\Tr_{C}$ denote the partial trace over the respective receiver subsystems, while $\id$ denotes the identity channel.

The operational motivation for broadcasting is to reduce the number of copies of an unknown quantum state required for downstream tasks.
Otherwise, independent preparation and distribution would suffice, rendering the task trivial.
To formalize this distinction, recent works~\cite{z2pr-zbwl,8g6j-w7ld} introduced sample efficiency.
Suppose Bob and Claire require $n_1$ and $n_2$ copies to estimate observables $\mO_B$ and $\mO_C$.
A protocol is sample efficient (SE) if the number of input copies it consumes, $n$, satisfies $n<n_1+n_2$.
Failing this, the naive strategy is strictly more efficient.

Three structural constraints naturally arise -- two reflecting symmetry and one enforcing consistency.
Unitary covariance (UC) requires that applying a unitary prior to broadcasting is equivalent to applying it locally afterwards, $\mE\circ\,\mU=\mU\otimes\mU\circ\,\mE$, while permutation invariance (PI) demands invariance under exchange of the receivers, $\mathrm{SWAP}\circ \mE = \mE$.
Classical consistency (CC) further imposes that complete dephasing reduces the transformation to a classical broadcasting map.
The no practical quantum broadcasting theorem shows that even when the admissible transformations are extended to arbitrary linear maps, no protocol can simultaneously satisfy SE, UC, PI, and CC.
Subsequent works~\cite{okada2025virtualphasecovariantquantumbroadcasting} explored virtual broadcasting, replacing UC with weaker symmetries such as phase or flip covariance, yet practicality remains unattained as SE is still violated.

The limitation originates from a fundamental tension between SE and BC: any virtual operation obeying BC necessarily fails SE. 
Let $n_Q:=\max\{n_1,n_2\}$ denote the maximal number of copies required by Bob and Claire to estimate their observables. 
Virtual broadcasting requires $((3d-1)/(d+1))^2n_Q$ input copies~\cite{PhysRevA.110.012458}, where the coefficient depends on the system dimension $d:=\dim A$.
By contrast, the naive strategy of preparing independent copies requires only $2n_Q$ samples. 
Achieving sample efficiency therefore demands $((3d-1)/(d+1))^2<2$, which implies $d<1.522$, incompatible with any quantum system, and thus leads to the following lemma

\begin{lem}\label{lem:MT-no-virtual-broadcasting}
No 1-to-2 virtual operation can simultaneously satisfy BC and SE.    
\end{lem}

This lemma complements, but does not imply, the no practical quantum broadcasting theorem. 
The latter applies to arbitrary linear maps, whereas the present result is restricted to virtual operations.

%%%%%%%%%%%%%%%%%%%%%%%%%%%%%%%%%%%%%%%%%%%%%%%%%%%%%%%%%%%%%%%%%%%%%%%%%%%%%%%%%%%%%%%%%%%%%%%%%%%%%%%%%%%%%%%%%%%%%%

\noindent \textbf{Approximate Broadcasting}--The broadcasting condition is fundamentally incompatible with SE for virtual operations.
In practice, however, perfect transformations are unattainable owing to environmental noise and control imperfections, motivating a relaxation of BC that allows deviations at the receivers and gives rise to approximate virtual broadcasting (ABC).
The central question is whether such a relaxation can reconcile BC with SE while preserving an advantage over protocols restricted to solely quantum operations.
We show that this is indeed the case: in the 1-to-2 setting, ABC emerges as the only paradigm compatible with SE.

Deviations from ideal process are quantified using the channel fidelity $F_{\text{Chan}}$~\cite{Kretschmann_2004}.
For a quantum channel $\mE$ with Choi operator $J^{\mE}$, the normalized Choi state is $\phi^{\mE}:= J^{\mE}/d$, and $F_{\text{Chan}}(\mE):=\Tr[\phi^{\mE}\cdot\phi^{+}]$, with $\phi^{+}$ the maximally entangled state.
Receiver errors $\epsilon_1$ and $\epsilon_2$ thus imposes $F_{\text{Chan}}(\Tr_{C}\circ\,\mE_{A\to BC})=1-\epsilon_1$ and $F_{\text{Chan}}(\Tr_{B}\circ\,\mE_{A\to BC})=1-\epsilon_2$.
Any virtual operation admits a linear decomposition into quantum channels, $a\mE_1 - b\mE_2$, which can be implemented probabilistically with classical outcomes rescaling $a+b$, incurring a sample cost $(a+b)^2 n_Q$.
Optimizing over all such decompositions yields the minimal sample overhead, characterized by the following semidefinite programming (SDP).
\begin{align}\label{eq:MT_SDP_Approximate}
    u_2(\epsilon_1, \epsilon_2):=
    \min \quad 
    & a+b\\
    \text{s.t.} \quad 
    &\Tr[\Tr_{B}[J_1-J_2]\cdot\phi^{+}]/d=1-\epsilon_2,\notag\\
    &\Tr[\Tr_{C}[J_1-J_2]\cdot\phi^{+}]/d=1-\epsilon_1,\notag\\
    & J_1\geqslant0,\,\, 
    \Tr_{BC}[J_1]=a\,\1_A,\notag\\
    &J_2\geqslant0,\,\,
    \Tr_{BC}[J_2]=b\,\1_A,\,a-b=1.\notag
\end{align}

We begin by examining the symmetry structure of the underlying quantum dynamics.
Remarkably, although UC is not imposed in the formulation of ABC, the protocol achieving the optimal sample complexity overhead in Eq.~\eqref{eq:MT_SDP_Approximate} satisfies UC automatically.
To characterize this structure, we introduce the triple-twirling $\mT(X):=\int_{\text{Haar}}dU \,
    \left(U^*\otimes U\otimes U\right)\cdot X \cdot \left(U^{\T}\otimes U^{\dagger}\otimes U^{\dagger}\right)$,
where the integration is taken over the Haar measure on the unitary group.
The feasibility conditions defining $u_2(\epsilon_1, \epsilon_2)$ are invariant under $\mT$:
if $\{a,b,J_1,J_2\}$ is feasible, then so is $\{a,b,\mT(J_1),\mT(J_2)\}$.
This restricts the search for optimal solutions, without loss of generality, to virtual broadcasting $\mE$ invariant under triple-twirling, i.e., $\mT(\mE)=\mE$, leading directly to

\begin{lem}[ABC UC Optimality]\label{lem:MT-ABC-UC-Optimality}
Approximate virtual broadcasting with minimal sample complexity can always be implemented by a UC map.
\end{lem}

\begin{figure*}[t]
    \centering
    \includegraphics[width=1\linewidth]{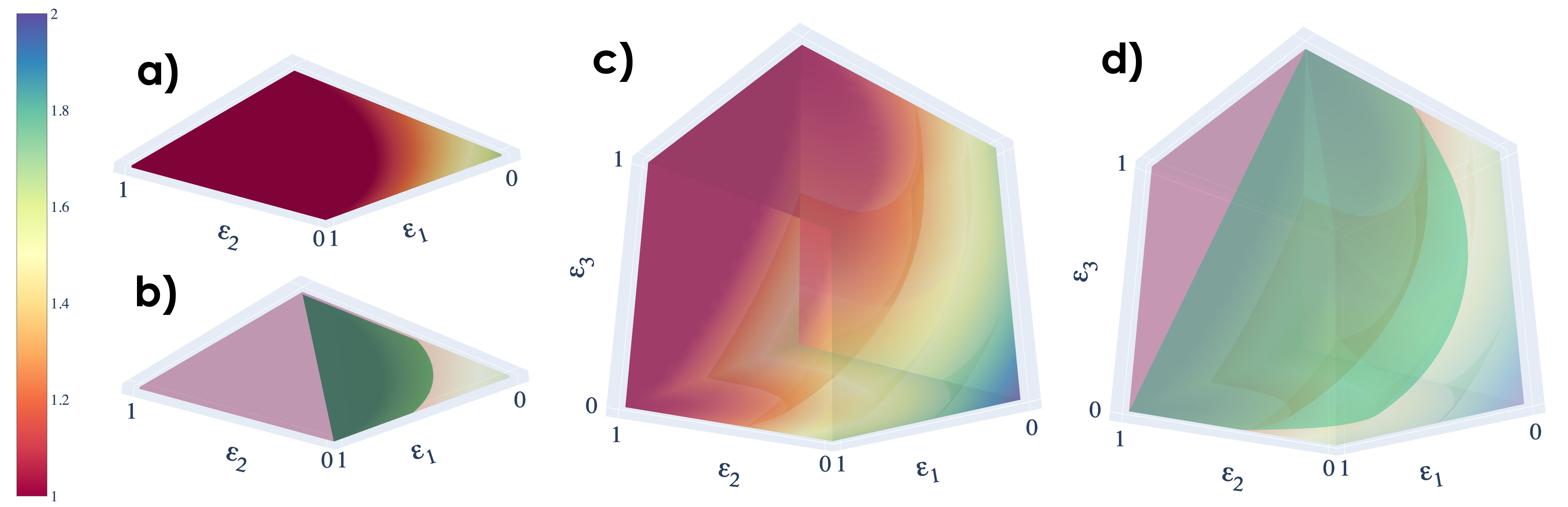}
    \caption{(Color online) \textbf{Practical Approximate Virtual Broadcasting}. 
    Sample complexity overhead for approximate virtual broadcasting as a function of the tolerated errors on the receivers, shown for the 1-to-2 and 1-to-3 settings in panels (a) and (c), respectively. 
    The green shaded regions in panels (b) and (d) delineate parameter regimes where the broadcasting rate exceeds that of the naive prepare-and-distribute strategy, and the protocols are therefore sample efficient.}
    \label{fig:AVB}
\end{figure*}

A consequence of Lem.~\ref{lem:MT-ABC-UC-Optimality} is that tracing out one receiver output system reduces the triple-twirling to the well-known isotropic twirling~\cite{PhysRevA.59.4206}, implying that the corresponding marginal is a depolarizing channel. 
In other words, when optimizing $u_2(\epsilon_1, \epsilon_2)$, one can restrict attention to solutions whose marginals take the form of 
\begin{align}
    \Tr_{B}[J_1-J_2]&=
    \frac{d\epsilon_2}{d^2-1}\1_{AC}
    +
    \frac{d^2(1-\epsilon_2)-1}{d^2-1}
    \Gamma_{AC},\\
    \Tr_{C}[J_1-J_2]&=
    \frac{d\epsilon_1}{d^2-1}\1_{AB}
    +
    \frac{d^2(1-\epsilon_1)-1}{d^2-1}
    \Gamma_{AB},
\end{align}
where $\Gamma:=d\phi^{+}$ denotes the unnormalized maximally entangled state.
All proofs of the results presented in the main text are provided in the Appendix.

Our main result for ABC establishes an analytic expression for the optimal sample complexity overhead, formalized as

\begin{thm}[ABC Optimal Sample Complexity]\label{thm:MT-ABC-Optimal-Sample-Complexity}
The minimal sample complexity overhead $u_2(\epsilon_1, \epsilon_2)$ for ABC is given by
\begin{align}\label{eq:MT-Thm-ABC-Optimal-Sample-Complexity}
    \max\left\{\frac{d^2(3-2(\epsilon_1+\epsilon_2))-4d\sqrt{(1-\epsilon_1)(1-\epsilon_2)}+1}{d^2-1},1\right\}.
\end{align}
Here $d$ denotes the system dimension.
\end{thm}
 
It is also worth noting that the present result naturally includes the error-free limit as a special case.
Setting $\epsilon_1=\epsilon_2=0$, the first term in Eq.~\eqref{eq:MT-Thm-ABC-Optimal-Sample-Complexity} reduces to $(3d-1)/(d+1)$, which exceeds 1 for all $d\geqslant2$, thereby recovering~\cite{PhysRevA.110.012458}.

We complete the framework of ABC by revisiting SE, which in this setting requires redefinition.
The naive strategy yields noiseless marginals for each receiver, whereas approximate virtual broadcasting is intrinsically noisy: even at optimality, its marginals are necessarily depolarizing, and this noise cannot be suppressed by increasing the number of samples.
Consequently, judging performance solely by total sample consumption is inadequate; 
benchmarks such as $2n_Q$ fail to capture the interplay between resource usage and operational performance.
To resolve this, we introduce the {\it broadcasting rate}, defined as the average fidelity generated per sample, normalized by $1/n_Q$, thereby elevating efficiency per resource as the relevant figure of merit, in direct analogy with power in classical mechanics and information rate in Shannon theory.

For the naive protocol, the corresponding broadcasting rate is defined as $R_{\text{naive}}:=((1+1)/2)/(n_Q+n_Q)=1/2$.
Here $n_Q$ serves as a unit, but we keep it explicitly in the formulation to preserve the physical interpretation: 
the protocol distributes $n_Q$ copies of the sample to Bob and an additional $n_Q$ copies to Claire.
For approximate virtual broadcasting, we denote the corresponding quantity by $R_{\text{ABC}}:=((1-\epsilon_1+1-\epsilon_2)/2)/u^2_2(\epsilon_1, \epsilon_2)n_Q$, and define a protocol to be SE whenever $R_{\text{ABC}}\geqslant R_{\text{naive}}$.
Our numerical experiments establish the existence of such protocols in the 1-to-2 setting, as illustrated in Fig.~\ref{fig:AVB}(a) and~\ref{fig:AVB}(b). 
This advantage persists beyond the minimal configuration: 
in the 1-to-3 scenario, SE remains accessible, as shown in Fig.~\ref{fig:AVB}(c) and~\ref{fig:AVB}(d). 
A systematic treatment of the general 1-to-$N$ case is provided in the Appendix.

\begin{cor}[1-to-2 Practical ABC]\label{cor:MT-Practical-Approximate-Virtual-Broadcasting}
1-to-2 approximate virtual broadcasting admits implementations that satisfy SE, and is therefore practically viable.  
\end{cor}

%A central implication of Thm.~\ref{thm:MT-Optimal-Sample-Complexity} is that it identifies the regime in which practical advantages can arise over both the naive protocol and broadcasting strategies that rely solely on quantum channels.

%%%%%%%%%%%%%%%%%%%%%%%%%%%%%%%%%%%%%%%%%%%%%%%%%%%%%%%%%%%%%%%%%%%%%%%%%%%%%%%%%%%%%%%%%%%%%%%%%%%%%%%%%%%%%%%%%%%%%%

\noindent \textbf{Probabilistic Broadcasting}--Beyond ABC, one may instead relax the requirement of unit success probability in the distribution of quantum information.
This leads to the notion of probabilistic virtual broadcasting (PBC), where the task succeeds only with a given probability, and we examine whether such a relaxation can yield a practically implementable form of quantum information distribution.
Two lemmas will play a central role in shaping the analysis and simplifying the formulation of the problem. 
The first lemma concerns the local success probability for each receiver, and states that

\begin{lem}[Local Success Probability]\label{lem:MT-Local-Success-Probability}
In probabilistic broadcasting, the success probability associated with each receiver is independent of the input state.
\end{lem}

The above lemma characterizes the structure of the local success probability: taking Bob's side as an example, if $\Tr_C[\mE(\rho_1)]=p_1\rho_1$ and $\Tr_C[\mE(\rho_2)]=p_2\rho_2$, then probabilistic broadcasting enforces $p_1=p_2$, so that the success probability is independent of the input state.
The second lemma characterizes the global success probability for each receiver, given by

\begin{lem}[Global Success Probability]\label{lem:MT-Global-Success-Probability}
In probabilistic broadcasting, the success probabilities across all receivers are necessarily identical.
\end{lem}

This lemma identifies the structure of the global success probability: 
if $\Tr_C\circ\,\mE=p_{\text{Bob}}\id_{A\to B}$ and $\Tr_C\circ\,\mE=p_{\text{Claire}}\id_{A\to C}$, then we have $p_{\text{Bob}}=p_{\text{Claire}}$.
As shown in the Appendix, both lemmas hold for quantum channels as well as for virtual operations.
Taken together, they significantly simplify the analysis of PBC, enforcing that the success probability is uniform across input states at the local level and identical across receivers at the global level.
With this structure in place, the minimal sample complexity overhead for general 1-to-$N$ PBC, where a sender Alice distributes quantum information to receivers $B_1$, $\ldots$, $B_N$, can be expressed as
\begin{align}\label{eq:MT-PVB-SDP-N}
    s_{N}(p)
    =
    \min \quad 
    & 
    a+b
    \\
    \text{s.t.} \quad 
    &\Tr_{B_1\cdots B_{i-1}B_{i+1}\cdots B_{N}}
    [J_1-J_2]=p\Gamma_{AB_i},\,\forall i,\notag\\
    &J_1\geqslant0,\,\, \1_{B_1\cdots B_N}\star J_1\leqslant a\1_{A},\notag\\
    &J_2\geqslant0,\,\, \1_{B_1\cdots B_N}\star J_2\leqslant b\1_{A},\,a-b\leqslant1.\notag
\end{align}
The constraint $a-b\leqslant1$ ensures that the resulting virtual operation is trace-nonincreasing (TNI), and therefore represents a probabilistic transformation.

We next analyze the symmetry structure of PBC at minimal sample complexity.
In analogy with ABC, optimal protocols can, without loss of generality, be taken to be invariant under triple-twirling, as captured by

\begin{lem}[PBC UC Optimality]\label{lem:MT-PBC-UC-Optimality}
Probabilistic virtual broadcasting at minimal sample complexity can always be realized by a UC map.
\end{lem}

Our central result for PBC is a closed-form expression for the optimal sample-complexity overhead.

\begin{thm}[PBC Optimal Sample Complexity]\label{thm:MT-PBC-Optimal-Sample-Complexity}
The minimal sample-complexity overhead $s_{N}(p)$ for PBC with success probability $p$ is given by
\begin{align}\label{eq:MT-Thm-Optimal-Sample-Complexity}
p\left(\frac{2dN}{N+d-1}-1\right),
\end{align}
where $d=\dim A=\dim B_i$ is the system dimension.
\end{thm}

In the probabilistic setting, the protocol is allowed to realize noiseless marginal channels with nonzero success probability, thereby placing it on comparable footing with the naive strategy.
This feature distinguishes PBC from ABC: whereas the latter is assessed via a broadcasting rate, here it is more meaningful to quantify performance in terms of the total number of samples directly. 
For the naive protocol, this cost scales as $n_{\text{naive}}:=Nn_Q$, corresponding to distributing $n_Q$ copies to each of the $N$ receivers.
In contrast, for PBC with success probability $p$, the rescaling of measurement outcomes contributes a factor $s^2_{N}(p)/p^2$, while conditioning on successful events introduces an additional overhead $1/p^2$. 
Together, these yield an overall sample complexity $n_{\text{prob}}(p):=s^2_{N}(p)n_Q/p^4$.
We therefore regard a protocol as SE whenever $n_{\text{prob}}(p)<n_{\text{naive}}$.
With the SE benchmark in place, we introduce the following no-go theorem.

\begin{figure}[t]
    \centering
    \includegraphics[width=0.95\linewidth]{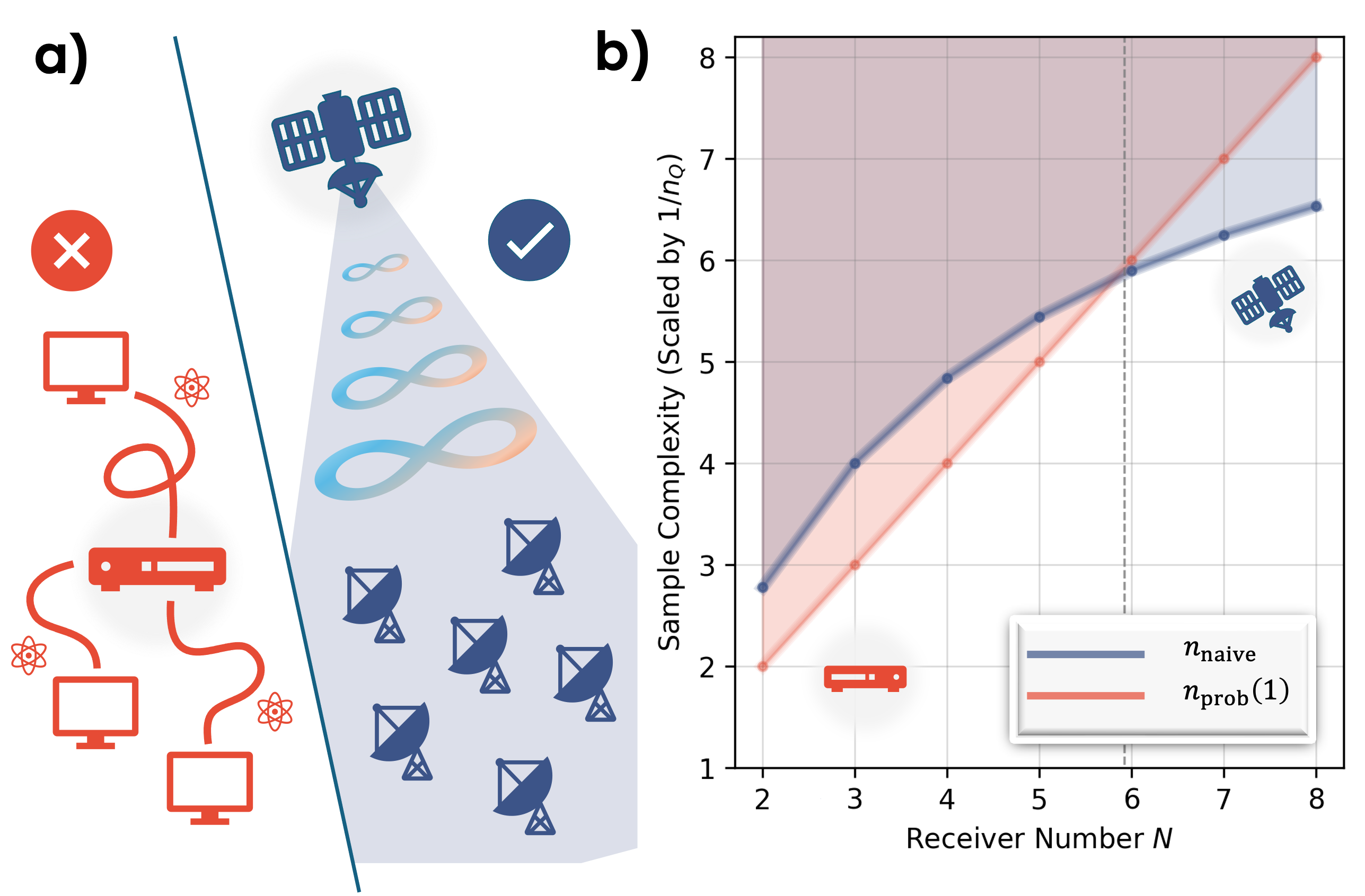}
    \caption{(Color online) \textbf{Broadcasting Across Scales}
    (a) SE rules out 1-to-2 PBC for arbitrary success probability and system dimension. In contrast, for qubit systems, 1-to-6 virtual broadcasting satisfies SE, even in the deterministic limit.
    (b) Comparison of sample requirements for the naive prepare-and-distribute protocol (red) and PBC with unit success probability (blue). Once the number of receivers exceeds six, the protocol enters the SE regime.}
    \label{fig:1-to-6}
\end{figure}

\begin{cor}[No 1-to-2 Practical PBC]\label{cor:MT-No-Practical-Probabilistic-Virtual-Broadcasting}
For all success probabilities and all system dimensions, 1-to-2 PBC fails to meet SE, and is thus intrinsically impractical.
\end{cor}

For qubit systems at unit success probability, that is, in the deterministic limit of $N=2$, virtual broadcasting cannot satisfy SE as $n_{\text{prob}}(1)=s^2_{2}(1)n_Q\approx2.778n_Q>2n_Q$.
Retaining the qubit setting while increasing the number of receivers reveals a different behavior. 
At larger scale, for $N=6$, Thm.~\ref{thm:MT-PBC-Optimal-Sample-Complexity} gives $n_{\text{prob}}(1)=s^2_{6}(1)n_Q\approx5.898n_Q<6n_Q$, as demonstrated in Fig.~\ref{fig:1-to-6}, placing the protocol within the SE threshold.

%which lies within the SE threshold, indicating the fourth result

\begin{cor}[1-to-6 Practical BC]\label{cor:MT-Practical-Probabilistic-Virtual-Broadcasting}
For qubit systems, there exist 1-to-6 virtual broadcasting protocols that satisfy SE, thereby establishing their practical viability.
\end{cor}

At first sight, this behavior appears to conflict with the standard intuition of quantum information theory, where resource considerations such as sample complexity are absent:
in that setting, the impossibility of 1-to-2 broadcasting excludes any 1-to-6 realization, since tracing out four subsystems of a valid 1-to-6 map would yield a legitimate 1-to-2 broadcasting.
Once SE is imposed, however, the picture changes:
no 1-to-2 virtual protocol satisfies SE, yet for qubit systems 1-to-6 virtual schemes can. 
The resolution lies in the fact that the reduction argument does not extend to the resource level; tracing out subsystems preserves the broadcasting condition, but does not lower the associated sample complexity. 
This separation highlights a nontrivial interface between quantum information and complexity, and points to the need for a deeper understanding of quantum protocols through the lens of sample complexity.

%At first sight, this behavior appears to contradict the standard intuition from quantum information theory, where resource considerations such as sample complexity are absent. 
%In that framework, the impossibility of 1-to-2 broadcasting automatically precludes any 1-to-6 realization: otherwise, tracing out four subsystems of a valid 1-to-6 broadcasting would yield a legitimate 1-to-2 broadcasting channel.

%The situation is markedly different for virtual broadcasting once SE is imposed. 
%No 1-to-2 virtual protocol can satisfy the SE requirement. 
%Yet, for qubit systems, there exist 1-to-6 virtual broadcasting schemes that are compatible with SE. 
%This apparent tension is resolved by noting that the reduction argument fails at the level of resources: tracing out subsystems preserves BC but not its sample complexity. 
%In particular, realizing the 1-to-6 protocol requires at least $5.898n_Q$ samples, which remains strictly above both the $2.778n_Q$ threshold associated with 1-to-2 virtual broadcasting and the baseline $2n_Q$.
%Consequently, tracing out subsystems cannot produce a practical virtual broadcasting protocol, and no contradiction with conventional quantum information theory arises.

%%%%%%%%%%%%%%%%%%%%%%%%%%%%%%%%%%%%%%%%%%%%%%%%%%%%%%%%%%%%%%%%%%%%%%%%%%%%%%%%%%%%%%%%%%%%%%%%%%%%%%%%%%%%%%%%%%%%%%

\noindent \textbf{Discussions}--To assess the practicality of virtual broadcasting, understood as compatibility with SE, we have developed frameworks for both ABC and PBC.
For ABC, we derived an analytic expression for the optimal sample complexity overhead, showing that in the 1-to-2 setting, allowing deviations at the receivers suffices to restore practicality.
For PBC, we extended the analysis to general 1-to-$N$ transformations and obtained a closed-form characterization of the corresponding sample complexity, from which a strong no-go theorem follows:
no 1-to-2 probabilistic protocol satisfies SE for arbitrary success probability and system dimension.
Strikingly, this limitation does not extend to larger scales: for qubit systems, 1-to-6 virtual broadcasting can satisfy SE. 
A clear departure from the standard intuition of quantum information theory emerges: the impossibility of 1-to-2 broadcasting no longer excludes larger-scale realizations.
These results point to a fundamental interplay between sample complexity and quantum information processing, and call for a deeper exploration of their intersection.

A detailed analysis of both frameworks reveals a common structural feature: for both ABC and PBC, optimal protocols at minimal sample complexity can always be taken to satisfy UC.
In the probabilistic setting, the uniformity of the success probability across receivers further enforces PI, while for ABC, PI can be imposed whenever the error tolerances at the receivers are symmetric.
We also identify the regimes in which SE is attainable. 
Against this backdrop, a single condition remains fundamentally incompatible with SE, namely, CC.
Imposing CC uniquely singles out a linear map satisfying UC, PI, and CC, the canonical virtual broadcasting map, which, by the no practical quantum broadcasting theorem, necessarily violates SE.
This delineates a fundamental boundary: broadcasting can be reconciled with symmetry and resource constraints, yet resists full compatibility with classical consistency, leaving the prospect of a unified framework bridging the quantum and classical worlds out of reach.

%%%%%%%%%%%%%%%%%%%%%%%%%%%%%%%%%%%%%%%%%%%%%%%%%%%%%%%%%%%%%%%%%%%%%%%%%%%%%%%%%%%%%%%%%%%%%%%%%%%%%%%%%%%%%%%%%%%%%%

\section*{Acknowledgments}
Yunlong Xiao acknowledges Zhenhuan Liu, Kun Fang, and Yuxiang Yang for stimulating discussions during Quantum Resources 2026 workshop in Tokyo.  
This research is supported by A*STAR under its Career Development Fund (C243512002).

%%%%%%%%%%%%%%%%%%%%%%%%%%%%%%%%%%%%%%%%%%%%%%%%%%%%%%%%%%%%%%%%%%%%%%%%%%%%%%%%%%%%%%%%%%%%%%%%%%%%%%%%%%%%%%%%%%%%%%

\section*{~}

{\it Note added.}
See also independent work on approximate virtual broadcasting, where the problem is formulated using the diamond norm~\cite{Matthew2026}. 
Our approach differs in both physical scope and mathematical treatment.
From a physical perspective, we explore both approximate and probabilistic virtual broadcasting and uncover a counterintuitive phenomenon: although no practical 1-to-2 virtual broadcasting protocol exists, practical 1-to-6 schemes do arise for qubit systems, marking a departure from conventional quantum broadcasting. 
From a mathematical perspective, in the approximate regime we begin with a semidefinite programming (SDP) characterization of the minimal sample complexity overhead, reduce it to a second-order cone programming (SOCP), and ultimately obtain an analytic solution. 
In the probabilistic setting, the SDP further simplifies to a linear programming (LP), enabling a closed-form characterization of the optimum.

%%%%%%%%%%%%%%%%%%%%%%%%%%%%%%%%%%%%%%%%%%%%%%%%%%%%%%%%%%%%%%%%%%%%%%%%%%%%%%%%%%%%%%%%%%%%%%%%%%%%%%%%%%%%%%%%%%%%%%

\bibliography{Bib}

%%%%%%%%%%%%%%%%%%%%%%%%%%%%%%%%%%%%%%%%%%%%%%%%%%%%%%%%%%%%%%%%%%%%%%%%%%%%%%%%%%%%%%%%%%%%%%%%%%%%%%%%%%%%%%%%%%%%%%

\newpage

\onecolumngrid

\appendix

\section*{Appendices}

%%%%%%%%%%%%%%%%%%%%%%%%%%%%%%%%%%%%%%%%%%%%%%%%%%%%%%%%%%%%%%%%%%%%%%%%%%%%%%%%%%%%%%%%%%%%%%%%%%%%%%%%%%%%%%%%%%%%%%

\section{Preliminaries}\label{sec:Pre}

For completeness, this section briefly reviews the concepts and conventions underlying the analysis that follows. 
Subsection~\ref{subsec:Notations} introduces the notation used throughout the work. 
Subsection~\ref{subsec:Link_Product} recalls the composition rule for Choi operators of quantum channels, known as the link product. 
The broadcasting condition is presented in Subsec.~\ref{subsec:Broadcasting}, followed in Subsec.~\ref{subsec:Broadcasting_Fidelity} by the measures used to quantify broadcasting performance, namely the broadcasting fidelity and its averaged version. 
Subsection~\ref{subsec:Channel_Twirling} introduces channel twirling, which will play a key role in simplifying the analysis of approximate broadcasting in Subsec.~\ref{subsec:Marginals}. 
Readers already familiar with these notions may wish to proceed directly to the next section.

%%%%%%%%%%%%%%%%%%%%%%%%%%%%%%%%%%%%%%%%%%%%%%%%%%%%%%%%%%%%%%%%%%%%%%%%%%%%%%%%%%%%%%%%%%%%%%%%%%%%%%%%%%%%%%%%%%%%%%

\subsection{Notations}\label{subsec:Notations}

We begin by fixing the notation and conventions used throughout this work. 
The Hilbert space associated with system $A$ is denoted by $\mH_A$, and we will simply refer to it as system $A$.
A quantum state $\rho$ on $A$ is represented by a positive semidefinite operator, $\rho\geqslant0$, satisfying the normalization condition $\Tr[\rho]=1$.
The most general physical transformation of a quantum system is described by a quantum channel, formally defined as a completely positive and trace-preserving (CPTP) linear map. 
We will frequently employ the Choi–Jamio\l kowski isomorphism to represent such channels as operators on a bipartite Hilbert space. 
In particular, for a quantum channel $\mE:A\to B$, we define its associated Choi operator $J^{\mE}$ as
\begin{align}\label{eq:Choi}
    J^{\mE}_{AB}:=\id_{A\to A}\otimes\mE_{A'\to B}
    \left(\Gamma_{AA'}\right),
\end{align}
where 
$\Gamma$ denotes the unnormalized maximally entangled state (UMES) on systems $A$ and $A'$, with $A'$ a Hilbert space isomorphic to $A$. 
Explicitly, $\Gamma$ is defined as
\begin{align}\label{eq:UMES}
    \Gamma:=\ketbra{\Gamma}{\Gamma},
\end{align}
with
\begin{align}
    \ket{\Gamma}:=\sum_i\ket{ii}.
\end{align}
Assuming $\dim A=\dim A'=d$, the corresponding maximally entangled state (MES) $\ket{\phi^{+}}$ is defined as 
\begin{align}
    \ket{\phi^{+}}
    :=\frac{1}{\sqrt{d}}\ket{\Gamma}
    =\frac{1}{\sqrt{d}}\sum_i\ket{ii},
\end{align}
with
\begin{align}\label{eq:MES}
    \phi^{+}:=\ketbra{\phi^{+}}{\phi^{+}}.
\end{align}
Whenever the underlying system is clear from the context, we omit the corresponding subscript to streamline the notation.
In terms of the Choi operator representation defined in Eq.~\eqref{eq:Choi}, complete positivity (CP) of the quantum channel $\mE:A\to B$ is equivalent to
\begin{align}\label{eq:CP}
    J^{\mE}_{AB}\geqslant0,
\end{align}
while the trace-preserving (TP) condition is expressed as
\begin{align}\label{eq:TP}
    \Tr_{B}[J^{\mE}_{AB}]=\1_{A},
\end{align}
that is, tracing out the output system of the Choi operator $J^{\mE}$ yields the identity on the input system. 
All these properties of quantum channels admit a natural representation within the diagrammatic language of tensor networks; see Refs.~\cite{wood2015tensornetworksgraphicalcalculus,Coecke_Kissinger_2017,Bridgeman_2017,biamonte2020lecturesquantumtensornetworks,Collura_2024,xiao2025superchanneltearsgeneralizedoccams}.
Beyond quantum channels, this work also considers probabilistic protocols for quantum broadcasting. To describe such processes, it is convenient to introduce the notion of a subchannel. 
A linear map $\mE$ from system $A$ to system $B$ is called a subchannel if it is completely positive (CP) and trace non-increasing (TNI). 
In the Choi representation, the TNI condition takes the form
\begin{align}\label{eq:TNI}
    \Tr_{B}[J^{\mE}_{AB}]\leqslant\1_{A}.
\end{align}
Comprehensive accounts of quantum channels are provided in Refs.~\cite{Nielsen_Chuang_2010,RevModPhys.86.1203,Wilde_2017,Watrous_2018,khatri2024principlesquantumcommunicationtheory}.

Quantum channels, including subchannels, describe the class of transformations that can be implemented within standard quantum mechanics. 
Recent developments in quantum information processing, however, suggest that additional capabilities may emerge when quantum operations are combined with classical post-processing. 
Such hybrid strategies can lead to operational advantages beyond those attainable by physical quantum channels alone, with examples including virtual cooling~\cite{PhysRevX.9.031013}, virtual quantum resource distillation~\cite{PhysRevLett.132.050203,PhysRevA.109.022403}, reversal of unknown unitary operations~\cite{PhysRevLett.133.030801}, and quantum error mitigation~\cite{RevModPhys.95.045005}.

Mathematically, these virtual operations are described by linear maps that are Hermitian-preserving and trace-preserving (HPTP). 
This class strictly contains the set of completely positive and trace-preserving (CPTP) maps, thereby extending the operational framework beyond physical quantum channels.
A virtual operation $\mE$ from system $A$ to system $B$ can likewise be formulated in terms of its Choi operator $J^{\mE}$, for which the Hermitian-preserving (HP) condition takes the form
\begin{align}\label{eq:HP}
    (J^{\mE}_{AB})^{\dagger}=J^{\mE}_{AB}.
\end{align}
In other words, a linear map $\mE$ is Hermitian-preserving (HP) if and only if its Choi operator $J^{\mE}$ is Hermitian.

%%%%%%%%%%%%%%%%%%%%%%%%%%%%%%%%%%%%%%%%%%%%%%%%%%%%%%%%%%%%%%%%%%%%%%%%%%%%%%%%%%%%%%%%%%%%%%%%%%%%%%%%%%%%%%%%%%%%%%

\subsection{Link Product}\label{subsec:Link_Product}

This subsection reviews the link product~\cite{959270,PhysRevLett.101.060401,PhysRevA.80.022339}, a compact algebraic representation of the composition of quantum processes at the level of Choi operators. 
The link product provides a powerful framework for analyzing quantum channels and their interconnections, significantly simplifying many calculations involving composite processes. 
In particular, even the Born rule can be expressed naturally within this formalism. 
Mathematically, the link product $\star$ is defined as follows.

\begin{mydef}
{Link Product~\cite{959270,PhysRevLett.101.060401,PhysRevA.80.022339}}
{Link_Product}
Given two operators $M$ and $N$ with a common subsystem $A$, their link product, denoted by $M\star N$, is given by
\begin{align}\label{eq:Link_Product}
    M\star N:=\Tr_{A}[M\cdot N^{\T_{A}}].
\end{align}
Here $^{\T_{A}}$ denotes the partial transpose with respect to subsystem $A$.
\end{mydef}

Equipped with the link product, the composition of quantum channels can be expressed directly at the level of their Choi operators. 
In particular, for two quantum channels $\mE:A\to B$ and $\mF: B\to C$, the Choi operator of their sequential composition $\mF\circ\mE:A\to C$ is given by
\begin{align}\label{eq:LP_E_F}
    J^{\mF\circ\mE}=J^{\mE}\star J^{\mF}.
\end{align}
We assume a fixed underlying ordering of subsystems, which uniquely specifies the order of tensor products. 
Throughout, we omit the $\otimes$ symbol when writing products of operators acting on disjoint subsystems, with the convention that each operator is implicitly tensored with the identity on all other subsystems. 
For instance, in Eq.~\eqref{eq:LP_E_F} above, we have
\begin{align}
    J^{\mE}\star J^{\mF}
    =
    \Tr_{B}[
    \left(J^{\mE}_{AB}\otimes\1_{C}\right)
    \cdot
    \left(\1_{A}\otimes J^{\mF}_{BC}\right)^{\T_B}
    ]
    =
    \Tr_{B}[
    \left(\1_{A}\otimes J^{\mF}_{BC}\right)
    \cdot
    \left(J^{\mE}_{AB}\otimes\1_{C}\right)^{\T_B}
    ]
    =
    J^{\mF}\star J^{\mE}.
\end{align}
The Born rule also admits a compact formulation in terms of the link product. 
For a quantum state $\rho$ followed by a positive operator-valued measure (POVM) $\{M_x\}_x$, the probability of obtaining outcome $x$ can be expressed as
\begin{align}\label{eq:Born}
    p_x=\rho\star M_x^{\T}.
\end{align}
In this representation, the Choi operator of state preparation is identified with its density matrix $\rho$, while the Choi operator of each measurement effect $M_x$ is represented by its transposed operator $M_x^{\T}$.
Finally, consider the action of a quantum channel $\mE$ on a state $\rho$. 
The resulting output state can be expressed as the link product of their corresponding Choi operators, namely,
\begin{align}
    \mE(\rho)=J^{\mE}\star\rho.
\end{align}

Another useful property of the link product is that tracing out a subsystem corresponds, at the level of Choi operators, to linking with the identity operator $\1$ on that subsystem. For example, consider a broadcasting map $\mE:A\to BC$. 
Tracing out the output system $C$ yields a marginal channel whose Choi operator takes the form
\begin{align}
    J^{\Tr_C\circ\,\mE}_{AB}=\1_{C}\star J^{\mE}_{ABC}.
\end{align}
This property will be used frequently in our analysis of broadcasting throughout this work.

%%%%%%%%%%%%%%%%%%%%%%%%%%%%%%%%%%%%%%%%%%%%%%%%%%%%%%%%%%%%%%%%%%%%%%%%%%%%%%%%%%%%%%%%%%%%%%%%%%%%%%%%%%%%%%%%%%%%%%

\subsection{Broadcasting}\label{subsec:Broadcasting}

Quantum cloning and broadcasting provide the conceptual backdrop for the analysis developed here, centered on the broadcasting condition. We begin by recalling the notion of quantum cloning.

\begin{mydef}
{Quantum Cloning}
{Quantum_Cloning}
A 1 to 2 quantum channel $\mE:A\to BC$, mapping a state from a sender Alice to two receivers Bob and Claire, is said to be a cloning map if, for every input state $\rho$, one has
\begin{align}\label{eq:Cloning}
    \mE(\rho)=\rho\otimes\rho.
\end{align}
\end{mydef}

However, the linearity of quantum mechanics rules out the existence of such a channel, leading to the {\it no-cloning theorem}~\cite{Wootters1982clone,DIEKS1982clone}.
For comprehensive reviews of this topic, see Ref.~\cite{RevModPhys.77.1225}.
This motivates a relaxed requirement, known as the broadcasting condition, in which the output of the channel 
$\mE$ is not required to equal $\rho\otimes\rho$. 
Instead, the channel produces a bipartite state $\sigma:=\mE(\rho)$ whose reduced states both coincide with $\rho$, that is,

\begin{mydef}
{Quantum Broadcasting}
{Quantum_Broadcasting}
A 1 to 2 quantum channel $\mE:A\to BC$, which distributes a state from the sender Alice to two receivers, Bob and Claire, is called a broadcasting map if, for every input state $\rho$, one has
\begin{align}\label{eq:BC_1}
    \Tr_{B}[\mE(\rho)]=\rho_{C}, \quad\text{and}\quad
    \Tr_{C}[\mE(\rho)]=\rho_{B}.
\end{align}
\end{mydef}

The above broadcasting condition (BC) of Eq.~\eqref{eq:BC_1} admits an equivalent formulation purely at the level of the channel $\mE$, namely,
\begin{align}\label{eq:BC_2}
    \Tr_{B}\circ\,\mE_{A\to BC}=\id_{A\to C}, \quad\text{and}\quad
    \Tr_{C}\circ\,\mE_{A\to BC}=\id_{A\to B}.
\end{align}
Expressed in terms of the Choi operator (see Eq.~\eqref{eq:Choi}) and link product $\star$ (see Def.~\ref{def:Link_Product}), the broadcasting condition takes the form
\begin{align}\label{eq:BC_3}
    \1_{B}\star J^{\mE}_{ABC}=\Gamma_{AC}, \quad\text{and}\quad
    \1_{C}\star J^{\mE}_{ABC}=\Gamma_{AB}.
\end{align}
One may then ask whether a broadcasting channel can exist within quantum mechanics. 
Even with the relaxation from cloning to broadcasting, such a channel remains impossible, as established by the {\it no-broadcasting theorem}~\cite{PhysRevLett.76.2818}.

%%%%%%%%%%%%%%%%%%%%%%%%%%%%%%%%%%%%%%%%%%%%%%%%%%%%%%%%%%%%%%%%%%%%%%%%%%%%%%%%%%%%%%%%%%%%%%%%%%%%%%%%%%%%%%%%%%%%%%

\subsection{Broadcasting Fidelity}
\label{subsec:Broadcasting_Fidelity}

To quantify the performance of quantum broadcasting, we introduce the notion of broadcasting fidelity. 
In entanglement theory, the quality of entanglement is commonly characterized by the entanglement fidelity. 
This idea later inspired the concept of channel fidelity. 
Building on this notion, we employ channel fidelity to define the broadcasting fidelity and the corresponding average broadcasting fidelity.

Given two systems $A$ and $B$ of equal dimension $d$, the entanglement fidelity $F_{\text{Ent}}$ of a bipartite state $\rho_{AB}$ is defined as~\cite{10.5555/2011706.2011707,RevModPhys.81.865} 
\begin{align}\label{eq:Ent_Fidelity}
    F_{\text{Ent}}(\rho):=\Tr[\rho\cdot\phi^{+}].
\end{align}
Here $\phi^{+}$ represents the maximally entangled state (MES) introduced in Eq.~\eqref{eq:MES}.
Similarly, for a quantum channel $\mE:A\to B$, the corresponding channel fidelity $F_{\text{Chan}}$ is defined as~\cite{Kretschmann_2004}
\begin{align}\label{eq:Channel_Fidelity}
    F_{\text{Chan}}(\mE):=\Tr[
    \left(\frac{1}{d}J^{\mE}\right)\cdot\phi^{+}].
\end{align}
In the same spirit, we introduce the notion of broadcasting fidelity.

\begin{mydef}
{Broadcasting Fidelity on System $B$}
{BF_B}
In a broadcasting protocol $\mE$, the sender Alice attempts to distribute quantum information to two receivers, Bob and Claire. 
The performance on Bob's side is characterized by the marginal channel
\begin{align}
    \Tr_{C}\circ\,\mE_{A\to BC},
\end{align}
obtained from the broadcasting process $\mE:A\to BC$ by tracing out system $C$. 
The corresponding broadcasting fidelity on system $B$ is then defined as
\begin{align}\label{eq:Broadcasting_Fidelity_B}
    F_{\text{Bcast}, B}(\mE):=
    &F_{\text{Chan}}(\Tr_{C}\circ\,\mE_{A\to BC})\\
    =
    &\Tr[
    \left(\frac{1}{d}J^{\Tr_{C}\circ\,\mE_{A\to BC}}\right)\cdot\phi^{+}]
    =
    \frac{1}{d}\Tr[\left(\1_{C}\star J^{\mE}_{ABC}\right)\cdot\phi^{+}]
    =
    \frac{1}{d}\Tr[J^{\mE}_{AB}\cdot\phi^{+}]
    .
\end{align}
Here $J^{\mE}_{AB}$ denotes the reduced Choi operator obtained from $J^{\mE}_{ABC}$ by tracing out system $C$.
\end{mydef}

The broadcasting fidelity associated with system $C$ is defined analogously

\begin{mydef}
{Broadcasting Fidelity on System $C$}
{BF_C}
The broadcasting performance on Claire's side is quantified by the marginal channel obtained from the broadcasting map $\mE:A\to BC$ by tracing out system $B$,
\begin{align}
    \Tr_{B}\circ\,\mE_{A\to BC}.
\end{align}
The associated broadcasting fidelity is defined as
\begin{align}\label{eq:Broadcasting_Fidelity_C}
    F_{\text{Bcast}, C}(\mE):=
    F_{\text{Chan}}(\Tr_{B}\circ\,\mE_{A\to BC})
    =
    \frac{1}{d}\Tr[J^{\mE}_{AC}\cdot\phi^{+}]
    ,
\end{align}
where $J^{\mE}_{AB}$ is the reduced Choi operator obtained from $J^{\mE}_{ABC}$ by tracing out system $B$.
\end{mydef}

Thus far, broadcasting fidelity has been defined only in terms of the marginal output states. 
However, marginal criteria alone are insufficient to characterize the overall performance of a broadcasting protocol. 
In particular, they fail to capture how information is distributed jointly among the receivers.

This limitation can be illustrated by a simple example. 
Consider a map in which a noiseless quantum channel connects the sender Alice to the receiver Bob, while a maximally mixed state is prepared independently on Claire's side; that is
\begin{align}
    \mE_{A\to BC}=\id_{A\to B}\otimes \frac{1}{d}\1_{C}.
\end{align}
In this case, the broadcasting fidelity associated with Bob is $1$, reflecting perfect performance on his subsystem. 
By contrast, the broadcasting fidelity on Claire's side is $1/d^2$, as no information about the input state is conveyed to her.

This example highlights that perfect performance on one marginal can coexist with a complete lack of information transfer on the other. 
To meaningfully assess the overall quality of a broadcasting process, a global figure of merit is therefore required. 
Motivated by this observation, the average broadcasting fidelity is introduced as

\begin{mydef}
{Average Broadcasting Fidelity}
{ABF}
For a broadcasting map, $\mE:A\to BC$, from the sender Alice to the receivers Bob and Claire, the average broadcasting fidelity $F_{\text{Bcast}}(\mE)$ is defined as the mean of the fidelities associated with its marginal channels, namely
\begin{align}
    F_{\text{Bcast}}(\mE):=\frac{1}{2}
    \left(
    F_{\text{Bcast}, B}(\mE)
    +
    F_{\text{Bcast}, C}(\mE)
    \right)
\end{align}
where $F_{\text{Bcast}, B}$ and $F_{\text{Bcast}, C}$ are defined in Defs.~\ref{def:BF_B} and~\ref{def:BF_C}, respectively.
\end{mydef}

%%%%%%%%%%%%%%%%%%%%%%%%%%%%%%%%%%%%%%%%%%%%%%%%%%%%%%%%%%%%%%%%%%%%%%%%%%%%%%%%%%%%%%%%%%%%%%%%%%%%%%%%%%%%%%%%%%%%%%

\subsection{Channel Twirling}
\label{subsec:Channel_Twirling}

The final subsection of the Preliminaries Sec.~\ref{sec:Pre} focuses on channel twirling, a symmetric averaging procedure in which a random unitary is applied before a quantum channel and undone afterward by its inverse. 
Averaging over the unitary group with respect to the Haar measure maps the channel to a depolarizing channel. 
Formally, channel twirling is defined as

\begin{mydef}
{Channel Twirling}
{Channel_Twirling}
Given a quantum channel $\mE:A\to B$ with $d=\dim A=\dim B$, the twirled channel is defined as
\begin{align}
    \mT_{\text{Chan}}(\mE)(\cdot):=
    \int_{\text{Haar}}dU \,
    U^{\dagger} \left(\mE(U(\cdot)U^{\dagger})\right) U,
\end{align}
where the average is taken over random unitaries drawn according to the Haar measure~\cite{Haar1933}.
\end{mydef}

One can then readily verify that the Choi operator of the twirled channel $\mT_{\text{Chan}}(\mE)$ takes the form
\begin{align}
    J^{\mT_{\text{Chan}}(\mE)}
    =
    \int_{\text{Haar}}dU \,
    \left(U^*\otimes U\right)\cdot J^{\mE} \cdot \left(U^{\T}\otimes U^{\dagger}\right).
\end{align}
By taking the partial transpose $^{\T_{A}}$ on the first system, $A$, we obtain

\begin{align}
    (J^{\mT_{\text{Chan}}(\mE)})^{\T_{A}}
    =
    \int_{\text{Haar}}dU \,
    \left(U\otimes U\right)\cdot (J^{\mE})^{\T_{A}} \cdot \left(U^{\dagger}\otimes U^{\dagger}\right)
    =
    p\,\1_{A}\otimes\frac{\1_{B}}{d}
    +
    (1-p) F_{AB},
\end{align}
where the second equality follows from {\it Schur–Weyl duality}~\cite{fulton2004}, and $F$ denotes the unitary operator associated with the SWAP gate acting on systems $A$ and $B$.
Performing the partial transpose $^{\T_{A}}$ on the first system a second time gives
\begin{align}
    J^{\mT_{\text{Chan}}(\mE)}
    =
    (p\,\1_{A}\otimes\frac{\1_{B}}{d}
    +
    (1-p) F_{AB})^{\T_{A}}
    =
    p\,\1_{A}\otimes\frac{\1_{B}}{d}
    +
    (1-p) \Gamma_{AB}.
\end{align}
Finally, applying $\mT_{\text{Chan}}(\mE)$ to an input state $\rho$ yields
\begin{align}
    \mT_{\text{Chan}}(\mE)(\rho)
    =
    J^{\mT_{\text{Chan}}(\mE)}\star\rho
    =
    (p\cdot\1_{A}\star\rho)
    \,\frac{\1_{B}}{d}
    +
    (1-p) \Gamma_{AB}\star\rho
    =
    p\,\frac{\1_{B}}{d}
    +
    (1-p)\rho_{B}.
\end{align}
In other words, the twirled channel $\mT_{\text{Chan}}(\mE)$ reduces to a quantum depolarizing channel $\mD_{p}$, i.e., 
\begin{align}\label{eq:Twirling_Depolarizing}
    \mT_{\text{Chan}}(\mE)= \mD_{p},
\end{align}
with parameter $p$, determined by the original channel $\mE$.

%%%%%%%%%%%%%%%%%%%%%%%%%%%%%%%%%%%%%%%%%%%%%%%%%%%%%%%%%%%%%%%%%%%%%%%%%%%%%%%%%%%%%%%%%%%%%%%%%%%%%%%%%%%%%%%%%%%%%%

\section{Operational Limits on Quantum Broadcasting}
\label{sec:Limits}

The basic constraints that any broadcasting map must satisfy are introduced here. 
Subsec.~\ref{subsec:Conditions} discusses the conditions of sample efficiency (SE), unitary covariance (UC), permutation invariance (PI), and classical consistency (CC). 
Subsec.~\ref{subsec:SAnalysis_SComplexity} then develops a statistical analysis of quantum observable estimation, clarifying the associated sample overhead arising in virtual quantum information processing.

%%%%%%%%%%%%%%%%%%%%%%%%%%%%%%%%%%%%%%%%%%%%%%%%%%%%%%%%%%%%%%%%%%%%%%%%%%%%%%%%%%%%%%%%%%%%%%%%%%%%%%%%%%%%%%%%%%%%%%

\subsection{Assumptions and Conditions}
\label{subsec:Conditions}

In this subsection, we specify the conditions that a broadcasting map must satisfy, beginning with the requirement of sample efficiency (SE)~\cite{z2pr-zbwl,8g6j-w7ld}. 
Suppose that a sender, Alice, wishes to distribute an unknown quantum state $\rho$ to two receivers, Bob and Claire, in such a way that local measurements performed by each reproduce the statistics of $\rho$. 
Let Bob and Claire be interested in different observables, $\mO_B$ and $\mO_C$, which require $n_1$ and $n_2$ independent samples of $\rho$, respectively, to achieve a prescribed estimation accuracy. 
A nontrivial broadcasting protocol should accomplish this task using strictly fewer than $n_1+n_2$ copies of $\rho$.
Otherwise, simply distributing $n_1$ copies to Bob and $n_2$ copies to Claire would reduce the problem to trivial state preparation and distribution.
The requirement of SE can therefore be formally stated as

\begin{mydef}
{Sample Efficiency (SE)}
{Sample_Efficiency}
Consider a broadcasting task in which Alice is provided with $n$ copies of an initial state $\rho$, while the receivers Bob and Claire require $n_1$ and $n_2$ copies, respectively, to attain a prescribed accuracy in estimating certain observables on their local systems. 
A broadcasting protocol implemented by Alice is said to be sample efficient (SE) if
\begin{align}\label{eq:SE}
    n<n_1+n_2,
\end{align}
holds for any initial state $\rho$.
\end{mydef}

The second condition imposed on a broadcasting map is unitary covariance (UC). 
This condition states that applying a unitary operation to Alice's system prior to broadcasting $\mE$ is operationally equivalent to first performing the broadcasting map and then applying the same unitary to both Bob's and Claire's systems. 
Formally, the UC condition reads

\begin{mydef}
{Unitary Covariance (UC)}
{Unitary_Covariance}
A 1 to 2 quantum channel $\mE:A\to BC$ is said to be unitarily covariant (UC) if the following relation holds for any unitary map $\mU$.
\begin{align}\label{eq:UC}
    \mE\circ\mU=\mU\otimes\mU\circ\mE.
\end{align}
Here the action of $\mU$ on an input state $\rho$ is defined as $\mU(\rho)=U\rho U^{\dagger}$.
\end{mydef}

The third requirement is permutation invariance (PI), which demands that the broadcasting map be symmetric under an exchange of the two receivers. 
Concretely, applying a SWAP operation between Bob's and Claire's systems after broadcasting must leave the overall output state invariant, thereby ensuring symmetry between the two recipients.
The PI condition can be formally written as

\begin{mydef}
{Permutation Invariance (PI)}
{Permutation_Invariance}
A 1 to 2 quantum channel $\mE:A\to BC$ is permutation invariant (PI) if the following relation holds.
\begin{align}\label{eq:PI}
    \text{SWAP}\circ\mE=\mE.
\end{align}
Here SWAP denotes the quantum channel that swaps the output systems.
\end{mydef}

The final requirement is classical consistency (CC), which stipulates that when all quantum systems are subjected to a completely dephasing channel $\Delta(\cdot):=\sum_{i}\bra{i}\cdot\ket{i}\ketbra{i}{i}$ (in some basis $\{\ket{i}\}$), the broadcasting protocol must reduce to a purely classical broadcasting process $\mB_{\text{cl}}(\ketbra{i}{j}):=\delta_{ij}\ketbra{i}{i}\otimes\ketbra{i}{i}$ between the recipients.

\begin{mydef}
{Classical Consistency (CC)}
{Classical_Consistency}
A 1 to 2 quantum channel $\mE:A\to BC$ is said to satisfy classical consistency (CC) if, when complete dephasing $\Delta$ is applied to all its systems, the resulting map reduces to a classical broadcasting channel $\mB_{\text{cl}}$, namely
\begin{align}\label{eq:CC}
    \Delta\otimes\Delta
    \circ\mE\circ\Delta
    =\mB_{\text{cl}}.
\end{align}
\end{mydef}

In recent work by some of the present authors, a no practical quantum broadcasting theorem was established~\cite{z2pr-zbwl,8g6j-w7ld}, demonstrating that no map, extending even to the most general linear transformations, can simultaneously satisfy all four conditions above, including those characterizing virtual operations~\cite{PhysRevLett.132.110203}, namely Hermitian-preserving trace-preserving (HPTP) maps~\cite{Regula2021operational,Jiang2021physical}.
Formally, this no-go result can be stated as the following theorem

\begin{mythm}
{No Practical Quantum Broadcasting~\cite{z2pr-zbwl,8g6j-w7ld}}
{no_Qbroadcasting}
There exists no 1-to-2 linear map that simultaneously fulfills the requirements of SE (see Eq.~\eqref{eq:SE}), UC (see Eq.~\eqref{eq:UC}), PI (see Eq.~\eqref{eq:PI}), and CC (see Eq.~\eqref{eq:CC}).
\end{mythm}

The above theorem shows that no practical quantum broadcasting protocol exists. 
However, it does not rule out the existence of HPTP maps that satisfy the conditions of UC, PI, and CC~\cite{PhysRevLett.132.110203}. 
The limitation instead lies in their sample complexity, which fails to meet the requirement of SE. 
In other words, simply preparing a sufficient number of copies and distributing them directly to the receivers provides a more efficient strategy for broadcasting quantum information.

%%%%%%%%%%%%%%%%%%%%%%%%%%%%%%%%%%%%%%%%%%%%%%%%%%%%%%%%%%%%%%%%%%%%%%%%%%%%%%%%%%%%%%%%%%%%%%%%%%%%%%%%%%%%%%%%%%%%%%

\subsection{Statistical Analysis and Sample Complexity}
\label{subsec:SAnalysis_SComplexity}

This subsection determines how many copies of a quantum state are required to certify a given confidence level in experimentally obtained data. 
The analysis is carried out within a finite-sample framework, using an $\epsilon-\delta$ test to quantify the sample complexity associated with the implementation of the (virtual) broadcasting map~\cite{z2pr-zbwl,8g6j-w7ld}.
To control statistical fluctuations arising from finite data, the argument is based on {\it Hoeffding's inequality}~\cite{Hoeffding01031963}. 
This concentration bound applies to independent, bounded random variables and yields an explicit estimate for the probability that an empirical average deviates from its expected value. 
In the present setting, it provides a direct and non-asymptotic relation between the number of state copies and the achievable confidence level. 
Formally, the inequality states that

\begin{mylem}
{Hoeffding's Inequality~\cite{Hoeffding01031963}}{Hoeffding}
Let $X_1, \ldots, X_n$ be independent random variables, each taking values in a bounded interval $X_i\in[a_i,b_i]$ with $-\infty<a_i\leqslant b_i<\infty$.
Denote the empirical mean by
\begin{align}
    \overline{X}:=\frac{1}{n}\sum_{i}X_i,
\end{align}
and let
\begin{align}
    e:=\mathbb{E}[X_i]
\end{align}
be the common expectation value.
Then, for any $\epsilon>0$,
\begin{align}
    \Pr(|\overline{X}-e|\geqslant\epsilon)\leqslant
    2\exp{-\frac{2n^2\epsilon^2}{\sum_{i}(b_i-a_i)^2}}.
\end{align}
\end{mylem}

The statistical bound above can now be connected to the operational setting of quantum measurements. 
Consider a quantum system prepared in a state $\rho$, and suppose the expectation value of an observable $\mO$ is to be estimated, where $\mO$ is a Hermitian operator. 
Repeated measurements of $\mO$ on independent copies of $\rho$ yield a sequence of outcomes whose empirical average serves as an estimator for $\Tr[\rho\mO]$.

To connect this estimation procedure with realistic experimental implementations, it is convenient to express the observable in its spectral form
\begin{align}
    \mO=\sum_{k}o_k\ketbra{u_k}{u_k},
\end{align}
where $o_k$ are the eigenvalues of $\mO$ and $u_k:=\ketbra{u_k}{u_k}$ are the associated orthogonal projectors. 
Measuring $\mO$ therefore corresponds to sampling from $M:=\{\ketbra{u_k}{u_k}\}_{k}$, 
\begin{align}
    \Tr[\rho\ketbra{u_k}{u_k}],
\end{align}
with outcome $o_k$.
Denote by $n_k$ the number of times the outcome $o_k$ is observed, and let $N$ be the total number of experimental trials. 
The empirical average obtained from the experiment can then be written as
\begin{align}
    \overline{X}=\frac{1}{N}\sum_k n_ko_k.
\end{align}
The quantity of interest, namely the true expectation value of the observable in the state $\rho$, is given by
\begin{align}
    e=\Tr[\rho\mO].
\end{align}
These formulations provide the bridge between the statistical bounds discussed above in Lem.~\ref{lem:Hoeffding} and the resource requirements of practical experiments: the number of copies of $\rho$ required to estimate $\Tr[\rho\mO]$ with a prescribed precision and confidence level is governed by concentration inequalities applied to the empirical average of the measurement outcomes.

The estimation is deemed successful if the deviation between this empirical value and the true expectation does not exceed a prescribed tolerance $\epsilon>0$; otherwise, it is considered unsuccessful. Imposing a success probability of at least $1-\delta$ leads to a quantitative bound on the number of measurement samples required.
This naturally motivates the introduction of an $\epsilon-\delta$ test, which formalizes statistical reliability in finite-sample quantum experiments.

\begin{mydef}
{$\epsilon-\delta$ Test~\cite{z2pr-zbwl,8g6j-w7ld}}
{epsilon-delta}
Consider a quantum experiment designed to estimate the expectation value of an observable $\mO$ on a quantum state $\rho$.
Repeated measurements produce outcomes $X_i$, whose empirical mean is denoted by $\overline{X}$. 
The experiment is said to pass the $\epsilon-\delta$ test if the probability that the estimation error, defined as $|\overline{X}-\Tr[\rho\mO]|$, exceeds a prescribed error tolerance $\epsilon$ is no greater than $\delta$.
\end{mydef}

To simplify the analysis, we introduce $c$ as an upper bound on the variation of the measurement outcomes, defined by
\begin{align}\label{eq:c}
    c:=\max_{i} (b_i-a_i).
\end{align}
Passing the $\epsilon-\delta$ test in a quantum experiment is therefore equivalent to the condition that
\begin{align}
    2\exp{-\frac{2n^2\epsilon^2}{n c^2}}
    =
    2\exp{-\frac{2n\epsilon^2}{c^2}}
    \leqslant\delta,
\end{align}
which in turn implies that
\begin{align}
    n\geqslant\frac{c^2}{2\epsilon^2}
    \ln{\frac{2}{\delta}}.
\end{align}
For notational convenience, the number of copies required to achieve accuracy $\epsilon$ with confidence level $1-\delta$ will be denoted simply by $\mO((c^2\ln{1/\delta})/\epsilon^2)$, i.e.,
\begin{align}
    n=\mO\left(\frac{c^2}{\epsilon^2}
    \ln{\frac{1}{\delta}}\right).
\end{align}
Equipped with this result, the sample complexity required to estimate the observable $\mO$ while passing the $\epsilon-\delta$ test can be summarized in the following corollary.

\begin{mycor}
{Sample Cost}{Sample Cost}
Given $n$ independent copies of the state $\rho$ on which the observable $\mO$ is measured, passing the $\epsilon-\delta$ test (See Def.~\ref{def:epsilon-delta}) requires that
\begin{align}
    n=\mO\left(\frac{c^2}{\epsilon^2}
    \ln{\frac{1}{\delta}}\right),
\end{align}
where $c$ denotes the maximal range of the measurement outcomes.
\end{mycor}

A minimal benchmark for quantum broadcasting can now be identified. 
Suppose that two receivers, Bob and Claire, aim to estimate observables $\mO_{B}$ and $\mO_{C}$, respectively, each subject to an $\epsilon-\delta$ test that requires $n_1$ and $n_2$ independent copies of the state $\rho$.
In this benchmark protocol, the sender Alice prepares $n_1+n_2$ copies of $\rho$ and distributes them independently to Bob and Claire. 
This procedure trivially satisfies both statistical requirements and therefore serves as a natural reference against which the sample efficiency of canonical virtual broadcasting can be evaluated.
It also clarifies the motivation behind the definition of the sample efficienc (SE) condition adopted in Eq.~\eqref{eq:SE}.

Thus far, the discussion has focused on sample complexity in settings involving only quantum states and quantum observables, without incorporating classical data processing. 
However, when classical post-processing is combined with quantum channels, one is naturally led to the notion of {\it virtual operations} (or {\it virtual channels}), which can exhibit enhanced power in quantum information processing.

The analysis now turns to such virtual operations and examines how the inclusion of classical post-processing modifies the associated sample complexity requirements.
In a conventional quantum experiment, given a state $\rho$, the quantity of interest is the expectation value $\Tr[\mE(\rho)\mO]$ with $\mE$ being a quantum channel. 
By contrast, in the virtual setting, one is concerned more generally with quantities of the form $\Tr[\mE(\rho)\mO]$, where $\mE$ denotes a virtual operation. 

A first step towards understanding the resulting changes in sample complexity is therefore to clarify the physical realization of these virtual operations.
In practice, a virtual operation $\mE$ can be represented as a linear combination of two quantum channels, $\mE_1$ and $\mE_2$, with real coefficients $a$ and $b$. 
\begin{align}\label{eq:Virtual_Decomposition}
    \mE=a\mE_1-b\mE_2.
\end{align}
Since both $\mE_1$ and $\mE_2$ are trace preserving (TP), and the virtual operation $\mE$ is likewise required to be TP, the coefficients satisfy the constraint 
\begin{align}
    a-b=1.
\end{align}
With this decomposition in place, the quantity of interest, i.e., $\Tr[\mE(\rho)\mO]$, can be rewritten as

\begin{align}
    \Tr[\mE(\rho)\mO]
    =
    &\Tr[(a\mE_1-b\mE_2)(\rho)\mO]\\
    =
    &a\Tr[\mE_1(\rho)\mO]-b\Tr[\mE_2(\rho)\mO]\\
    =
    &(a+b)\left(
    \frac{a}{a+b}\Tr[\mE_1(\rho)\mO]-\frac{b}{a+b}\Tr[\mE_2(\rho)\mO]
    \right).
\end{align}

Physically, the above decomposition implies that the target virtual operation $\mE$ can be simulated by randomly implementing one of the two quantum channels $\mE_1$ and $\mE_2$. 
Specifically, $\mE_1$ is applied with probability $a/(a+b)$ and $\mE_2$ with probability $b/(a+b)$. 
Regardless of which channel is realized in a given run, the corresponding measurement outcome is rescaled by a factor of $a+b$ so as to reproduce the expectation value $\Tr[\mE(\rho)\mO]$.

This rescaling enlarges the effective range of the measurement outcomes from $c$ to $c(a+b)$. 
Consequently, when Hoeffding's inequality is applied, the required number of samples increases accordingly in order to achieve the same $\epsilon-\delta$ performance. 
In particular, the sample complexity is amplified by a factor of $(a+b)^2$.

\begin{mycor}
{Virtual Sample Cost}{Virtual Sample Cost}
Consider $n$ independent copies of a quantum state $\rho$ used to estimate the quantity $\Tr[\mE(\rho)\mO]$ via measurements of the observable, with $\mE$ being a virtual operation. 
Satisfying the $\epsilon-\delta$ criterion (See Def.~\ref{def:epsilon-delta}) requires that
\begin{align}
    n\geqslant\frac{(a+b)^2c^2}{2\epsilon^2}\ln{\frac{2}{\delta}},
\end{align}
or, equivalently
\begin{align}
    n=\mO\left(\frac{(a+b)^2c^2}{\epsilon^2}\ln{\frac{1}{\delta}}\right),
\end{align}
where $c$ denotes the maximal range of the measurement outcomes. 
The coefficients $a$ and $b$ originate from the decomposition of the virtual operation in Eq.~\eqref{eq:Virtual_Decomposition}.
\end{mycor}

Therefore, among all admissible decompositions of a given virtual operation, identifying the one that minimizes $a+b$ in Eq.~\eqref{eq:Virtual_Decomposition} is of central importance. 
This minimum directly determines the lowest achievable sample complexity for estimating the observable $\mO$ with a prescribed confidence level under the $\epsilon-\delta$ test.

%Motivated by recent developments in virtual operations and their applications, the framework of virtual quantum broadcasting has emerged. In this setting, the objective of broadcasting is reformulated: instead of requiring the physical preparation of marginal states, it suffices that expectation values of observables on the marginals reproduce those of the original state. Consequently, virtual broadcasting is realized through a combination of quantum channels and classical post-processing, without the need to physically prepare states on Bob's and Claire's systems whose marginals coincide with the original input.

%%%%%%%%%%%%%%%%%%%%%%%%%%%%%%%%%%%%%%%%%%%%%%%%%%%%%%%%%%%%%%%%%%%%%%%%%%%%%%%%%%%%%%%%%%%%%%%%%%%%%%%%%%%%%%%%%%%%%%

\section{Approximate Quantum Broadcasting}
\label{sec:AQB}

Exact virtual broadcasting cannot be achieved. 
We therefore shift our focus to the non-exact regime and introduce a general framework for approximate virtual broadcasting, examining its practical viability in terms of sample efficiency and its potential advantages over quantum broadcasting implemented solely through quantum channels.
The section is organized as follows. 
Subsection.~\ref{subsec:No_Practical_VB} establishes the impossibility of practical virtual broadcasting in the exact setting. 
Subsection.~\ref{subsec:AVB} introduces the framework of approximate virtual broadcasting and formulates the associated sample complexity overhead as a semidefinite programming (SDP). 
In Subsec.~\ref{subsec:Symmetry_Simplifications}, we exploit Schur–Weyl duality to simplify the analysis of both the sample complexity and the structure of the virtual broadcasting maps. 
This symmetry reduction reveals that the optimal virtual broadcasting protocol necessarily induces depolarizing channels on the marginals in Subsec.~\ref{subsec:Marginals}.
The dual formulation of the SDP is derived in Subsec.~\ref{subsec:SDP_Dual}, where the central constraint reduces to an eigenvalue problem that is solved in Subsec.~\ref{subsec:M_xy}. 
In Subsec.~\ref{subsec:AVB_SOCP}, the optimization problem is further transformed from an SDP into a second-order cone programming (SOCP), enabling an analytic solution presented in Subsec.~\ref{subsec:Analytic_Solution}. 
The sample efficiency of approximate broadcasting are analyzed in Subsec.~\ref{subsec:SE_AB}. 
Finally, Subsec.~\ref{subsec:Learning_AQB} discusses the application of approximate virtual broadcasting to quantum state learning.

%%%%%%%%%%%%%%%%%%%%%%%%%%%%%%%%%%%%%%%%%%%%%%%%%%%%%%%%%%%%%%%%%%%%%%%%%%%%%%%%%%%%%%%%%%%%%%%%%%%%%%%%%%%%%%%%%%%%%%

\subsection{
No Practical Virtual Broadcasting}
\label{subsec:No_Practical_VB}

The violation of the SE condition (see Def.~\ref{def:Sample_Efficiency}) does not arise directly from the assumptions of UC (see Def.~\ref{def:Unitary_Covariance}), PI (see Def.~\ref{def:Permutation_Invariance}), and CC (see Def.~\ref{def:Classical_Consistency}). 
In fact, the UC and CC conditions already imply both the BC (see Def.~\ref{def:Quantum_Broadcasting}) and HP (see Eq.~\eqref{eq:HP}) properties, which together give rise to virtual broadcasting. 
Here, we further show that, for any HPTP map, imposing the BC condition inevitably leads to a violation of the SE requirement. 
Consequently, virtual broadcasting cannot be realized in a practically sample efficient manner.

As discussed in Thm.~\ref{thm:no_Qbroadcasting}, previous work has established a no practical quantum broadcasting theorem, proving that no linear map can satisfy all four conditions in Subsec.~\ref{subsec:Conditions} simultaneously~\cite{z2pr-zbwl,8g6j-w7ld}.
This naturally raises the question of which conditions must be retained, and which must be abandoned, to make quantum broadcasting operationally meaningful. 
However, recent work shows that even replacing unitary covariance (UC) with alternative symmetry constraints, such as phase or flip covariance~\cite{PhysRevA.74.042309}, does not restore practicality~\cite{okada2025virtualphasecovariantquantumbroadcasting}: 
the resulting protocols still violate sample efficiency (SE). 
Violating SE renders the broadcasting task operationally trivial, since one could instead simply prepare a sufficient number of copies of the initial state and distribute them directly to the receivers, thereby achieving the same goal with no greater, and often lower, sample cost.

In fact, making broadcasting practical, i.e., ensuring at least sample efficiency (SE), cannot be achieved by simply modifying or discarding the conditions UC, PI, or CC. 
Rather, the obstruction is fundamental, stemming from an inherent conflict between BC in Eq.~\eqref{eq:BC_1} and the SE requirement in Eq.~\eqref{eq:SE}. 
For any HPTP map, the corresponding sample complexity of virtual broadcasting admits a formulation as a semidefinite programming~\cite{doi:10.1137/1038003},
\begin{align}\label{eq:SDP_BC}
    u_2:=
    \min \quad 
    & a+b\\
    \text{s.t.} \quad 
    &\1_{B}\star (J_1-J_2)=\Gamma_{AC},\label{eq:VB_C}\\
    &\1_{C}\star (J_1-J_2)=\Gamma_{AB},\label{eq:VB_B}\\
    & J_1\geqslant0,\,\, 
    \Tr_{BC}[J_1]=a\,\1_A,\\
    &J_2\geqslant0,\,\,
    \Tr_{BC}[J_2]=b\,\1_A,
\end{align}
where $J_1$ and $J_2$ are operators acting on systems $A$, $B$ and $C$, and $\star$ denotes the link product, as defined in Eq.~\eqref{eq:Link_Product}.
Thanks to Ref.~\cite{PhysRevA.110.012458}, the above semidefinite programming is efficiently solvable and yields a closed-form expression. 

\begin{mylem}
{Sample Complexity of Exact Virtual Broadcasting~\cite{PhysRevA.110.012458}}
{SC_EVB}
For systems of dimension $d$, with $\dim A=\dim B=\dim C=d$, the sample complexity of exact virtual broadcasting defined in Eq.~\eqref{eq:SDP_BC} takes the form
\begin{align}
    u_2=\frac{3d-1}{d+1}.
\end{align}
\end{mylem}

To satisfy the requirement of SE, the quantity $u_2$ must be strictly smaller than $\sqrt{2}$, namely $u_2^2<2$, which is equivalent to imposing an upper bound on the system dimension d,
\begin{align}\label{eq:Upper_Bound_No_VBroadcasting}
    d<\frac{1+\sqrt{2}}{3-\sqrt{2}}
    \approx1.5224.
\end{align}
However, for any nontrivial quantum system the minimal dimension is $d=2$, corresponding to a qubit. 
Consequently, the bound of Eq.~\eqref{eq:Upper_Bound_No_VBroadcasting} cannot be met by any physical quantum system. 
This observation leads directly to the following theorem.

\begin{mythm}
{No Practical Virtual Broadcasting}
{no_Vbroadcasting}
There exists no HPTP map that simultaneously fulfills the requirements of SE (see Def.~\ref{def:Sample_Efficiency}) and BC (see Def.~\ref{def:Quantum_Broadcasting}).
\end{mythm}

Remark that Thm.~\ref{thm:no_Vbroadcasting} differs from Thm.~\ref{thm:no_Qbroadcasting} and serves as a complementary result. 
Nevertheless, it does not imply Thm.~\ref{thm:no_Qbroadcasting}, since the latter holds for arbitrary linear maps~\cite{z2pr-zbwl,8g6j-w7ld}, whereas Thm.~\ref{thm:no_Vbroadcasting} is confined to the class of HPTP maps.

%%%%%%%%%%%%%%%%%%%%%%%%%%%%%%%%%%%%%%%%%%%%%%%%%%%%%%%%%%%%%%%%%%%%%%%%%%%%%%%%%%%%%%%%%%%%%%%%%%%%%%%%%%%%%%%%%%%%%%

\subsection{Approximate Virtual Broadcasting}
\label{subsec:AVB}

Theorem~\ref{thm:no_Vbroadcasting}, established in Subsec.~\ref{subsec:No_Practical_VB}, rules out the existence of an exact virtual broadcasting map. 
This naturally motivates the question of whether the sample complexity required for virtual broadcasting can be substantially reduced by relaxing the requirement of exactness. 
To address this, we introduce {\it approximate virtual broadcasting}, in which the broadcasting performance at Bob's and Claire's outputs is quantified by parameters $\epsilon_1$ and $\epsilon_2$ that measure the deviation from exact virtual broadcasting. 
Within this framework, the sample complexity of implementing an approximate broadcasting map is characterized by the following SDP
\begin{align}\label{eq:SDP_Approximate}
    u_2(\epsilon_1, \epsilon_2):=
    \min \quad 
    & a+b\\
    \text{s.t.} \quad 
    &\frac{1}{d}\Tr[\left(\1_{B}\star (J_1-J_2)\right)\cdot\phi^{+}]=1-\epsilon_2,\label{eq:Approx_C}\\
    &\frac{1}{d}\Tr[\left(\1_{C}\star (J_1-J_2)\right)\cdot\phi^{+}]=1-\epsilon_1,\label{eq:Approx_B}\\
    & J_1\geqslant0,\,\, 
    \Tr_{BC}[J_1]=a\,\1_A,\label{eq:J_1_a}\\
    &J_2\geqslant0,\,\,
    \Tr_{BC}[J_2]=b\,\1_A,\label{eq:J_2_b}\\
    &a-b=1.\label{eq:VB_TP_Approximate}
\end{align}
Here, Eqs.~\eqref{eq:Approx_C} and~\eqref{eq:Approx_B} specify the approximation accuracy at Claire's and Bob's outputs, respectively, as measured by the broadcasting fidelity in Eqs.~\eqref{eq:Broadcasting_Fidelity_C} and~\eqref{eq:Broadcasting_Fidelity_B}.

In the SDP formulation of Eq.~\eqref{eq:SDP_BC} for estimating the sample cost of deterministic and exact virtual broadcasting, the trace-preserving (TP) condition for the virtual broadcasting map, namely $a-b=1$, is not written explicitly. 
This omission does not indicate that the condition is unnecessary. 
On the contrary, it is automatically enforced and therefore implicitly assumed.
Specifically, whenever either the broadcasting condition (BC) Eq.~\eqref{eq:VB_C} or Eq.~\eqref{eq:VB_B} is satisfied, the relation $a-b=1$ follows.
To see this, consider Eq.~\eqref{eq:VB_C}:
by applying a partial trace over system $B$, one directly obtains 
\begin{align}
    \1_{BC}\star (J_1-J_2)
    =
    (a-b)\1_{A}
    =
    \1_{C}\star\Gamma_{AC}
    =
    \1_{A},
\end{align}
which immediately implies $a-b=1$. 
An analogous argument applies when starting from Eq.~\eqref{eq:VB_B}.
By contrast, in the approximate setting, the conditions in Eqs.~\eqref{eq:Approx_C} and~\eqref{eq:Approx_B} no longer enforce the TP constraint. 
As a result, the condition $a-b=1$ cannot be derived from the approximation conditions alone. 
This is why the TP condition for virtual broadcasting map must be imposed explicitly, as stated in Eq.~\eqref{eq:VB_TP_Approximate}.

Before solving the above SDP, we first exploit the symmetries encoded in Eqs.~\eqref{eq:Approx_C} and~\eqref{eq:Approx_B} to simplify the analysis. 
A first useful consequence of these symmetries is summarized in the following lemma.

\begin{mylem}
{Triple-Twirling Invariance}
{Twirling_3_1}
Suppose that an operator $J$, acting on systems $A$, $B$, and $C$, satisfies the conditions specified below, 
\begin{align}
    \Tr_{AC}\left[\Tr_{B}[J]\cdot\Gamma_{AC}\right]&=x,\label{eq:x}\\
    \Tr_{AB}[\Tr_{C}[J]\cdot\Gamma_{AB}]&=y,\label{eq:y}
\end{align}
with $\Gamma$ denoting the UMES (see Eq.~\eqref{eq:UMES}). 
Then the operator $\mT(J)$ obtained by applying three-fold twirling, or triple twirling, defined as 
\begin{align}\label{eq:Twirling}
    \mT(J):=\int_{\text{Haar}}dU \,
    \left(U^*\otimes U\otimes U\right)\cdot J \cdot \left(U^{\T}\otimes U^{\dagger}\otimes U^{\dagger}\right),
\end{align}
satisfies the same conditions.
Here, $^*$ denotes the complex conjugate of a matrix.
\end{mylem}

We emphasize that $\mT$ generalizes isotropic twirling~\cite{PhysRevA.59.4206}, rather than standard twirling~\cite{PhysRevA.40.4277}, to the tripartite setting, a construction that has been extensively used in the study of distillable entanglement and entanglement cost~\cite{RevModPhys.81.865}. 
For simplicity and clarity, we will nevertheless refer to $\mT$ as twirling throughout.

\begin{proof}
Suppose that the operator $J$ satisfies Eq.~\eqref{eq:x}. Then the conjugated operator 
$\left(U^*\otimes U\otimes U\right)\cdot J \cdot \left(U^{\T}\otimes U^{\dagger}\otimes U^{\dagger}\right)$ also satisfies the same condition for any unitary matrix $U$, as
\begin{align}
    \Tr_{AC}\left[\Tr_{B}\left[\left(U^*\otimes U\otimes U\right)\cdot J \cdot \left(U^{\T}\otimes U^{\dagger}\otimes U^{\dagger}\right)\right]\cdot\Gamma\right]
    &=
    \Tr_{AC}\left[\left(U^*\otimes U\right)\cdot J_{AC} \cdot \left(U^{\T}\otimes U^{\dagger}\right)\cdot\Gamma\right]\\
    &=
    \Tr_{AC}\left[J_{AC} \cdot \left(U^{\T}\otimes U^{\dagger}\right)\cdot\Gamma\cdot\left(U^*\otimes U\right)\right]\\
    &=
    \Tr_{AC}\left[J_{AC} \cdot \Gamma\right]\\
    &=x.
\end{align}
Here $J_{AC}$ denotes the reduced operator of $J$ on systems $A$ and $C$. 
An entirely analogous argument shows that the same statement holds for Eq.~\eqref{eq:y}.
Combining these observations, we conclude that if $J$ satisfies Eqs.~\eqref{eq:x} and~\eqref{eq:y}, then so does $\left(U^*\otimes U\otimes U\right)\cdot J \cdot \left(U^{\T}\otimes U^{\dagger}\otimes U^{\dagger}\right)$ for any unitary $U$.
Averaging over all unitaries with respect to the Haar measure then completes the proof.
\end{proof}

We now analyze the symmetries embodied in Eqs.~\eqref{eq:J_1_a} and~\eqref{eq:J_2_b}, leading to

\begin{mylem}
{Triple-Twirling Invariance}
{Twirling_3_2}
Suppose that an operator $J$, acting on systems $A$, $B$, and $C$, satisfies
\begin{align}
    \Tr_{BC}\left[J\right]&=\1_{A},
\end{align}
where $\1_{A}$ stands for the identity operator on system $A$. 
Then its triple twirling $\mT(J)$, defined above in Eq.~\eqref{eq:Twirling}, obeys the same condition.
\end{mylem}

\begin{proof}
For any unitary matrix $U$, it holds that
\begin{align}
    \Tr_{BC}[\left(U^*\otimes U\otimes U\right)\cdot J \cdot \left(U^{\T}\otimes U^{\dagger}\otimes U^{\dagger}\right)]
    =
    U^*\cdot J_A\cdot U^{\T}
    =
    U^*\cdot \1_A\cdot U^{\T}
    =
    \1_{A}.
\end{align}
The statement remains valid under Haar averaging, completing the proof.
\end{proof}

Taken together, Lems.~\ref{lem:Twirling_3_1} and~\ref{lem:Twirling_3_2} imply the following theorem, which underpins the simplification of the analysis of $u_2(\epsilon_1, \epsilon_2)$ defined in Eq.~\eqref{eq:SDP_Approximate}.

\begin{mythm}
{Triple-Twirling Optimality}
{Twirling_3}
Without loss of generality, we may assume that an optimal value of $u_2(\epsilon_1, \epsilon_2)$ (see Eq.~\eqref{eq:SDP_Approximate}) is attained at $\{a', b', J'_1, J'_2\}$:
\begin{align}
    \argmin u_2(\epsilon_1, \epsilon_2)
    =
    \{a', b', J'_1, J'_2\}.
\end{align}
Consequently, $\{a', b', \mT(J'_1), \mT(J'_2)\}$ also constitutes an optimal solution.
\end{mythm}

\begin{proof}
By the proofs of Lems.~\ref{lem:Twirling_3_1} and~\ref{lem:Twirling_3_2}, any feasible solution of the SDP in Eq.~\eqref{eq:SDP_Approximate} remains feasible after applying $U^*\otimes U\otimes U\left(\cdot\right)U^{\T}\otimes U^{\dagger}\otimes U^{\dagger}$ to $J_1$ and $J_2$, including the optimal solution.
The same conclusion holds upon taking the Haar average over all unitaries, thereby completing the proof.
\end{proof}

Theorem~\ref{thm:Twirling_3} allows us, without loss of generality, to restrict attention to optimal solutions of Eq.~\eqref{eq:SDP_Approximate} obtained via triple twirling. 
We therefore define
\begin{align}\label{eq:Opt_Appro_BC}
    J^{\mE}:=\mT(J'_1)-\mT(J'_2),
\end{align}
as the Choi operator of optimal approximate broadcasting map $\mE$ associated with error parameters $\epsilon_1$ and $\epsilon_2$, where $\{a', b', \mT(J'_1), \mT(J'_2)\}$ constitutes an optimal solution of the SDP in Eq.~\eqref{eq:SDP_Approximate}.
By construction, $J^{\mE}$ is invariant under the action of arbitrary unitaries $U$, namely 
\begin{align}\label{eq:UUU}
    \left(U^*\otimes U\otimes U\right)\cdot J^{\mE} \cdot \left(U^{\T}\otimes U^{\dagger}\otimes U^{\dagger}\right)= J^{\mE}.
\end{align}
This invariance implies that the virtual 1-to-2 approximate broadcasting map $\mE$ associated with $J^{\mE}$ is unitarily covariant and thus satisfies the UC conditions in Eq.~\eqref{eq:UC}. 

\begin{mycor}
{UC Optimality}
{UC_Optimality}
    There exists an optimal solution of the SDP in Eq.~\eqref{eq:SDP_Approximate} that is unitarily covariant (UC) (see Def.~\ref{def:Unitary_Covariance}).
\end{mycor}

The above Cor.~\ref{cor:UC_Optimality} reveals an interesting structural feature. 
In previous formulations of no practical quantum broadcasting~\cite{z2pr-zbwl,8g6j-w7ld} or virtual quantum broadcasting~\cite{PhysRevLett.132.110203}, the UC condition is typically introduced as an explicit assumption. 
In contrast, within the framework of approximate virtual broadcasting considered here, UC arises naturally as a consequence: 
whenever a virtual broadcasting protocol attains the optimal performance in the approximate setting, it can be chosen to satisfy the UC condition without loss of generality.

This symmetry greatly simplifies the characterization of the optimal solution of the SDP in Eq.~\eqref{eq:SDP_Approximate} and the structure of the corresponding optimal virtual broadcasting map. 
In particular, the optimal map can be constructed from permutation matrices of the symmetric group $\mathfrak{S}_3$, revealing the underlying permutation symmetry of the problem.

%%%%%%%%%%%%%%%%%%%%%%%%%%%%%%%%%%%%%%%%%%%%%%%%%%%%%%%%%%%%%%%%%%%%%%%%%%%%%%%%%%%%%%%%%%%%%%%%%%%%%%%%%%%%%%%%%%%%%%

\subsection{Simplifications from Symmetry}
\label{subsec:Symmetry_Simplifications}

The UC condition introduces a powerful symmetry that can be systematically exploited through Schur–Weyl duality. 
This symmetry substantially simplifies the analysis of the SDP in Eq.~\eqref{eq:SDP_Approximate} and enables a transparent characterization of the optimal approximate virtual broadcasting map that attains the minimal value of the optimization problem. 
In this subsection, we develop this analysis using elementary tools from representation theory, in particular Schur–Weyl duality, together with diagrammatic tensor network techniques that make the underlying structure explicit.

To begin the analysis, we exploit Schur–Weyl duality. 
More precisely, Schur–Weyl duality implies that the operator $(J^{\mE})^{\T_{A}}$ admits a decomposition in the basis formed by the six permutation operators generating the symmetric group $\mathfrak{S}_3$; namely,
\begin{align}
    (J^{\mE})^{\T_{A}}
    =
    c_1 
    \raisebox{-2.8ex}{\includegraphics[height=3em]{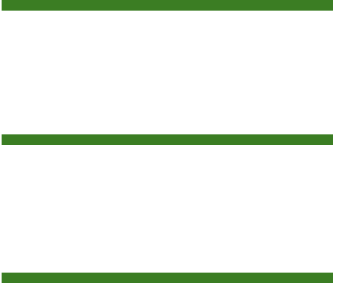}}
    +
    c_2
    \raisebox{-2.8ex}{\includegraphics[height=3em]{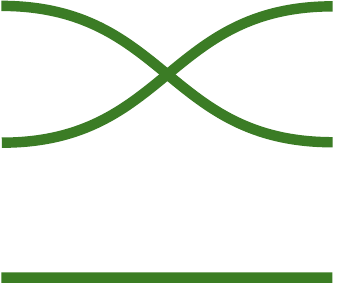}}
    +
    c_3
    \raisebox{-2.8ex}{\includegraphics[height=3em]{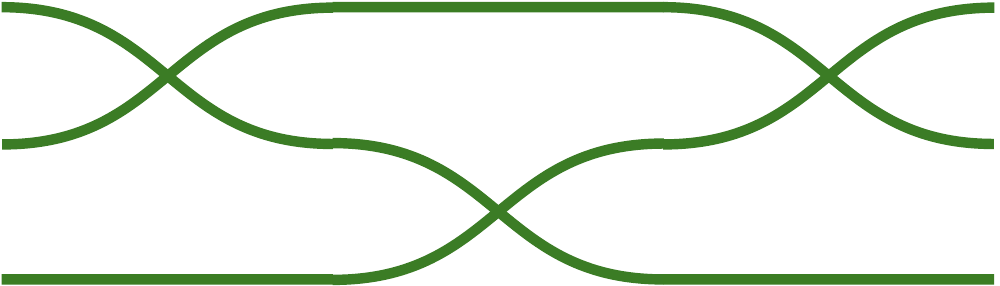}}
    +
    c_4
    \raisebox{-2.8ex}{\includegraphics[height=3em]{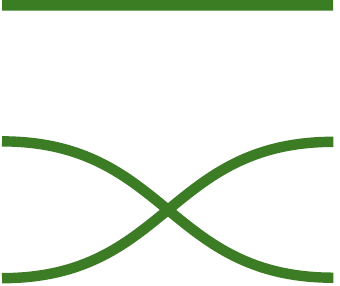}}
    +
    c_5
    \raisebox{-2.8ex}{\includegraphics[height=3em]{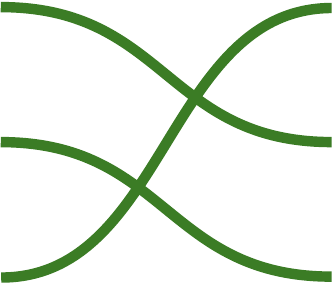}}
    +
    c_6
    \raisebox{-2.8ex}{\includegraphics[height=3em]{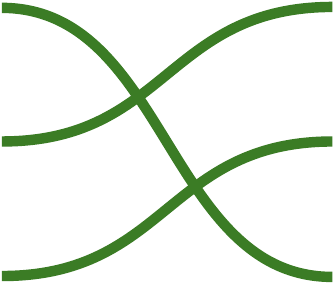}}\, .\label{eq:Haar-2-3}
\end{align}
Given an arbitrary input state $\rho$, the action of $\mE$ is given by
\begin{align}
    \mE(\rho)
    =
    &J^{\mE}\star\rho
    =
    \Tr_{A}[(J^{\mE})^{\T_{A}}\cdot(\rho\otimes\1_B\otimes\1_C)]\label{eq:mERho}\\
    =&
    c_1
    \raisebox{-1.8ex}{\includegraphics[height=1.8em]{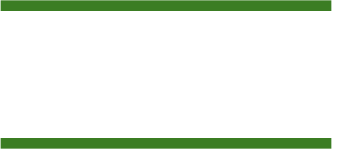}}
    +
    c_2
    \raisebox{-1.8ex}{\includegraphics[height=2.5em]{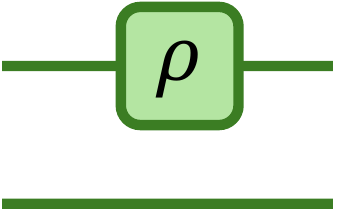}}
    +
    c_3
    \raisebox{-3.3ex}{\includegraphics[height=2.5em]{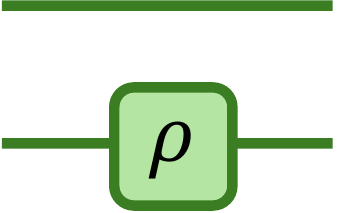}}
    +
    c_4
    \raisebox{-1.8ex}{\includegraphics[height=1.8em]{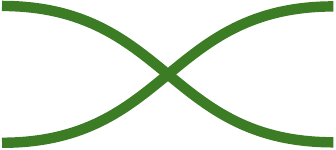}}
    +
    c_5
    \raisebox{-1.8ex}{\includegraphics[height=2.5em]{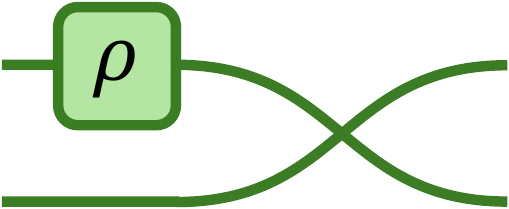}}
    +
    c_6
    \raisebox{-3.3ex}{\includegraphics[height=2.5em]{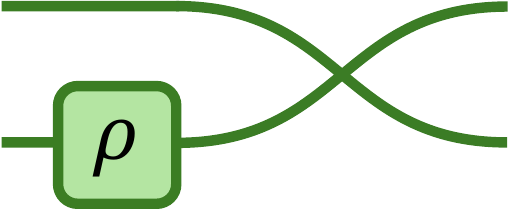}}\, .
\end{align}
Here the coefficients $c_i$ ($i\in\{1,\ldots,6\}$) are constants determined by the broadcasting map $\mE$. 
An analogous structure applies to the optimal solutions $J_1$ and $J_2$ in Eq.~\eqref{eq:SDP_Approximate}. 
Writing $\mE_1$ and $\mE_2$ for the associated channels, we assume that their expansions can be expressed as
\begin{align}
    (J_1)^{\T_{A}}
    &=
    a_1 
    \raisebox{-2.8ex}{\includegraphics[height=3em]{Id_3.pdf}}
    +
    a_2
    \raisebox{-2.8ex}{\includegraphics[height=3em]{P12.pdf}}
    +
    a_3
    \raisebox{-2.8ex}{\includegraphics[height=3em]{P13.pdf}}
    +
    a_4
    \raisebox{-2.8ex}{\includegraphics[height=3em]{N.pdf}}
    +
    a_5
    \raisebox{-2.8ex}{\includegraphics[height=3em]{P123.pdf}}
    +
    a_6
    \raisebox{-2.8ex}{\includegraphics[height=3em]{P132.pdf}}\, .\label{eq:Haar-2-3_J_1}
\end{align}
and
\begin{align}
    (J_2)^{\T_{A}}
    =
    b_1 
    \raisebox{-2.8ex}{\includegraphics[height=3em]{Id_3.pdf}}
    +
    b_2
    \raisebox{-2.8ex}{\includegraphics[height=3em]{P12.pdf}}
    +
    b_3
    \raisebox{-2.8ex}{\includegraphics[height=3em]{P13.pdf}}
    +
    b_4
    \raisebox{-2.8ex}{\includegraphics[height=3em]{N.pdf}}
    +
    b_5
    \raisebox{-2.8ex}{\includegraphics[height=3em]{P123.pdf}}
    +
    b_6
    \raisebox{-2.8ex}{\includegraphics[height=3em]{P132.pdf}}\, .\label{eq:Haar-2-3_J_2}
\end{align}

The analysis of the SDP in Eq.~\eqref{eq:SDP_Approximate} proceeds by reorganizing its constraints and translating the associated matrix inequalities into polynomial conditions.
We first recall the complete positivity (CP) of $J_1$,
\begin{align}
    J_1\geqslant0.
\end{align}
Within the tensor-network framework, this requirement is captured by exploring the action of partial transposition on the permutation basis of $\mathfrak{S}_3$:
\begin{align}
    \raisebox{-2.8ex}{\includegraphics[height=3em]{Id_3.pdf}}
    \quad&\xrightarrow{\T_{A}}\quad 
    \raisebox{-2.8ex}{\includegraphics[height=3em]{Id_3.pdf}},\\ 
    \vspace{6pt}\notag\\
    \raisebox{-2.8ex}{\includegraphics[height=3em]{P12.pdf}}
    \quad&\xrightarrow{\T_{A}}\quad 
    \raisebox{-2.8ex}{\includegraphics[height=3em]{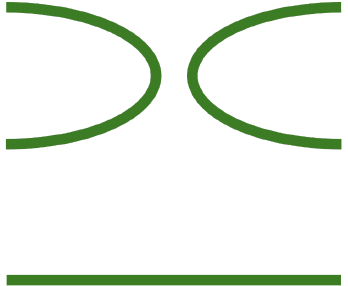}},\\ 
    \vspace{6pt}\notag\\
    \raisebox{-2.8ex}{\includegraphics[height=3em]{P13.pdf}}
    \quad&\xrightarrow{\T_{A}}\quad 
    \raisebox{-2.8ex}{\includegraphics[height=3em]{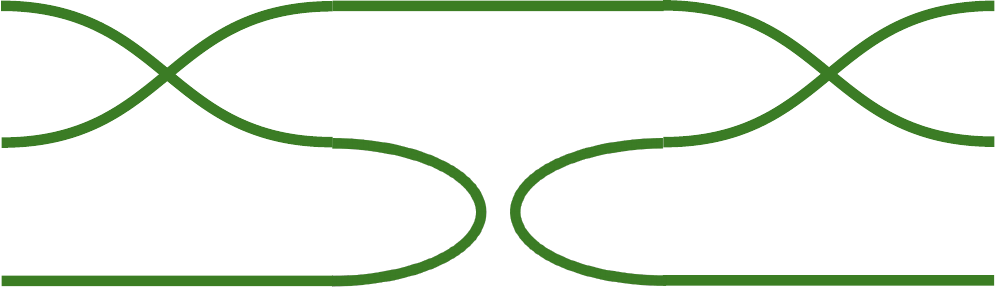}},\\ 
    \vspace{6pt}\notag\\
    \raisebox{-2.8ex}{\includegraphics[height=3em]{N.pdf}}
    \quad&\xrightarrow{\T_{A}}\quad 
    \raisebox{-2.8ex}{\includegraphics[height=3em]{N.pdf}},\\ 
    \vspace{6pt}\notag\\
    \raisebox{-2.8ex}{\includegraphics[height=3em]{P123.pdf}}
    \quad&\xrightarrow{\T_{A}}\quad 
    \raisebox{-2.8ex}{\includegraphics[height=3em]{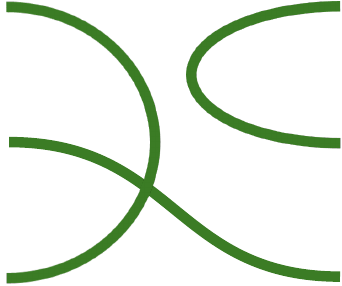}}
    ,\\ 
    \vspace{6pt}\notag\\
    \raisebox{-2.8ex}{\includegraphics[height=3em]{P132.pdf}}
    \quad&\xrightarrow{\T_{A}}\quad 
    \raisebox{-2.8ex}{\includegraphics[height=3em]{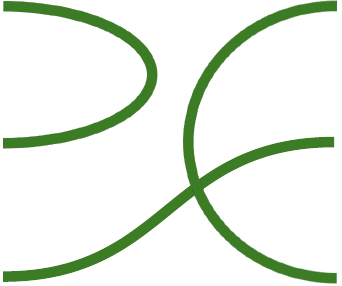}}.
\end{align}
The CP condition for $J_1$ is therefore equivalent to
\begin{align}
    J_1
    =
    a_1 
    \raisebox{-2.8ex}{\includegraphics[height=3em]{Id_3.pdf}}
    +
    a_2
    \raisebox{-2.8ex}{\includegraphics[height=3em]{M12.pdf}}
    +
    a_3
    \raisebox{-2.8ex}{\includegraphics[height=3em]{M13.pdf}}
    +
    a_4
    \raisebox{-2.8ex}{\includegraphics[height=3em]{N.pdf}}
    +
    a_5
    \raisebox{-2.8ex}{\includegraphics[height=3em]{M1.pdf}}
    +
    a_6
    \raisebox{-2.8ex}{\includegraphics[height=3em]{M2.pdf}}\,\geqslant0.\label{eq:J_1}
\end{align}
As the operator appearing on the left-hand side of Eq.~\eqref{eq:J_1}, equivalently Eq.~\eqref{eq:Haar-2-3_J_1}, is Hermitian, we conclude that
\begin{align}
    a_1, a_2, a_3, a_4&\in\mathbb{R},\\
    a_5&=a_6^*.
\end{align}
One may observe that the fifth and sixth terms on the left-hand side of Eq.~\eqref{eq:J_1} fail to be Hermitian.
This motivates the introduction of two Hermitian matrices adapted to the SDP formulation
\begin{align}
    M&:=\frac{1}{2}
    \left(\,\raisebox{-2.8ex}{\includegraphics[height=3em]{M1.pdf}}
    +
    \raisebox{-2.8ex}{\includegraphics[height=3em]{M2.pdf}}\,\right),\label{eq:M}\\
    N&:=\frac{1}{2i}
    \left(\,\raisebox{-2.8ex}{\includegraphics[height=3em]{M1.pdf}}
    -
    \raisebox{-2.8ex}{\includegraphics[height=3em]{M2.pdf}}\,\right).\label{eq:N}
\end{align}
Defining $\alpha_5$ and $\alpha_6$ as
\begin{align}
    \alpha_5&:=a_5+a_6,\label{eq:alpha_5}\\
    \alpha_6&:=i(a_5-a_6),\label{eq:alpha_6}
\end{align}
the CP condition can be reformulated as
\begin{align}
    a_1 
    \raisebox{-2.8ex}{\includegraphics[height=3em]{Id_3.pdf}}
    +
    a_2
    \raisebox{-2.8ex}{\includegraphics[height=3em]{M12.pdf}}
    +
    a_3
    \raisebox{-2.8ex}{\includegraphics[height=3em]{M13.pdf}}
    +
    a_4
    \raisebox{-2.8ex}{\includegraphics[height=3em]{N.pdf}}
    +
    \alpha_5 M
    +
    \alpha_6 N
    \geqslant0,
\end{align}
where all coefficients are real numbers and all matrices appearing in the above expansion are Hermitian.
Meanwhile, The CP condition of $J_2$ is equivalently expressed as
\begin{align}
    b_1 
    \raisebox{-2.8ex}{\includegraphics[height=3em]{Id_3.pdf}}
    +
    b_2
    \raisebox{-2.8ex}{\includegraphics[height=3em]{M12.pdf}}
    +
    b_3
    \raisebox{-2.8ex}{\includegraphics[height=3em]{M13.pdf}}
    +
    b_4
    \raisebox{-2.8ex}{\includegraphics[height=3em]{N.pdf}}
    +
    \beta_5 M
    +
    \beta_6 N
    \geqslant0,
\end{align}
with $\beta_5$ and $\beta_6$ defined as
\begin{align}
    \beta_5&:=b_5+b_6,\label{eq:beta_5}\\
    \beta_6&:=i(b_5-b_6).\label{eq:beta_6}
\end{align}

We next consider the trace-preserving (TP) condition for $J_1$, namely,
\begin{align}
    \Tr_{BC}[J_1]=a\,\1_A.
\end{align}
Using Eq.~\eqref{eq:Haar-2-3_J_1}, this condition is equivalent to
\begin{align}
    a_1d^2+(a_2+a_3+a_4)d+\alpha_5=a.
\end{align}
Here $d$ is the (common) dimension of systems $A$, $B$, and $C$, i.e., $d=\dim A=\dim B=\dim C$, and $\alpha_5$ is defined in Eq.~\eqref{eq:alpha_5}.
All arguments presented for $J_1$ apply analogously to $J_2$. 
For clarity, we therefore carry out the analysis explicitly only for $J_1$ and omit the corresponding derivations for $J_2$.
The TP condition for $J_1-J_2$ then reduces to
\begin{align}
    a-b
    =
    &(a_1d^2+(a_2+a_3+a_4)d+\alpha_5)
    -
    (b_1d^2+(b_2+b_3+b_4)d+\beta_5)\\
    =
    &(a_1-b_1)d^2+((a_2-b_2)+(a_3-b_3)+(a_4-b_4))d+(\alpha_5-\beta_5)\\
    =
    &1.
\end{align}

We finally consider the approximate broadcasting condition on Bob's side, namely,
\begin{align}\label{eq:Approximate_B}
    \frac{1}{d}\Tr[\left(\1_{C}\star (J_1-J_2)\right)\cdot\phi^{+}]=1-\epsilon_1.
\end{align}
Using Eq.~\eqref{eq:J_1}, $J_1-J_2$ reads
\begin{align}
    J_1-J_2
    =
    &(a_1-b_1) 
    \raisebox{-2.8ex}{\includegraphics[height=3em]{Id_3.pdf}}
    +
    (a_2-b_2)
    \raisebox{-2.8ex}{\includegraphics[height=3em]{M12.pdf}}
    +
    (a_3-b_3)
    \raisebox{-2.8ex}{\includegraphics[height=3em]{M13.pdf}}
    +
    (a_4-b_4)
    \raisebox{-2.8ex}{\includegraphics[height=3em]{N.pdf}}\\
    \vspace{6pt}\notag\\
    +
    &(a_5-b_5)
    \raisebox{-2.8ex}{\includegraphics[height=3em]{M1.pdf}}
    +
    (a_6-b_6)
    \raisebox{-2.8ex}{\includegraphics[height=3em]{M2.pdf}}\, .
\end{align}
Taking the link product with $\1_{C}$, which corresponds to tracing out system $C$, we obtain
\begin{align}
    \1_{C}\star (J_1-J_2)
    =
    &\left[(a_1-b_1)d+(a_3-b_3)+(a_4-b_4)\right]\,
    \raisebox{-1.8ex}{\includegraphics[height=1.8em]{I_2.pdf}}
    +
    \left[(a_2-b_2)d+(a_5-b_5)+(a_6-b_6)\right]\,
    \raisebox{-1.8ex}{\includegraphics[height=1.8em]{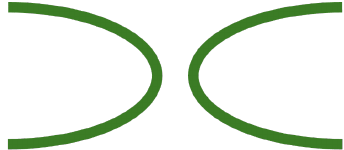}}.
\end{align}
Upon defining 
\begin{align}
    x:=(a_1-b_1)d+(a_3-b_3)+(a_4-b_4),
\end{align}
and
\begin{align}
    y:=(a_2-b_2)d+(a_5-b_5)+(a_6-b_6),
\end{align}
the left-hand side of Eq.~\eqref{eq:Approximate_B} becomes
\begin{align}
    \frac{1}{d}\Tr[\left(\1_{C}\star (J_1-J_2)\right)\cdot\phi^{+}]
    =
    &\frac{1}{d^2}
    \left(
    x\Tr[\raisebox{-1.8ex}{\includegraphics[height=1.8em]{I_2.pdf}}
    \cdot
    \raisebox{-1.8ex}{\includegraphics[height=1.8em]{UMES.pdf}}]
    +
    y\Tr[\raisebox{-1.8ex}{\includegraphics[height=1.8em]{UMES.pdf}}
    \cdot
    \raisebox{-1.8ex}{\includegraphics[height=1.8em]{UMES.pdf}}]
    \right)\\
    =
    &\frac{1}{d^2}
    \left(xd+yd^{2}\right)\\
    =
    &\frac{1}{d}
    (x+yd).
\end{align}
One can readily verify that Eq.~\eqref{eq:Approximate_B} is equivalent to
\begin{align}
    \left(
    (a_2-b_2)d+\alpha_5-\beta_5-1+\epsilon_1
    \right)d
    +
    \left(
    (a_1-b_1)d+(a_3-b_3)+(a_4-b_4)
    \right)
    =0.
\end{align}
The approximate broadcasting condition on Claire's side follows as
\begin{align}\label{eq:Approximate_C}
    \frac{1}{d}\Tr[\left(\1_{B}\star (J_1-J_2)\right)\cdot\phi^{+}]=1-\epsilon_2.
\end{align}
Here we have
\begin{align}\label{eq:approx_trace_B}
    \1_{B}\star (J_1-J_2)
    =
    &\left[(a_1-b_1)d+(a_2-b_2)+(a_4-b_4)\right]\,
    \raisebox{-1.8ex}{\includegraphics[height=1.8em]{I_2.pdf}}
    +
    \left[(a_3-b_3)d+(a_5-b_5)+(a_6-b_6)\right]\,
    \raisebox{-1.8ex}{\includegraphics[height=1.8em]{UMES.pdf}}.
\end{align}
Equation~\eqref{eq:Approximate_C} now takes the form
\begin{align}
    \left(
    (a_3-b_3)d+\alpha_5-\beta_5-1+\epsilon_2
    \right)d
    +
    \left(
    (a_1-b_1)d+(a_2-b_2)+(a_4-b_4)
    \right)
    =0.
\end{align}

The preceding analysis yields the following reformulation of the SDP in Eq.~\eqref{eq:SDP_Approximate}:

\begin{align}\label{eq:SDP_Approximate_2}
    u_2(\epsilon_1, \epsilon_2)=
    \min \quad 
    & 2\left(a_1d^2+(a_2+a_3+a_4)d+\alpha_5\right)-1\\
    \text{s.t.} \quad 
    &\left(
    (a_3-b_3)d+\alpha_5-\beta_5-1+\epsilon_2
    \right)d
    +
    \left(
    (a_1-b_1)d+(a_2-b_2)+(a_4-b_4)
    \right)
    =0,\\
    &\left(
    (a_2-b_2)d+\alpha_5-\beta_5-1+\epsilon_1
    \right)d
    +
    \left(
    (a_1-b_1)d+(a_3-b_3)+(a_4-b_4)
    \right)
    =0,\\
    & a_1 
    \raisebox{-2.8ex}{\includegraphics[height=3em]{Id_3.pdf}}
    +
    a_2
    \raisebox{-2.8ex}{\includegraphics[height=3em]{M12.pdf}}
    +
    a_3
    \raisebox{-2.8ex}{\includegraphics[height=3em]{M13.pdf}}
    +
    a_4
    \raisebox{-2.8ex}{\includegraphics[height=3em]{N.pdf}}
    +
    \alpha_5 M
    +
    \alpha_6 N
    \geqslant0,\\
    \vspace{6pt}\notag\\
    &b_1 
    \raisebox{-2.8ex}{\includegraphics[height=3em]{Id_3.pdf}}
    +
    b_2
    \raisebox{-2.8ex}{\includegraphics[height=3em]{M12.pdf}}
    +
    b_3
    \raisebox{-2.8ex}{\includegraphics[height=3em]{M13.pdf}}
    +
    b_4
    \raisebox{-2.8ex}{\includegraphics[height=3em]{N.pdf}}
    +
    \beta_5 M
    +
    \beta_6 N
    \geqslant0,\\
    &
    (a_1-b_1)d^2
    +
    ((a_2-b_2)+(a_3-b_3)+(a_4-b_4))d
    +
    (\alpha_5-\beta_5-1)=0,
\end{align}
where $M$ and $N$ denote the matrices defined in Eqs.~\eqref{eq:M} and~\eqref{eq:N}.

%%%%%%%%%%%%%%%%%%%%%%%%%%%%%%%%%%%%%%%%%%%%%%%%%%%%%%%%%%%%%%%%%%%%%%%%%%%%%%%%%%%%%%%%%%%%%%%%%%%%%%%%%%%%%%%%%%%%%%

\subsection{Physical Realization of Marginals}
\label{subsec:Marginals}

In previous subsections, we exploited symmetry, specifically, Schur–Weyl duality, to simplify the semidefinite programming by substantially reducing the number of independent variables. 
While this procedure is mathematically convenient, it naturally prompts a deeper question: what is the physical content of this symmetry reduction, and what insight does it provide beyond computational simplification?
In particular, does the imposed symmetry reveal intrinsic features of approximate virtual quantum broadcasting itself? 
We answer this question affirmatively. 
By focusing on the operator difference $J_1-J_2$, we demonstrate that the constraints in Eqs.~\eqref{eq:Approx_C} and~\eqref{eq:Approx_B} of the SDP in Eq.~\eqref{eq:SDP_Approximate} are equivalent to requiring that all marginal channels derived from $J_1-J_2$ take the form of quantum depolarizing channels $\mD_{p}$ discussed in Subsec.~\ref{subsec:Channel_Twirling}.

To make the underlying structure explicit, we consider an alternative SDP in which the marginal constraints are fixed to quantum depolarizing channels.

\begin{align}\label{eq:SDP_Approximate_v}
    v_2(p_1, p_2):=
    \min \quad 
    & a+b\\
    \text{s.t.} \quad 
    &\Tr_{B}[J_1-J_2]=
    \frac{p_2}{d}\,\1_{AC}
    +
    (1-p_2) \Gamma_{AC},\label{eq:Approx_C_v}\\
    &\Tr_{C}[J_1-J_2]=
    \frac{p_1}{d}\,\1_{AB}
    +
    (1-p_1) \Gamma_{AB},\label{eq:Approx_B_v}\\
    & J_1\geqslant0,\,\, 
    \Tr_{BC}[J_1]=a\,\1_A,\label{eq:J_1_v}\\
    &J_2\geqslant0,\,\,
    \Tr_{BC}[J_2]=b\,\1_A,\label{eq:J_2_v}\\
    &a-b=1.\label{eq:VB_TP_Approximate_v}
\end{align}

%%%%%%%%%%%% simplify %%%%%%%%%%
%Let
%\begin{align}
%    D_i =\frac{p_i}{d}\,\1+(1-p_i)\Gamma.
%\end{align}

%Let $J = J_1 - J_2$

%\begin{align}
%    v_2(p_1, p_2):=
%    \min \quad 
%    & 1+2b\\
%    \text{s.t.} \quad 
%    &\Tr_{B}[J]= D_2\\
%    &\Tr_{C}[J]= D_1\\
%    &\Tr_{BC}[J]=\1_A,\\
%    &\Tr_{BC}[J_2]=b\,\1_A,\\
%    & J_2\geqslant0,\,\, \\
%    &J + J_2\geqslant0,\,\,
%\end{align}
%%%%%%%%%%%% simplify %%%%%%%%%%

The equivalence between the corresponding optimal values $u_2(\epsilon_1, \epsilon_2)$ (see Eq.~\eqref{eq:SDP_Approximate}) and $v_2(p_1, p_2)$ (see Eq.~\eqref{eq:SDP_Approximate_v}) is formalized in the theorem below.

\begin{mythm}
{Depolarizing-Marginal SDP Formulation}
{Depolarizing_Marginal}
The quantities 
$u_2(\epsilon_1, \epsilon_2)$ (see Eq.~\eqref{eq:SDP_Approximate}) and $v_2(p_1, p_2)$ (see Eq.~\eqref{eq:SDP_Approximate_v}) coincide, namely 
\begin{align}
    u_2(\epsilon_1, \epsilon_2)=
    v_2(p_1, p_2),
\end{align}
whenever the following conditions are satisfied
\begin{align}\label{eq:p_epsilon_i}
    p_i=\frac{d^2 \epsilon_i}{d^2-1}, \quad
    \forall\, i\in\{1, 2\}.
\end{align}
\end{mythm}

\begin{proof}
We begin by showing that any feasible solution of $u_2(\epsilon_1, \epsilon_2)$ is also feasible for $v_2(p_1, p_2)$. 
As discussed in Subsec.~\ref{subsec:AVB}, one may, without loss of generality, restrict attention to solutions of $u_2(\epsilon_1, \epsilon_2)$ that are invariant under unitary conjugation of $U^*\otimes U\otimes U$, as shown in Eq.~\eqref{eq:UUU}.
Upon tracing out system $B$, this invariance implies that the resulting marginal channel is necessarily depolarizing (see also Eq.~\eqref{eq:approx_trace_B}).
\begin{align}
    \Tr_{B}[J_1-J_2]=
    \frac{p_2}{d}\,\1_{AC}
    +
    (1-p_2) \Gamma_{AC}.
\end{align}
Substituting this structure into Eq.~\eqref{eq:Approx_C}, we obtain
\begin{align}
    \frac{1}{d}\Tr[\left(\1_{B}\star (J_1-J_2)\right)\cdot\phi^{+}]
    =
    &\frac{1}{d^2}
    \left(
    \frac{p_2}{d}\Tr[\raisebox{-1.8ex}{\includegraphics[height=1.8em]{I_2.pdf}}
    \cdot
    \raisebox{-1.8ex}{\includegraphics[height=1.8em]{UMES.pdf}}]
    +
    (1-p_2)\Tr[\raisebox{-1.8ex}{\includegraphics[height=1.8em]{UMES.pdf}}
    \cdot
    \raisebox{-1.8ex}{\includegraphics[height=1.8em]{UMES.pdf}}]
    \right)\\
    =
    &\frac{1}{d^2}
    \left(
    \frac{p_2}{d}d+(1-p_2)d^{2}
    \right)\\
    =
    &1-(1-\frac{1}{d^2})p_2\\
    =
    &1-\epsilon_2,
\end{align}
which further implies that
\begin{align}
    p_2=\frac{d^2 \epsilon_2}{d^2-1}.
\end{align}
An analogous analysis applies to $p_1$ and $\epsilon_1$. 
Collecting these results, we conclude that whenever Eq.~\eqref{eq:p_epsilon_i} is satisfied, any feasible solution of $u_2(\epsilon_1, \epsilon_2)$ is also feasible for $v_2(p_1, p_2)$. 
The converse follows by a straightforward verification, which completes the proof.
\end{proof}

Theorem~\ref{thm:Depolarizing_Marginal} reveals that depolarizing marginals provide a complete characterization of feasible HPTP linear maps for approximate virtual broadcasting.

%%%%%%%%%%%%%%%%%%%%%%%%%%%%%%%%%%%%%%%%%%%%%%%%%%%%%%%%%%%%%%%%%%%%%%%%%%%%%%%%%%%%%%%%%%%%%%%%%%%%%%%%%%%%%%%%%%%%%%

\subsection{SDP Dual Formulation}
\label{subsec:SDP_Dual}

In this subsection, we derive the dual formulation of the sample complexity associated with approximate virtual broadcasting and show that it admits a substantial simplification under triple-twirling symmetry.
This symmetry reduction will ultimately enable an exact characterization of the optimal value in the subsequent subsections. 
To expose the structure of the problem in Eq.~\eqref{eq:SDP_Approximate}, we therefore consider the dual of the SDP in Eq.~\eqref{eq:SDP_Approximate_v}, corresponding to the quantity $v_2(p_1, p_2)$.
The associated Lagrangian is given by
\begin{align}
    \mL
    :=
    &
    (a+b)-(a-b-1)c\\
    &-
    \Tr_{AC}\left[
    \left(
    \Tr_{B}[J_1-J_2]-\frac{p_2}{d}\1_{AC}-(1-p_2)\Gamma_{AC}
    \right)\cdot
    X_{AC}
    \right]\\
    &
    -
    \Tr_{AB}\left[
    \left(
    \Tr_{C}[J_1-J_2]-\frac{p_1}{d}\1_{AB}-(1-p_1)\Gamma_{AB}
    \right)\cdot
    Y_{AB}
    \right]\\
    &
    -\Tr_{ABC}[J_1\cdot M_{ABC}]
    -\Tr_{ABC}[J_2\cdot N_{ABC}]
    -
    \Tr_{A}\left[
    \left(
    \Tr_{BC}[J_1]-a\1_{A}
    \right)
    \cdot P_A
    \right]
    -
    \Tr_{A}\left[
    \left(
    \Tr_{BC}[J_2]-b\1_{A}
    \right)
    \cdot Q_A
    \right],
\end{align}
where $X_{AC}$, $Y_{AB}$, $P_{A}$, $Q_{A}$ are Hermitian operators, $M_{ABC}$ and $N_{ABC}$ are positive semidefinite, and $c\in\mathbb{R}$.
Optimizing the Lagrangian over the primal variables then yields the dual SDP associated with $v_2(p_1, p_2)$; that is
\begin{align}\label{eq:Approximate_Dual_Form_1}
    v_2^{\text{Dual}}(p_1, p_2)
    :=
    \max \quad 
    & 
    1+\Tr[P]+
    \Tr[J^{\mD_{p_2}}\cdot X]+
    \Tr[J^{\mD_{p_1}}\cdot Y]
    \\
    \text{s.t.} \quad 
    &\Tr[P+Q]=-2,\\
    &P_{A}\otimes\1_{BC}+X_{AC}\otimes\1_{B}+Y_{AB}\otimes\1_{C}\leqslant0,\\
    & Q_{A}\otimes\1_{BC}-X_{AC}\otimes\1_{B}-Y_{AB}\otimes\1_{C}\leqslant0.
\end{align}
Here the Choi operator of the depolarizing channels $\mD_{p_i}$ is given by 
\begin{align}
    J^{\mD_{p_i}}=\frac{p_i}{d}\,\1+(1-p_i)\Gamma.
\end{align}

For notational convenience, we perform the substitution $P_{A}\mapsto-P_{A}$ in Eq.~\eqref{eq:Approximate_Dual_Form_1}, which yields a more compact formulation of the dual problem, stated below.

\begin{mylem}
{Dual Formulation of Approximate Broadcasting}
{Dual_Formulation_AVB} 
The sample complexity $v_2(p_1, p_2)$ of approximate virtual broadcasting in Eq.~\eqref{eq:SDP_Approximate_v} can be estimated via the following dual formulation
\begin{align}\label{eq:Approximate_Dual_Form_2}
    v_2^{\text{Dual}}(p_1, p_2)
    :=
    \max \quad 
    & 
    1-\Tr[P]+
    \Tr[J^{\mD_{p_2}}\cdot X]+
    \Tr[J^{\mD_{p_1}}\cdot Y]
    \\
    \text{s.t.} \quad 
    &\Tr[P-Q]=2,\label{eq:P_Q}\\
    &P_{A}\otimes\1_{BC}\geqslant X_{AC}\otimes\1_{B}+Y_{AB}\otimes\1_{C}\geqslant Q_{A}\otimes\1_{BC}\label{eq:P_XY_Q}.
\end{align}
\end{mylem}

The symmetry induced by triple twirling can now be exploited to substantially simplify the structure of this dual formulation, as formalized in the following lemma.

\begin{mylem}
{Dual Formulation Triple-Twirling Invariance}
{Dual_Formulation_TTI} 
The feasible set of the dual SDP $v_2^{\text{Dual}}(p_1, p_2)$ in Eq.~\eqref{eq:Approximate_Dual_Form_2} is invariant under triple twirling: 
if
\begin{align}
    \{P_{A}, Q_{A}, X_{AC}, Y_{AB}\}
\end{align}
is feasible, then so is 
\begin{align}
    \{\mT(P_{A}), \mT(Q_{A}), \mT(X_{AC}), \mT(Y_{AB})\},
\end{align}
where $\mT$ is the triple-twirling defined in Eq.~\eqref{eq:Twirling}.
\end{mylem}

Here we adopt the following convention: whenever the triple twirling map $\mT$ is applied to an operator that does not act on all three subsystems, the missing systems are implicitly initialized in the maximally mixed state. After the twirling operation, these auxiliary systems are traced out.
As an illustration, consider $\mT(P_{A})$ and $\mT(X_{AC})$. For the term $\mT(P_{A})$, one has
\begin{align}
    \mT(P_{A}):=
    &\Tr_{BC}\left[
    \int_{\text{Haar}}dU \,
    \left(U^*\otimes U\otimes U\right)\cdot P_{A}\otimes\frac{\1_{BC}}{d^2} \cdot \left(U^{\T}\otimes U^{\dagger}\otimes U^{\dagger}\right)\right]\\
    =
    &\int_{\text{Haar}}dU \,
    U^* P U^{\T}\\
    =
    &\frac{\Tr[P]}{d}\1_{A}.
\end{align}
Meanwhile, for $\mT(X_{AC})$, it follows that
\begin{align}
    \mT(X_{AC}):=
    &\Tr_{B}\left[
    \int_{\text{Haar}}dU \,
    \left(U^*\otimes U\otimes U\right)\cdot X_{AC}\otimes\frac{\1_{B}}{d} \cdot \left(U^{\T}\otimes U^{\dagger}\otimes U^{\dagger}\right)\right]\\
    =
    &\int_{\text{Haar}}dU \,
    \left(U^*\otimes U\right)\cdot X_{AC} \cdot \left(U^{\T}\otimes U^{\dagger}\right)\\
    =
    &x
    \raisebox{-1.8ex}{\includegraphics[height=1.8em]{I_2.pdf}}
    +
    y
    \raisebox{-1.8ex}{\includegraphics[height=1.8em]{UMES.pdf}},
\end{align}
holds for some parameters $x$ and $y$ determined by $X_{AC}$.

\begin{proof}
If $\{P_{A}, Q_{A}, X_{AC}, Y_{AB}\}$ is a feasible solution, its unitary conjugate under any $U$, namely 
\begin{align}
    \{
    U^* P_{A} U^{\T}, 
    U^* Q_{A} U^{\T}, 
    \left(U^*\otimes U\right)\cdot X_{AC} \cdot \left(U^{\T}\otimes U^{\dagger}\right),
    \left(U^*\otimes U\right)\cdot Y_{AB} \cdot \left(U^{\T}\otimes U^{\dagger}\right)
    \}
\end{align}
is also feasible. 
Haar averaging over all unitaries then completes the proof.
\end{proof}

In light of Lem.~\ref{lem:Dual_Formulation_TTI}, it is without loss of generality to restrict attention to solutions in which the operators $\{P_{A}, Q_{A}, X_{AC}, Y_{AB}\}$ in Eq.~\eqref{eq:Approximate_Dual_Form_2} have the following structure:
\begin{align}
    P_{A}&=p\1_{A},\\
    Q_{A}&=q\1_{A},\\
    X_{AC}&=x_1\1_{AC}+x_2\Gamma_{AC},\\
    Y_{AB}&=y_1\1_{AB}+y_2\Gamma_{AB}.
\end{align}
Now the cost function reads
\begin{align}
    &1-\Tr[P]+
    \Tr[J^{\mD_{p_2}}\cdot X]+
    \Tr[J^{\mD_{p_1}}\cdot Y]\\
    =
    &1-dp
    +d(x_1+y_1)+(p_2+d^2-p_2d^2)x_2+
    (p_1+d^2-p_1d^2)y_2\\
    =
    &1-d\left(p-(x_1+y_1)\right)
    +(p_2+d^2-p_2d^2)x_2
    +(p_1+d^2-p_1d^2)y_2,
\end{align}
Within this optimization problem, the quantities $p$, $x_i$ and $y_i$ ($i\in\{1,2\}$) are variables; 
all remaining symbols represent fixed scalar parameters.
Accordingly, Eq.~\eqref{eq:P_Q} can be rewritten as
\begin{align}
    \Tr[P-Q]=(p-q)\Tr[\1_A]=d(p-q)=2,
\end{align}
which implies that
\begin{align}\label{eq:p-q}
    p-q=\frac{2}{d}.
\end{align}
The remaining condition in Eq.~\eqref{eq:P_XY_Q} takes the form
\begin{align}
    p\1_{ABC}
    \geqslant
    (x_1+y_1)\1_{ABC}
    +
    x_2\Gamma_{AC}\otimes\1_{B}
    +
    y_2\Gamma_{AB}\otimes\1_{C}
    \geqslant
    q\1_{ABC},
\end{align}
leading to
\begin{align}
    \left(p-(x_1+y_1)\right)\1_{ABC}
    \geqslant
    x_2\Gamma_{AC}\otimes\1_{B}
    +
    y_2\Gamma_{AB}\otimes\1_{C}
    \geqslant
    \left(q-(x_1+y_1)\right)\1_{ABC}.
\end{align}
Substituting Eq.~\eqref{eq:p-q} into the above expression eliminates the variable $q$, yielding
\begin{align}
    \left(p-(x_1+y_1)\right)\1_{ABC}
    \geqslant
    x_2\Gamma_{AC}\otimes\1_{B}
    +
    y_2\Gamma_{AB}\otimes\1_{C}
    \geqslant
    \left(p-(x_1+y_1)-\frac{2}{d}\right)\1_{ABC}.
\end{align}
Upon redefining $z:=p-(x_1+y_1)$ and relabeling $x_2$ and $y_2$ as $x$ and $y$, respectively, the dual problem reduces to

\begin{mylem}
{Simplified Dual Formulation of Approximate Broadcasting}
{Dual_Formulation_AVB_S_1} 
The sample complexity $v_2^{\text{Dual}}(p_1, p_2)$ associated with approximate virtual broadcasting in Eq.~\eqref{eq:Approximate_Dual_Form_2} admits the following simplified formulation
\begin{align}\label{eq:Approximate_Dual_Form_3}
    v_2^{\text{Dual}}(p_1, p_2)
    =
    \max \quad 
    & 
    1-dz+ex+fy
    \\
    \text{s.t.} \quad 
    &z\1_{ABC}
    \geqslant
    x\Gamma_{AC}\otimes\1_{B}
    +
    y\Gamma_{AB}\otimes\1_{C}
    \geqslant
    \left(z-\frac{2}{d}\right)\1_{ABC},\label{eq:xGamma_yGamma}
\end{align}
where the coefficients $e$ and $f$ are defined as
\begin{align}
    e&:=p_2+d^2-p_2d^2=d^2(1-\epsilon_2),\label{eq:e}\\
    f&:=p_1+d^2-p_1d^2=d^2(1-\epsilon_1).\label{eq:f}
\end{align}
The above identities follow directly from Eq.~\eqref{eq:p_epsilon_i}.
\end{mylem}

Now the feasibility conditions of the above SDP are entirely characterized by the spectrum of the matrix $M(x,y)$ defined below
\begin{align}\label{eq:M_xy}
    M(x,y):=
    x\Gamma_{AC}\otimes\1_{B}
    +
    y\Gamma_{AB}\otimes\1_{C}.
\end{align}
Let $\lambda_{\max}$ and $\lambda_{\min}$ denote its largest and smallest eigenvalues, respectively. 
The constraints can therefore be reformulated as
\begin{align}
    \lambda_{\min}+\frac{2}{d}
    \geqslant
    z
    \geqslant
    \lambda_{\max},
\end{align}
and
\begin{align}
    \lambda_{\max}
    \geqslant
    \lambda_{\min}.
\end{align}
A detailed analysis of the spectrum of $M(x,y)$ is presented in Subsec.~\ref{subsec:M_xy}.

%%%%%%%%%%%%%%%%%%%%%%%%%%%%%%%%%%%%%%%%%%%%%%%%%%%%%%%%%%%%%%%%%%%%%%%%%%%%%%%%%%%%%%%%%%%%%%%%%%%%%%%%%%%%%%%%%%%%%%

\subsection{Spectral Analysis of $M(x,y)$}
\label{subsec:M_xy}

As discussed in Lem.~\ref{lem:Dual_Formulation_AVB_S_1}, determining the sample complexity of approximate virtual broadcasting ultimately reduces to analyzing the spectrum of the matrix $M(x,y)$ (see Eq.~\eqref{eq:M_xy}). 
This structure can be uncovered by exploiting the symmetry of $M(x,y)$.
As a first step, we introduce the following two vectors.
\begin{align}
    \ket{u_k}_{ABC}
    &:=
    \ket{\phi^{+}}_{AC}\otimes\ket{k}_{B},\\
    \ket{v_k}_{ABC}
    &:=
    \ket{\phi^{+}}_{AB}\otimes\ket{k}_{C},
\end{align}
and denote by
\begin{align}\label{eq:V_k}
    V_k:=\text{span}\{\ket{u_k}, \ket{v_k}\},
\end{align}
the subspace spanned by them.
It is now straightforward to verify that the actions of $\Gamma_{AC}\otimes\1_{B}$ and $\Gamma_{AB}\otimes\1_{C}$ on the vectors $\ket{u_k}$ and $\ket{v_k}$ gives the following relations
\begin{align}
    \Gamma_{AC}\otimes\1_{B}\ket{u_k}
    &=
    d\ket{u_k},\\
    \Gamma_{AC}\otimes\1_{B}\ket{v_k}
    &=
    \ket{u_k},\\
    \Gamma_{AB}\otimes\1_{C}\ket{u_k}
    &=
    \ket{v_k},\\
    \Gamma_{AB}\otimes\1_{C}\ket{v_k}
    &=
    d\ket{v_k}.
\end{align}
These relations allow us to determine the action of $M(x,y)$ on $\ket{u_k}$ and $\ket{v_k}$, which reads as follows
\begin{align}
    M(x,y)\ket{u_k}
    &=
    xd\ket{u_k}+y\ket{v_k},\\
    M(x,y)\ket{v_k}
    &=
    x\ket{u_k}+yd\ket{v_k}.
\end{align}
One may notice that the vectors $\ket{u_k}$ and $\ket{v_k}$ are not orthogonal, since
\begin{align}
    \braket{u_k}{v_k}=\frac{1}{d}.
\end{align}
To construct an orthonormal basis for the subspace $V_k$ (see Eq.~\ref{eq:V_k}), we apply the {\it Gram–Schmidt procedure} and introduce two new vectors $\ket{a_k}$ and $\ket{b_k}$ defined as
\begin{align}
    \ket{a_k}&:=\ket{u_k},\\
    \ket{b_k}&:=\frac{\ket{v_k}-\frac{1}{d}\ket{u_k}}{\sqrt{1-\frac{1}{d^2}}}.
\end{align}
The vectors $\ket{a_k}$ and $\ket{b_k}$ satisfy
\begin{align}
    \braket{a_k}{b_k}&=0,\\
    \braket{b_k}{b_k}&=1,
\end{align}
and 
\begin{align}
    V_k=\text{span}\{\ket{a_k}, \ket{b_k}\}.
\end{align}
Direct calculation shows that the action of $M(x,y)$ on $\ket{a_k}$ and $\ket{b_k}$ is characterized by
\begin{align}
    M(x,y)\ket{a_k}
    &=
    (xd+\frac{y}{d})\ket{a_k}
    +
    y\sqrt{1-\frac{1}{d^2}}\ket{b_k},\\
    M(x,y)\ket{b_k}
    &=
    y\sqrt{1-\frac{1}{d^2}}\ket{a_k}
    +
    yd(1-\frac{1}{d^2})\ket{b_k}.
\end{align}
Restricted to the subspace $V_k$, the operator $M(x,y)$ admits the following 2-by-2 matrix representation,
\begin{align}
M_2:=
\begin{pmatrix}
xd+\frac{y}{d} & y\sqrt{1-\frac{1}{d^2}} \\
y\sqrt{1-\frac{1}{d^2}} & yd(1-\frac{1}{d^2}) 
\end{pmatrix},
\end{align}
whose eigenvalues are given explicitly by
\begin{align}
    \lambda_{\pm}:=
    \frac{\Tr[M_2]\pm\sqrt{\Tr[M_2]^2-4\det(M_2)}}{2},
\end{align}
with
\begin{align}
    \Tr[M_2]&=(x+y)d,\\
    \det(M_2)&=xy(d^2-1).
\end{align}
Writing this out explicitly, we obtain
\begin{align}\label{eq:lambda_pm}
    \lambda_{\pm}=
    \frac{(x+y)d\pm\sqrt{(x-y)^2d^2+4xy}}{2}.
\end{align}

Before proceeding further, it is useful to analyze the zero eigenvalues of $M(x,y)$.
To this end, consider the rank of the matrix $M(x,y)$, 
\begin{align}
    \rank(M(x,y))
    =
    \rank (x\Gamma_{AC}\otimes\1_{B}
    +
    y\Gamma_{AB}\otimes\1_{C})
    \leqslant
    \rank(\Gamma_{AC}\otimes\1_{B})
    +
    \rank(\Gamma_{AB}\otimes\1_{C}))=2d.
\end{align}
which implies that $M(x,y)$ possesses at least $d^3-2d$ zero eigenvalues. 
Since $2d$ eigenvalues of the form $\lambda_{\pm}$ have already been identified, it follows that, under a suitable change of basis, the Hilbert space of the composite system $ABC$ can be decomposed as
\begin{align}
    \mH_{A}\otimes\mH_{B}\otimes\mH_{C}
    =
    \left(\bigoplus_{k=1}^{d}V_{k}\right)
    \oplus
    V^{\perp},
\end{align}
with $\dim V^{\perp}=d^3-2d$,
under which $M(x,y)$ takes the form
\begin{align}
    M(x,y)=
\begin{pmatrix}
\lambda_{+} & \, & \, &\, & \, & \, & \, \\
\, & \lambda_{-} & \, &\, & \, & \, & \, \\
\, & \, & \ddots &\, & \, & \, & \, \\
\, & \, & \, & \lambda_{+} &\, & \, & \, & \, \\
\, & \, & \, & \, & \lambda_{-} & \, & \, & \, \\
\, & \, & \, & \, & \, & 0 & \, & \, \\
\, & \, & \, & \, & \, & \, & \ddots & \, \\
\, & \, & \, & \, & \, & \, & \, & 0
\end{pmatrix}.
\end{align}
At this stage, the spectral structure of the matrix $M(x,y)$ (see Eq.~\eqref{eq:M_xy}), or equivalently $x\Gamma_{AC}\otimes\1_{B}+y\Gamma_{AB}\otimes\1_{C}$, which constitutes the core object in the constraint in Eq.~\eqref{eq:xGamma_yGamma}, has been fully characterized. 
This spectral information allows all matrix positivity conditions to be reduced to scalar constraints. 
As a result, the original optimization problem can be reformulated in the lemma.

\begin{mylem}
{Simplified Dual Formulation of Approximate Broadcasting}
{Dual_Formulation_AVB_S_2} 
The sample complexity $v_2^{\text{Dual}}(p_1, p_2)$ associated with approximate virtual broadcasting in Eq.~\eqref{eq:Approximate_Dual_Form_2} can be equivalently reformulated as the following form
\begin{align}\label{eq:Approximate_Dual_Form_4}
    v_2^{\text{Dual}}(p_1, p_2)
    =
    \max \quad 
    & 
    1-dz+ex+fy
    \\
    \text{s.t.} \quad 
    &\frac{2}{d}
    \geqslant
    z
    \geqslant
    0,\label{eq:SOCP_1}\\
    &\frac{2}{d}+\lambda_{-}
    \geqslant
    z
    \geqslant
    \lambda_{+},\label{eq:SOCP_2}
\end{align}
where $\lambda_{\pm}$ are given in Eq.~\eqref{eq:lambda_pm}, while $e$ and $f$ are introduced in Eqs.~\eqref{eq:e} and~\eqref{eq:f}, respectively.
\end{mylem}

In simplifying the optimization in Lem.~\ref{lem:Dual_Formulation_AVB_S_2}, we used only the fact that $M(x,y)$ admits three eigenvalues, namely $\{\lambda_{+}, \lambda_{-}, 0\}$. 
Importantly, we did not assume that $\lambda_{-}<0$, as this condition does not always hold. 
We discuss this point in more detail in the Subsec.~\ref{subsec:AVB_SOCP}.

%%%%%%%%%%%%%%%%%%%%%%%%%%%%%%%%%%%%%%%%%%%%%%%%%%%%%%%%%%%%%%%%%%%%%%%%%%%%%%%%%%%%%%%%%%%%%%%%%%%%%%%%%%%%%%%%%%%%%%

\subsection{Approximate Virtual Broadcasting via SOCP}
\label{subsec:AVB_SOCP}

The optimization problem characterizing the sample complexity required for approximate virtual broadcasting (see Subsec.~\ref{subsec:AVB}) is already significantly simpler than the original SDP formulation in Eq.~\eqref{eq:SDP_Approximate}, as all matrix constraints have been eliminated and replaced by scalar relations derived from the spectral structure of $M(x,y)$ (see Eq.~\eqref{eq:M_xy}).
This reduction, however, is not yet optimal. 
By further exploiting the relations obtained in Eq.~\eqref{eq:Approximate_Dual_Form_3}, the problem admits an alternative formulation that can be cast as a {\it Second-Order Cone Programming} (SOCP)~\cite{Boyd_Vandenberghe_2004}. 
Compared with general SDPs, SOCPs possess a more favorable computational complexity and can typically be solved more efficiently using convex optimization methods~\cite{Alizadeh2003}.
The remainder of this subsection develops this reformulation explicitly and shows how the characterization of the sample complexity for approximate virtual broadcasting, namely Eq.~\eqref{eq:Approximate_Dual_Form_4}, can be expressed in the SOCP form.

As a first step, introduce a new variable $z$ to replace the original quantity $dz-1$. 
With this substitution, the expression for Eq.~\eqref{eq:Approximate_Dual_Form_4} can be rewritten in the following equivalent form.

\begin{align}\label{eq:SOCP_opt_1}
    v_2^{\text{Dual}}(p_1, p_2)
    =
    \max \quad 
    & 
    ex+fy - z
    \\
    \text{s.t.} \quad 
    & 1
    \geqslant
    z
    \geqslant
    -1, \\
    &d\lambda_{-} + 1
    \geqslant
    z
    \geqslant
    d \lambda_{+} - 1.
\end{align}

Now consider the coefficients $\lambda_{+}$ and $\lambda_{-}$, which are given by
\begin{align}
    \lambda_{\pm}=
    \frac{(x+y)d\pm\sqrt{(x-y)^2d^2+4xy}}{2}.
\end{align}
In the second step, it is convenient to introduce a change of variables. 
Specifically, new variables $x$ and $y$ are defined to represent the linear combinations $x+y$ and $x-y$ appearing in the original formulation of Eq.~\eqref{eq:SOCP_opt_1}. 
This reparameterization simplifies the structure of the optimization problem and leads to 
\begin{align}\label{eq:lambda_pm_2}
    \lambda_{\pm}=\frac{xd\pm\sqrt{y^2(d^2-1)+x^2}}{2},
\end{align}
while $v_2^{\text{Dual}}(p_1, p_2)$ turns out to be
\begin{align}\label{eq:SOCP_opt_2}
    v_2^{\text{Dual}}(p_1, p_2)
    =
    \max \quad 
    & 
    \frac{e+f}{2}x+\frac{e-f}{2}y - z
    \\
    \text{s.t.} \quad 
    & 1
    \geqslant
    z
    \geqslant
    -1, \\
    &\frac{d}{2}\left(xd-\sqrt{y^2(d^2-1)+x^2}\right) + 1
    \geqslant
    z
    \geqslant
    \frac{d}{2}\left(xd+\sqrt{y^2(d^2-1)+x^2}\right) - 1.
\end{align}
We now return to the question raised at the end of Subsec.~\ref{subsec:M_xy} as to why $\lambda_{-}$ cannot be assumed to be always negative. 
Equation~\eqref{eq:lambda_pm_2} shows that $\lambda_{-}>0$ when $x>0$ and $y=0$, whereas $\lambda_{-}<0$ when $x=0$ and $y\neq0$. 
The sign of $\lambda_{-}$ therefore depends on the values of $x$ and $y$.
For further simplification, we rescale the variable by substituting $xd/2\to x$ and $yd\sqrt{d^2-1}/2\to y$, which yields
\begin{align}\label{eq:SOCP_3}
    v_2^{\text{Dual}}(p_1, p_2)
    =
    \max \quad 
    & 
    gx+hy - z
    \\
    \text{s.t.} \quad 
    & 1
    \geqslant
    z
    \geqslant
    -1, \\
    &1-\sqrt{y^2+x^2}
    \geqslant
    z-xd
    \geqslant
    -1+\sqrt{y^2+x^2},
\end{align}
with the coefficients $g$ and $h$ given by
\begin{align}
    g&:=\frac{e+f}{d},\label{eq:g}\\
    h&:=\frac{e-g}{d \sqrt{d^2-1}}.\label{eq:h}
\end{align}
Here, the quantities $e$ and $f$ are specified in Eqs.~\eqref{eq:e} and~\eqref{eq:f}, respectively.
Note that the term $\sqrt{y^2+x^2}$ appearing in the constraints corresponds precisely to the Euclidean norm, also known as $L^2$ norm, of the vector $(x, y)^{\T}$, namely
\begin{align}
    \left\|\begin{pmatrix}
        x \\
        y
    \end{pmatrix}\right\|_2
    =
    \sqrt{y^2+x^2}.
\end{align}
This observation reveals that the feasible region is characterized by second-order cone constraints. 
Consequently, the optimization problem admits a formulation as a Second-Order Cone Programming (SOCP), as summarized in the following theorem.

\begin{mythm}
{Approximate Broadcasting via SOCP}
{AVB_SOCP} 
The sample complexity $v_2^{\text{Dual}}(p_1, p_2)$ associated with approximate virtual broadcasting in Eq.~\eqref{eq:Approximate_Dual_Form_2} admits an equivalent formulation as the following Second-Order Cone Programming (SOCP).
\begin{align}\label{eq:AVB_SOCP}
    v_2^{\text{Dual}}(p_1, p_2)
    =
    \max \quad 
    & 
    gx+hy-z
    \\
    \text{s.t.} \quad 
    &1\geqslant|z|,\label{eq:AVB_SOCP_1}\\
    &1-\left\|\begin{pmatrix}
        x \\
        y
    \end{pmatrix}\right\|_2\geqslant|xd - z|,\label{eq:AVB_SOCP_2}
\end{align}
The coefficients $g$ and $h$ appearing in the linear objective function are defined in Eqs.~\eqref{eq:g} and~\eqref{eq:h}, respectively, and take the explicit form
\begin{align}
    g&=d(2-(\epsilon_1+\epsilon_2)),\label{eq:g_explicit}\\
    h&=\frac{d}{\sqrt{d^2-1}}(\epsilon_1-\epsilon_2).\label{eq:h_explicit}
\end{align}
\end{mythm}

Up to this point, we have derived an alternative formulation of the dual problem for the sample complexity in terms of SOCP, which in general entails lower computational complexity than the original SDP formulation.
However, this representation still does not provide an analytic solution.
As obtaining such a solution requires a more detailed analysis, we defer this task to the subsequent subsection.

%%%%%%%%%%%%%%%%%%%%%%%%%%%%%%%%%%%%%%%%%%%%%%%%%%%%%%%%%%%%%%%%%%%%%%%%%%%%%%%%%%%%%%%%%%%%%%%%%%%%%%%%%%%%%%%%%%%%%%

\subsection{Analytic Solution for Approximate Virtual Broadcasting}
\label{subsec:Analytic_Solution}

In this subsection, we continue the analysis of the optimal value $v_2^{\text{Dual}}(p_1, p_2)$ (see Lem.~\ref{lem:Dual_Formulation_AVB}) and derive its analytic solution, thereby fully resolving the problem.

The quantities $\sqrt{d^2-1}$, $1 - \epsilon_1$, and $1 - \epsilon_2$ appear repeatedly in the analytic analysis. For convenience, we introduce the following shorthand notations:

\begin{align}
    k &:= \sqrt{d^2-1},\label{eq:k}\\
    t_1 &:= 1 - \epsilon_1,\label{eq:t_1}\\
    t_2 &:= 1-\epsilon_2.\label{eq:t_2}
\end{align}

Among the coefficients $g$ (see Eq.~\eqref{eq:g}), $h$ (see Eq.~\eqref{eq:h}), and $k$ (see Eq.~\eqref{eq:k}), two inequalities play an important role in the analysis. 
We summarize these observations in the following lemma.

\begin{mylem}{Key Inequalities for the Coefficients}{valid_bound}
Let $g$, $h$, and $k$ be the coefficients defined in Eqs.~\eqref{eq:g}, \eqref{eq:h}, and \eqref{eq:k}. 
Then the following inequalities hold:
\begin{align}
    &1\geqslant \frac{k}{d}\geqslant|h|,\label{eq:Inequaility_ghk_1}\\
    &g^2 > d^2 h^2 \geqslant k^2 h^2.\label{eq:Inequaility_ghk_2}
\end{align}
Here $d$ denotes the dimension of the systems involved in the broadcasting task. 
Specifically, for broadcasting from system $A$ to systems $B$ and $C$, we assume $d=\dim A= \dim B= \dim C$.
\end{mylem}

\begin{proof}
For the error parameters $\epsilon_i$ allowed in approximate broadcasting, we have 
\begin{align}
    \epsilon_i \in [0, 1-\frac{1}{d^2}], 
    \quad\forall\, i\in\{1, 2\}.
\end{align}
Equivalently, in terms of $t_i$, this relation can be written as
\begin{align}
    t_i \in  \qty[\frac{1}{d^2}, 1], 
    \quad\forall\, i\in\{1, 2\}.
\end{align}
It is then straightforward to verify that
\begin{align}
    |t_2 -t_1| \leqslant 1- \frac{1}{d^2} = \frac{d^2 -1}{d^2} = \frac{k^2}{d^2},
\end{align}
where the first inequality in the chain follows from
\begin{align}
    |h| = \frac{d}{k}|t_2 - t_1| \leqslant \frac{k}{d} \leqslant 1.
\end{align}

From the definition of $k$ in Eq.~\eqref{eq:k}, the second part of Eq.~\eqref{eq:Inequaility_ghk_2} follows immediately. 
We now turn to the first inequality in Eq.~\eqref{eq:Inequaility_ghk_2}, namely
\begin{align}
    g^2 - d^2 h^2
    &= d^2\qty((t_1+t_2)^2 - \frac{d^2}{d^2-1} (t_2-t_1)^2) \\
    &= d^2\qty(4 t_1 t_2 - \frac{1}{d^2-1} (t_2-t_1)^2) \\
    &= d^2  t_1 t_2 \qty(4 - \frac{1}{d^2-1} \qty(\frac{t_2}{t_1}+\frac{t_1}{t_2} - 2)) \\
    &\geqslant d^2  t_1 t_2 \qty(4 - \frac{1}{d^2-1} \qty(d^2 + \frac{1}{d^2} - 2)) \label{ln:max_t}\\
    &= d^2  t_1 t_2 \qty(4 - \frac{1}{d^2-1} \qty(d^2 -1)+ \frac{1}{d^2-1} \qty(1 - \frac{1}{d^2})) \\
    &\geqslant d^2  t_1 t_2 \qty(4 - 1) \\
    &> 0 \ .
\end{align}
The inequality in Eq.~\eqref{ln:max_t} follows from the bound
\begin{align}
    \frac{t_2}{t_1} \leqslant d^2,
\end{align}
which implies
\begin{align}
    \frac{t_2}{t_1} + \frac{t_1}{t_2} \leqslant d^2 + \frac{1}{d^2}.
\end{align}
The inequality is strict since $t_1, t_2 >0$.
\end{proof}

To simplify the square-root expressions $\sqrt{y^2+x^2}$ arising in the constraint of Eq.~\eqref{eq:AVB_SOCP_2}, we introduce a parametrization in terms of hyperbolic functions, which leads to the following theorem.

\begin{mythm}{Hyperbolic Parametrization of SOCP}{Hyperbolic_Parametrization}
The sample complexity $v_2^{\text{Dual}}(p_1, p_2)$ for approximate virtual broadcasting in Eq.~\eqref{eq:Approximate_Dual_Form_2} admits the following equivalent hyperbolic parametrization.
    \begin{align}
        v_2^{\text{Dual}}(p_1, p_2)
        = \max \quad  \label{ln:hyperbolic}
        &  \qty(g+ h\sinh{\theta})x - z \\
        \text{s.t.} \quad 
        &1\geqslant|z|,\label{ln:hyperbolic_1}\\
        &1-\abs{x} \cosh{\theta} \geqslant \abs{xd - z} \ .\label{ln:hyperbolic_2}
    \end{align}
\end{mythm}

\begin{proof}
The original optimization can be divided into two cases. 
If $x=0$, the maximal value is 1. 
If $x\neq0$, we introduce the hyperbolic parametrization $\sinh(\theta)=y/x$, which eliminates the square-root constraint. 
Since the maximal value in this regime is never smaller than 1, it suffices to consider this case.

Let $\max_{(x,y,z)\in S}\xi(x,y,z)$ denote the optimal value of the optimization problem in Eq.~\eqref{eq:AVB_SOCP}, with
\begin{align}
    \xi(x, y, z):=gx+hy-z,
\end{align}
and 
\begin{align}
    S:=\left\{(x,y,z)\,\middle\vert\, 1\geqslant|z|, 1-\left\|\begin{pmatrix}
        x \\
        y
    \end{pmatrix}\right\|_2\geqslant|xd - z|\right\}\subset\mathbb{R}^3.
\end{align}

We now analyze this problem in detail, beginning with the first case $x=0$, for which we obtain
\begin{align}
    \max_{(x,y,z)\in S}
    \left\{\xi(x,y,z)\middle\vert_{x=0}\right\}
    =
    \max_{y} \max_{z} \quad 
    &  hy-z \\
    \text{s.t.} \quad 
    & 1 \geqslant|z|,\\
    & 1 - |y| \geqslant|z| \ .
\end{align}
In the above optimization, the maximal value is achieved by setting $z=-(1 - |y|)$, which leads to
\begin{align}
    hy-z=hy+1-|y|\leqslant h\cdot|y|+1-|y|=1+(h-1)|y|\leqslant1.
\end{align}
The final inequality follows from the bound $|h|\leqslant1$ in Eq.~\eqref{eq:Inequaility_ghk_1} of Lem.~\ref{lem:valid_bound}.
It then follows that
\begin{align}
    \max_{(x,y,z)\in S}
    \left\{\xi(x,y,z)\middle\vert_{x=0}\right\}
    =
    \xi(0,0,-1)
    =
    1.
\end{align}

We proceed to the case $x\neq0$, where the ratio $y/x$ can be expressed through a hyperbolic parametrization. 
More precisely, by setting
\begin{align}
    \sinh(\theta) := \frac{y}{x},
\end{align}
we have
\begin{align}
    \left\|\begin{pmatrix}
        x \\
        y
    \end{pmatrix}\right\|_2
    =
    \abs{x}\sqrt{1 + \sinh^2{\theta}}
    =
    \abs{x} \cosh{\theta},
\end{align}
and the optimization problem becomes
\begin{align}
    \max_{(x,y,z)\in S}
    \left\{\xi(x,y,z)\Big|_{x\neq0}\right\}
    =
    \max \quad 
    &  \qty(g+ h\sinh{\theta})x - z \\
    \text{s.t.} \quad 
    &1\geqslant|z|,\\
    &1-\abs{x} \cosh{\theta} \geqslant \abs{xd - z}.
\end{align}
Under this parametrization, the variables $x$, $y$, and $z$ are expressed in terms of $\theta$, $x$, and $z$. 
For clarity, we introduce the function $\zeta(\theta, x, z)$,
\begin{align}\label{eq:zeta}
    \zeta(\theta, x, z):= \qty(g+ h\sinh{\theta})x - z,
\end{align}
subject to the following constraints
\begin{align}\label{eq:Set_T}
    T
    :=
    \left\{(\theta,x,z)\,|\, 1\geqslant|z|, 1-\abs{x} \cosh{\theta} \geqslant \abs{xd - z}\right\}
    \subset\mathbb{R}^3.
\end{align}
From the above analysis, we have
\begin{align}
    \max_{(x,y,z)\in S}
    \left\{\xi(x,y,z)\Big|_{x\neq0}\right\}
    =
    \max_{(\theta,x,z)\in T}
    \left\{\zeta(\theta,x,z)\Big|_{x\neq0}\right\}.
\end{align}

To complete the proof, we consider the case $x=0$ in $\zeta(\theta, x, z)$, which yields
\begin{align}
    \max_{(\theta,0,z)\in T}
    \left\{\zeta(\theta,0,z)\right\}
    =& \max_{|z|\leqslant 1} - z \\
    =&\zeta(\theta,0,-1) \\
    =&1 \\
    =&\max_{(0,y,z)\in S} \left\{\xi(0,y,z)\right\} \ .
\end{align}
Equivalently, we have 
\begin{align}
    \max_{(x,y,z)\in S}\xi(x,y,z)=\max_{(\theta,x,z)\in T}\zeta(\theta,x,z),
\end{align}
which completes the proof.
\end{proof}

\begin{mylem}{Bounds on the Coefficients}{f_bound}
If the absolute value of the hyperbolic sine of $\theta$ is bounded below by $k$ (defined in Eq.~\eqref{eq:k}),
\begin{align}
    |\sinh{\theta}| \geqslant k,
\end{align}
then the following chain of inequalities holds
\begin{equation}
    \frac{1 + \cosh{\theta}}{d} \geqslant g+h\sinh\theta -d \geqslant \frac{1 - \cosh{\theta}}{d}.
\end{equation}
\end{mylem}

\begin{proof}
We first examine the coefficient of $x$ appearing in the objective function of the optimization problem in Eq.~\eqref{ln:hyperbolic}.
\begin{align}
    g + h\sinh{\theta}
    & = d\qty((t_1+t_2) + \frac{d}{k}(t_2-t_1)\sinh{\theta}) \\
    & = d\qty(t_1(1-\frac{\sinh{\theta}}{k})+t_2(1+\frac{\sinh{\theta}}{k})) \ .
\end{align}
For the case $\sinh{\theta} \geqslant k$, we have
\begin{align}
    \frac{\sinh{\theta}}{k}\geqslant1,
\end{align}
and hence
\begin{align}
    1 - \frac{\sinh{\theta}}{k} &\leqslant 0,\\
    1 + \frac{\sinh{\theta}}{k} &\geqslant 0.
\end{align}
Taking the boundary values of $t_1$ and $t_2$ then yields the following lower bound on $g + h\sinh{\theta}$.
\begin{align}
    g + h\sinh{\theta}
    & \geqslant d\qty((1 - \frac{\sinh{\theta}}{k})+\frac{1}{d^2}(1+\frac{\sinh{\theta}}{k})) \\
    & = d\qty(\frac{d^2+1}{d^2} - \frac{\sinh{\theta}}{k} + \frac{\sinh{\theta}}{kd^2}) \\
    & = d\qty(\frac{d^2+1}{d^2} - \frac{\sinh{\theta}}{d^2}) \\
    & = \frac{d^2 + 1 - \sinh{\theta}}{d} \ .
\end{align}
The corresponding upper bound can be derived as
\begin{align}
    g + h\sinh{\theta}
    & \leqslant d\qty(\frac{1}{d^2}(1 - \frac{\sinh{\theta}}{k})+(1+\frac{\sinh{\theta}}{k})) \\
    & = d\qty(\frac{d^2+1}{d^2} + \frac{\sinh{\theta}}{k} - \frac{\sinh{\theta}}{kd^2}) \\
    & = d\qty(\frac{d^2+1}{d^2} + \frac{\sinh{\theta}}{d^2}) \\
    & = \frac{d^2 + 1 + \sinh{\theta}}{d} \ .
\end{align}

An analogous bound holds for the case $\sinh{\theta} \leqslant -k$.
In summary, when $\sinh{\theta}\geqslant k$, we obtain
\begin{align}
    \frac{d^2 + 1 + \sinh{\theta}}{d} \geqslant g + h\sinh{\theta} \geqslant \frac{d^2 + 1 - \sinh{\theta}}{d}.
\end{align}
Conversely, if $\sinh{\theta}\leqslant -k$, the following chain of inequalities follows
\begin{align}
    \frac{d^2 + 1 - \sinh{\theta}}{d} \geqslant g + h\sinh{\theta} \geqslant \frac{d^2 + 1 + \sinh{\theta}}{d}.
\end{align}
Combining these chains of inequalities yields
\begin{align}
    \frac{d^2 + 1 + |\sinh{\theta}|}{d} \geqslant g + h\sinh{\theta} \geqslant \frac{d^2 + 1 - |\sinh{\theta}|}{d},
\end{align}
together with $|\sinh{\theta}| \leqslant \cosh{\theta}$, the main statement follows.
\end{proof}

We next analyze the constraint of Eq.~\eqref{ln:hyperbolic_1} and \eqref{ln:hyperbolic_2}, i.e., $|z| \leqslant 1$, and $1-\abs{x} \cosh{\theta} \geqslant \abs{xd - z}$.
The signs of $x$ and $xd - z$ partition the $(z,x)$-plane into four distinct regions.
The following lemma characterizes their geometry and determines the corresponding optimal value in each region.

\begin{mylem}{Boundary Points}{feasible}
The constraints in Eq.~\eqref{ln:hyperbolic_1} and \eqref{ln:hyperbolic_2} define a feasible region in the $(z,x)$-plan.
The boundary of this region is fully characterized by the maximal and minimal values of $x$ across different ranges of $z$.
In particular, the maximal value $x$, corresponding to the upper boundary of the feasible region, is given by
\begin{align}
  \max{x} &= \left\{
  \begin{aligned}
      &\frac{z+1}{d+\cosh{\theta}} \quad \text{if}& \quad z \leqslant \frac{d}{\cosh{\theta}}, \\
      &\frac{z-1}{d-\cosh{\theta}} \quad \text{if}& \quad 1 \geqslant z \geqslant \frac{d}{\cosh{\theta}}.
  \end{aligned}
  \right. \label{eq:region1}
\end{align}
The minimal value of $x$, corresponding to the lower boundary of the feasible region, is characterized by
\begin{align}
    \min{x} &= \left\{
  \begin{aligned}
      &\frac{z-1}{d+\cosh{\theta}} \quad \text{if}& \quad z \geqslant -\frac{d}{\cosh{\theta}}, \\
      &\frac{z+1}{d-\cosh{\theta}} \quad \text{if}& \quad -1\leqslant z \leqslant -\frac{d}{\cosh{\theta}}.
  \end{aligned}
  \right. \label{eq:region2}
\end{align}
\end{mylem}

In particular, in case $\cosh{\theta} \geqslant d$, the region $1\geqslant z \geqslant \frac{d}{\cosh\theta}$ in \eqref{eq:region1} and the region $-1 \leqslant z\leqslant \frac{d}{\cosh\theta}$ do not exist.
In that case, the feasible region of $x$ is sandwiched by two straight lines.

\begin{figure}[t]
    \centering
    \includegraphics[width=0.9\linewidth]{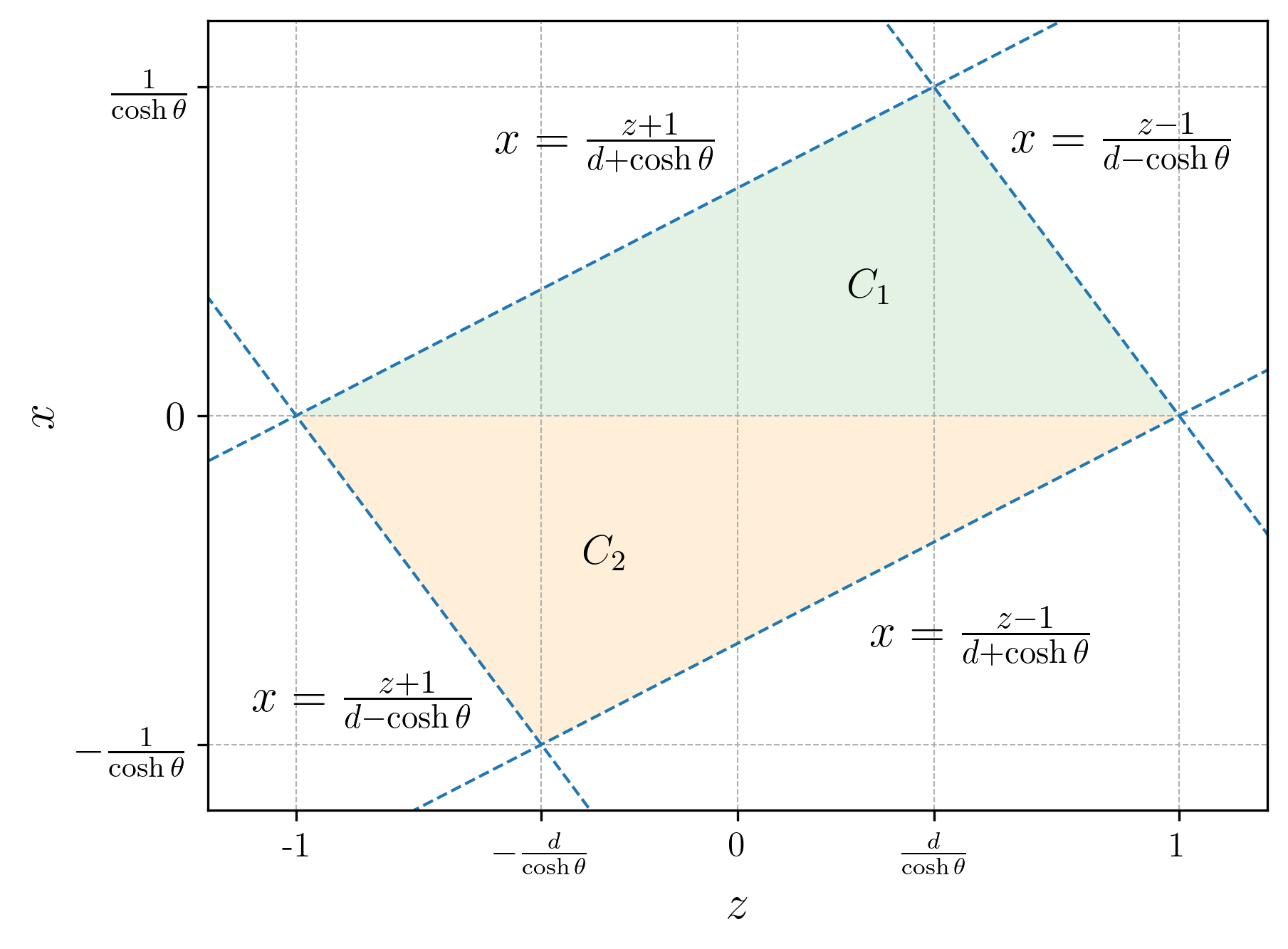}
    \caption{(Color online) \textbf{Partition of the Feasible Region}.
    Constraints~\eqref{ln:hyperbolic_1} and \eqref{ln:hyperbolic_2} define connected regions $C_1$ (see Eq.~\eqref{eq:Set_C_1}) and $C_2$ (see Eq.~\eqref{eq:Set_C_2}) in the $(z,x)$-plan, highlighted by the colored areas in the figure.}
    \label{fig:feasible_x}
\end{figure}

\begin{proof}
According to the signs of $x$ and $xd - z$, the $(z,x)$-plane can be divided into four distinct regions. 
However, only two of these are relevant within the feasible set. 
We begin by constructing the first region, denoted by $C_1$ and illustrated in Fig.~\ref{fig:feasible_x}.
In this region we consider the case $x\geqslant0$. 
Depending on the sign of $xd - z$, Eq.~\eqref{ln:hyperbolic_2} reduces to the following two inequalities
\begin{align}
    xd - z\geqslant0 \quad&\Rightarrow\quad z+1\geqslant (d+\cosh{\theta})x, \\
    xd - z\leqslant0 \quad&\Rightarrow\quad z-1\leqslant (d-\cosh{\theta})x. \label{ln:px_bound}
\end{align}
If $d-\cosh{\theta} \geqslant 0$, Eq.~\eqref{ln:px_bound} reduces to
\begin{align}
    \frac{z-1}{d-\cosh{\theta}} \leqslant x.
\end{align}
Since $x\geqslant0$ and $z\leqslant1$, this inequality is automatically satisfied and therefore imposes no additional constraint on $x$ and $z$.
By contrast, when $d-\cosh{\theta} \leqslant 0$, Eq.~\eqref{ln:px_bound} becomes
\begin{align}
    \frac{z-1}{d-\cosh{\theta}} \geqslant x.
\end{align}
The region $C_1$ can now be defined as
\begin{align}\label{eq:Set_C_1}
    C_1:=\left\{(x,z)\,\middle\vert\, x\geqslant0, 
    x\leqslant\frac{z+1}{d+\cosh{\theta}},
    x\leqslant\frac{z-1}{d-\cosh{\theta}}
    \right\}\subset\mathbb{R}^2,
\end{align}
which is highlighted in green in Fig.~\ref{fig:feasible_x}.

The second region, denoted by $C_2$ and illustrated in Fig.~\ref{fig:feasible_x}, corresponds to the case $x\leqslant0$. 
In this regime, Eq.~\eqref{ln:hyperbolic_2}, i.e., $1-\abs{x} \cosh{\theta} \geqslant \abs{xd - z}$, reduces to
\begin{align}
    xd - z\leqslant0 \quad&\Rightarrow\quad z-1\leqslant (d+\cosh{\theta})x, \\
    xd - z\geqslant0 \quad&\Rightarrow\quad z+1\geqslant (d-\cosh{\theta})x. \label{ln:nx_bound}
\end{align}
If $d-\cosh{\theta} \geqslant 0$, Eq.~\eqref{ln:nx_bound} reduces to
\begin{align}
    \frac{z+1}{d-\cosh{\theta}} \geqslant x.
\end{align}
Because $x\leqslant0$ and $z\geqslant-1$, this condition is automatically fulfilled and thus introduces no additional restriction on the feasible region in the $(z,x)$-plane. 
In contrast, when $d-\cosh{\theta} \leqslant 0$, Eq.~\eqref{ln:nx_bound} becomes
\begin{align}
    \frac{z+1}{d-\cosh{\theta}} \leqslant x.
\end{align}
The corresponding region, denoted by $C_2$, is therefore defined as
\begin{align}\label{eq:Set_C_2}
    C_2:=\left\{(x,z)\,\middle\vert\, x\leqslant0, 
    x\geqslant\frac{z-1}{d+\cosh{\theta}},
    x\geqslant\frac{z+1}{d-\cosh{\theta}}
    \right\}\subset\mathbb{R}^2,
\end{align}
and is highlighted in yellow in Fig.~\ref{fig:feasible_x}.

The above analysis and the regions corresponding to the different cases can be visualized geometrically in the $(z,x)$-plane, as illustrated in Fig.~\ref{fig:feasible_x}. 
From this representation, it becomes clear that when $z\leqslant d/\cosh{\theta}$, the maximal value of $x$ is achieved at $(z+1)/(d+\cosh{\theta})$.
By contrast, when $z\geqslant d/\cosh{\theta}$, which is only possible if $\cosh{\theta} \geqslant d$, the maximal value of $x$ is achieved at $(z-1)/(d-\cosh{\theta})$.
In addition to the point where $x$ attains its maximum, we also consider the cases where $x$ reaches its minimum. When $z\geqslant-d/\cosh{\theta}$, the minimal value of $x$ is attained at $(z-1)/(d+\cosh{\theta})$. 
By contrast, when $z\leqslant-d/\cosh{\theta}$, which is only possible if $\cosh{\theta} \geqslant d$, the minimum value of $x$ is achieved at $(z+1)/(d-\cosh{\theta})$.
\end{proof}

Equipped with the feasible regions shown in Fig.~\ref{fig:feasible_x} and the corresponding extremal values, we return to the objective function $(g + h\sinh{\theta}) x - z$ of Eq.~\eqref{ln:hyperbolic}. 
For convenient, we define $f(\theta)$ as
\begin{equation}
    f(\theta) = \frac{g+h\sinh\theta}{d + \cosh{\theta}}. \label{eq:def_f}
\end{equation}
The following lemma summarizes the result for the optimization problem with a non-negative objective coefficient $g+ \sinh\theta$.

\begin{mylem}{Non-Negative Objective Coefficient}{upper_bound_opt}
    
    For any fixed $\theta$, the optimal value of the optimization problem below
    \begin{align}
        \max \quad 
        &  (g+ \sinh\theta) x - z \\
        \text{s.t.} \quad 
        & 1 \geqslant|z|,\\
        & 1 - \abs{x} \cosh{\theta} \geqslant \abs{xd - z} \\
        & g+ \sinh\theta \geqslant 0 \ .
    \end{align}
    is given by
    \begin{align}\label{eq:opt_max_case_1}
        \max_{\theta} \left\{2f(\theta) - 1, 1\right\},
    \end{align}
    where $f(\theta)$ is defined in Eq.~\eqref{eq:def_f}.
\end{mylem}

\begin{proof}
Since $g+h\sinh\theta \geqslant 0$, the maximum is attained at the largest admissible value of $x$ for each fixed $z$.
According to Lem.~\ref{lem:feasible}, this maximal value depends on $z$, leading naturally to two distinct regimes. We begin with the case $z \leqslant d/\cosh{\theta}$, for which the objective function $\zeta(\theta, x, z)$ (see Eq.~\eqref{eq:zeta}) becomes
\begin{align}
    \zeta(\theta, x, z)\Big|_{x=\frac{z+1}{d+\cosh{\theta}}}
    =
    \left((g+ \sinh\theta) x - z\right)
    \Big|_{x=\frac{z+1}{d+\cosh{\theta}}}
    =
    (g+ \sinh\theta) \frac{z+1}{d+\cosh{\theta}} - z.
\end{align}
As the objective function is linear in $z$, the maximum occurs at the boundary values $\pm1$ or $d/\cosh{\theta}$ (see Lem.~\ref{lem:feasible}). 
Substituting these into the above expression yields
\begin{align}
    \zeta(\theta, \frac{z+1}{d+\cosh{\theta}}, -1)&=1,\\
    \zeta(\theta, \frac{z+1}{d+\cosh{\theta}}, \frac{d}{\cosh{\theta}})&=
    \frac{(g+h\sinh\theta)}{d+\cosh{\theta}}(\frac{d}{\cosh{\theta}}+1) - \frac{d}{\cosh{\theta}} = \frac{(g+h\sinh\theta) - d}{\cosh{\theta}},\\
    \zeta(\theta, \frac{z+1}{d+\cosh{\theta}}, -1)&=\frac{2(g+h\sinh\theta)}{d+\cosh{\theta}} - 1 = 2f(\theta) - 1.
\end{align}
Hence, in this case, the maximal value is bounded above by
\begin{align}\label{eq:NN_Case_1_Max}
    \max_{\theta}\left\{2f(\theta) - 1, \frac{(g+h\sinh\theta) - d}{\cosh{\theta}}, 1\right\}.
\end{align}

In contrast, when $1 \geqslant z \geqslant \frac{d}{\cosh{\theta}}$, the feasible region for $z$ has two boundary points, namely $1$ and $d/\cosh{\theta}$ (see Lem.~\ref{lem:feasible}).
In this case, the objective function $\zeta(\theta, x, z)$ becomes
\begin{align}
    \zeta(\theta, x, z)\Big|_{x=\frac{z-1}{d-\cosh{\theta}}}
    =
    \left((g+ \sinh\theta) x - z\right)
    \Big|_{x=\frac{z-1}{d-\cosh{\theta}}}
    =
    (g+ \sinh\theta) \frac{z-1}{d-\cosh{\theta}} - z.
\end{align}
Evaluating the objective function at these points gives
\begin{align}
    \zeta(\theta, \frac{z-1}{d-\cosh{\theta}}, 1)&=-1,\\
    \zeta(\theta, \frac{z-1}{d-\cosh{\theta}}, \frac{d}{\cosh{\theta}})&=
    \frac{(g+h\sinh\theta)}{d-\cosh{\theta}}(\frac{d}{\cosh{\theta}}-1) - \frac{d}{\cosh{\theta}} = \frac{(g+h\sinh\theta) - d}{\cosh{\theta}}.
\end{align}
It follows that the maximal value in this case is
\begin{align}\label{eq:NN_Case_2_Max}
    \max_{\theta}\left\{\frac{(g+h\sinh\theta) - d}{\cosh{\theta}}, -1\right\}.
\end{align}
By combining Eqs.~\eqref{eq:NN_Case_1_Max} and~\eqref{eq:NN_Case_2_Max}, we immediately obtain the following upper bound.
\begin{align}\label{eq:NN_Case_1_bound}
    \max_{\theta}\left\{2f(\theta) - 1, \frac{(g+h\sinh\theta) - d}{\cosh{\theta}}, 1, -1\right\}
    =
    \max_{\theta}\left\{2f(\theta) - 1, \frac{(g+h\sinh\theta) - d}{\cosh{\theta}}, 1\right\}
    .
\end{align}

Feasible solutions can always be found for which the maximal admissible value of $((g+h\sinh\theta) - d)/\cosh{\theta}$ is bounded by either $1$ or $2f(\theta) - 1$.
The maximization over $\theta$ can therefore be reduced to these two cases.
We first consider $\cosh{\theta} \leqslant d$, which yields
\begin{align}
2f(\theta) - 1=
    \frac{2(g+h\sinh\theta)}{d+\cosh{\theta}} - 1 \geqslant \frac{2(g+h\sinh{\theta})}{2\cosh{\theta}} - 1 = \frac{g + h\sinh\theta - \cosh{\theta}}{\cosh{\theta}} \geqslant \frac{(g+h\sinh{\theta}) -d}{\cosh{\theta}}.
\end{align}
By contrast, when $\cosh{\theta} \geqslant d$, equivalently $\abs{\sinh{\theta}} \geqslant k$, Lemma~\ref{lem:f_bound} implies that
\begin{align}
    g+h\sinh\theta - d \leqslant \frac{\cosh{\theta} + 1}{d}.
\end{align}
Thus, we have
\begin{align}
    \frac{g+h\sinh\theta - d}{\cosh{\theta}} \leqslant \frac{\cosh{\theta}+1}{d\cosh{\theta}} \leqslant \frac{1}{d} + \frac{1}{d^2}  \leqslant 1.
\end{align}
Taken together, these observations show that the optimal value is ultimately bounded above by
\begin{align}\label{eq:NN_Case_Max}
    \max_{\theta}\left\{2f(\theta) - 1, 1\right\},
\end{align}
thereby completing the proof.
\end{proof}

It is worth noting that taking the maximum in Eq.~\eqref{eq:opt_max_case_1} is necessary, as $2f(\theta) - 1$ is not always greater than 1; equivalently, $2f(\theta)$ is not necessarily greater than 1. 
For instance, consider approximate virtual broadcasting with equal error parameters on the receivers' side, namely $\epsilon_1=\epsilon_2$. 
In this case, we have
\begin{align}
    h=\frac{d}{\sqrt{d^2-1}}(\epsilon_1-\epsilon_2)=0.
\end{align}
As a result, $f(\theta)$ simplifies to
\begin{align}
    f(\theta)\Big|_{h=0}=\frac{g}{d + \cosh{\theta}}
    \leqslant\frac{g}{d+1}
    =\frac{d(2-(\epsilon_1+\epsilon_2))}{d+1}.
\end{align}
It is then straightforward to see that when $\epsilon_1=\epsilon_2=1/2$, $f(\theta)$ is strictly smaller than 1, i.e., $f(\theta)=d/(d+1)<1$, and $\max \left\{2f(\theta) - 1, 1\right\}=1$.

We next consider the case where the coefficient in the objective function is non-positive, namely $g + h\sinh{\theta} \leqslant0$. 
In this regime, maximizing the objective function requires taking the minimal value of $x$. 
The resulting optimization problem with the additional constraint is summarized in the lemma below.

\begin{mylem}{Non-Positive Objective Coefficient}{upper_bound_other}
    
    For any fixed $\theta$, the optimal value of the following optimization problem equals 1; that is,
    \begin{align}
        1=\max \quad 
        &  (g + h\sinh{\theta}) x - z \\
        \text{s.t.} \quad 
        & 1 \geqslant|z|,\\
        & 1 - \abs{x} \cosh{\theta} \geqslant \abs{xd - z} \\
        & g + h\sinh{\theta} \leqslant0 \ .
    \end{align}
\end{mylem}

\begin{proof}
Since $(g+h\sinh\theta \leqslant 0)$, the maximization over $x$ is achieved at its smallest admissible value for each fixed $z$.
Lemma~\ref{lem:feasible} shows that this extremal value varies with $z$, leading to two regimes.
We first examine $z \geqslant -d/\cosh{\theta}$, for which the objective function $\zeta(\theta, x, z)$ (see Eq.~\eqref{eq:zeta}) reduces to
\begin{align}
    \zeta(\theta, x, z)\Big|_{x=\frac{z-1}{d+\cosh{\theta}}}
    =
    \left((g+ \sinh\theta) x - z\right)
    \Big|_{x=\frac{z-1}{d+\cosh{\theta}}}
    =
    (g+ \sinh\theta) \frac{z-1}{d+\cosh{\theta}} - z.
\end{align}
Because the objective function is linear in $z$, the maximum occurs at the boundary points $\pm1$ or $-d/\cosh{\theta}$ (see Lem.~\ref{lem:feasible}).
However, the assumption $g+h\sinh{\theta} \leqslant 0$ implies $g\leqslant |h\sinh{\theta}|$.
Together with the bound $g \geqslant k|h|$ (Lem.~\ref{lem:valid_bound}), this yields $|\sinh\theta| \geqslant k$, which is equivalent to $\cosh{\theta} \geqslant \sqrt{k^2-1} = d$.
Consequently, the point $z= -1 \leqslant -d/\cosh\theta$ does not satisfy the admissibility conditions and must be excluded.
Evaluating the expression at the remaining boundary points $z=1$ and $z = -d/\cosh\theta$ gives
\begin{align}
    \zeta(\theta, \frac{z-1}{d+\cosh{\theta}}, 1)&=-1,\\
    \zeta(\theta, \frac{z-1}{d+\cosh{\theta}}, -\frac{d}{\cosh{\theta}})&=
    -\frac{g+h\sinh\theta}{d + \cosh{\theta}} \qty(\frac{d}{\cosh{\theta}}+1) + \frac{d}{\cosh{\theta}} = \frac{d - g - h\sinh\theta}{\cosh{\theta}}.
\end{align}
Thus, the maximal value in this case is
\begin{align}
    \max_{\theta}\left\{\frac{d - (g+h\sinh\theta)}{\cosh{\theta}}, -1\right\}.
\end{align}

In contrast, when $-1 \leqslant z \leqslant -d/\cosh\theta$, the objective function $\zeta(\theta, x, z)$ becomes
\begin{align}
    \zeta(\theta, x, z)\Big|_{x=\frac{z+1}{d-\cosh{\theta}}}
    =
    \left((g+ \sinh\theta) x - z\right)
    \Big|_{x=\frac{z+1}{d-\cosh{\theta}}}
    =
    (g+ \sinh\theta) \frac{z+1}{d-\cosh{\theta}} - z.
\end{align}
The values of the objective function at the boundary points $z= -1$ and $z = -d/\cosh\theta$ are
\begin{align}
    \zeta(\theta, \frac{z+1}{d-\cosh{\theta}}, -1)&=1,\\
    \zeta(\theta, \frac{z+1}{d-\cosh{\theta}}, -\frac{d}{\cosh{\theta}})&=
    \frac{g+h\sinh\theta}{d - \cosh{\theta}} \qty(\frac{d}{-\cosh{\theta}}-1) + \frac{d}{\cosh{\theta}} = \frac{d - g - h\sinh\theta}{\cosh{\theta}}.
\end{align}
Therefore, the maximum value in this case is
\begin{align}
    \max_{\theta}\left\{\frac{d - (g+h\sinh\theta)}{\cosh{\theta}}, 1\right\}.
\end{align}

To complete the proof, it suffices to show $(d - (g+h\sinh\theta))/\cosh{\theta}\leqslant1$.
The condition $g+h\sinh{\theta} \leqslant 0$ implies that $|\sinh{\theta}| \geqslant k$, and Lem.~\ref{lem:f_bound} then gives
\begin{align}
    d - g - h\sinh\theta \leqslant \frac{\cosh{\theta} - 1}{d}.
\end{align}
Consequently, it follows that
\begin{align}
    \frac{d - g - h\sinh\theta}{\cosh{\theta}} \leqslant \frac{\cosh{\theta} - 1}{d\cosh{\theta}} \leqslant \frac{1}{d} - \frac{1}{d\cosh{\theta}} \leqslant 1
\end{align}
This establishes the desired bound and completes the proof.
\end{proof}

Lemmas~\ref{lem:upper_bound_opt} and~\ref{lem:upper_bound_other} together yield a universal result for the optimization problem in Eq.~\eqref{ln:hyperbolic}, irrespective of the sign of $g+ \sinh\theta$, as stated in the following theorem.

\begin{mythm}{Universal Upper Bound}{general_solution}
The optimal value of $v_2^{\text{Dual}}(p_1, p_2)$ defined in Eq.~\eqref{ln:hyperbolic} takes the form
\begin{align}
    v_2^{\text{Dual}}(p_1, p_2) = \max_{\theta} \left\{ 2f(\theta) - 1, 1\right\},
\end{align}
where $f(\theta)$ is specified in Eq.~\eqref{eq:def_f}.
\end{mythm}

The optimization problem therefore reduces to maximizing $f(\theta)$, which is established in the following lemma.

\begin{mylem}{Maximum of $f(\theta)$}{landscape}
The function $f(\theta)$ defined in Eq.~\eqref{eq:def_f} attains a unique maximum at $\theta^*$, determined by
    \begin{align}\label{eq:e_theta}
    e^{\theta^*}
    &= \frac{h + \sqrt{g^2 - k^2h^2}}{g-dh} \ .
    \end{align}
The corresponding landscape of $f(\theta)$ is shown in Fig.~\ref{fig:landscape}.
\end{mylem}

\begin{figure}[t]
    \centering
    \includegraphics[width=0.7\linewidth]{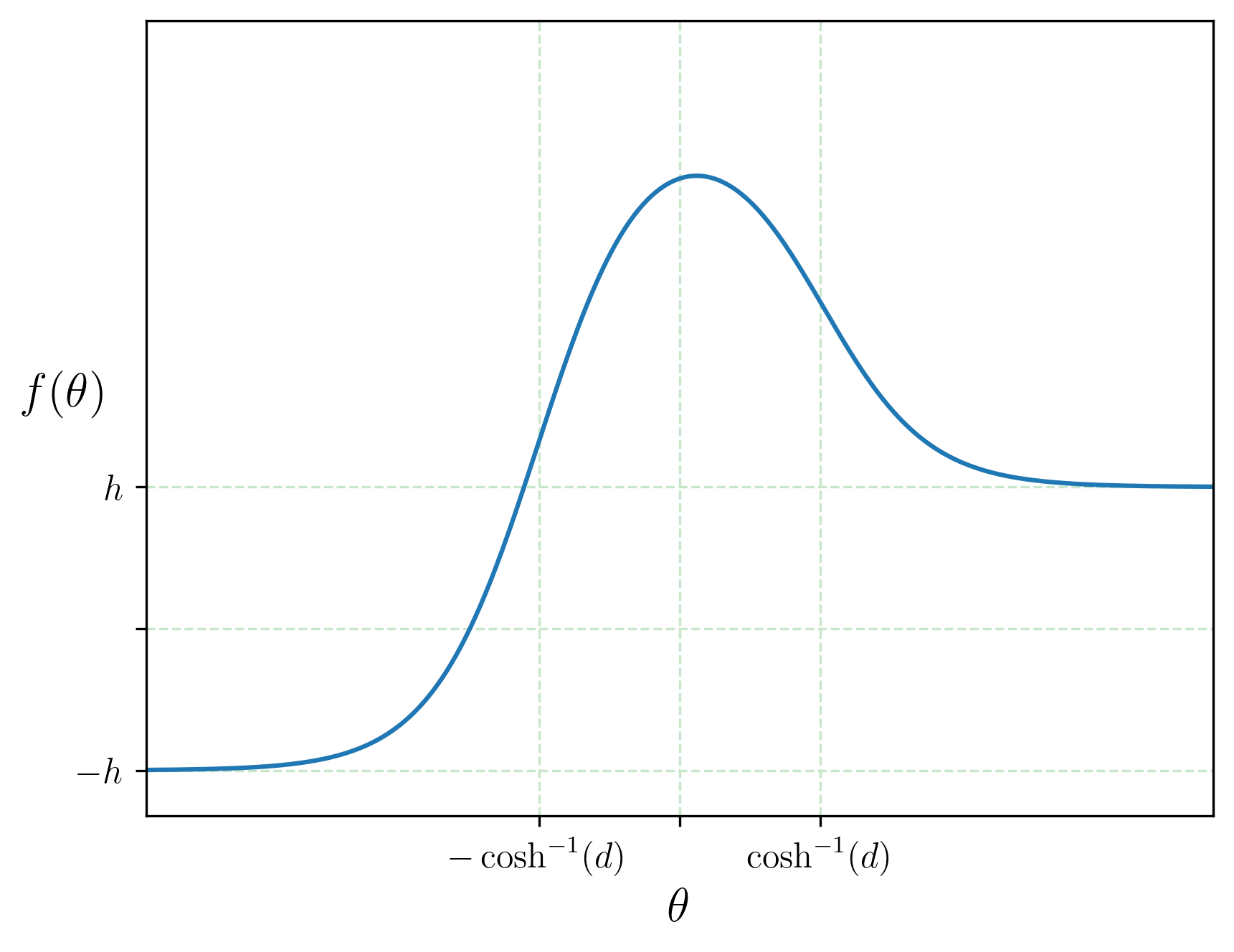}
    \caption{(Color online) \textbf{Plot of $f(\theta)$}.
    The figure plots $f(\theta)$ defined in Eq.~\eqref{eq:def_f} as a function of $\theta$. 
    The maximum always occurs in the regime $\cosh{\theta}\leqslant d$.}
    \label{fig:landscape}
\end{figure}

\begin{proof}
To locate the maximum, we analyze the derivative of $f(\theta)$

\begin{align}
    \dv{}{\theta}f(\theta)
    &=\dv{}{\theta}\left(\frac{g+h\sinh\theta}{d + \cosh{\theta}}\right)\\
    &= \frac{h \cosh{\theta}}{d+\cosh{\theta}} - \frac{g \sinh{\theta} + 2h \sinh^2{\theta}}{(d+\cosh{\theta})^2} \\
    &= \frac{dh \cosh{\theta} - g \sinh{\theta} + h \cosh^2{\theta} - h \sinh^2{\theta}}{(d+\cosh{\theta})^2} \\
    &= \frac{dh \cosh{\theta}-g \sinh{\theta} + h}{(d+\cosh{\theta})^2} \\
    &= \frac{(dh-g) e^{\theta} + (dh+g)e^{-\theta} + 2h}{2(d+\cosh{\theta})^2} \ .
\end{align}

Since the denominator of the derivative is always positive, the sign of the derivative is determined entirely by the numerator. 
The critical points are therefore obtained from the condition
\begin{align}
    (dh-g) e^{\theta} + (dh+g)e^{-\theta} + 2h = 0,
\end{align}
which can be written equivalently as
\begin{align}
    (g-dh) e^{2\theta} - 2he^{\theta} - (g + dh) = 0.
\end{align}

From Lem.~\ref{lem:valid_bound} we have $g - dh> 0$, so the equation is quadratic in $e^{\theta}$ and admits the solutions 
\begin{align}
    e^{\theta}
    = \frac{2h \pm \sqrt{4 h^2 + 4 g^2 - 4d^2 h^2}}{2(g-dh)} 
    = \frac{h \pm \sqrt{g^2 - (d^2-1) h^2}}{g-dh}
    = \frac{h \pm \sqrt{g^2 - k^2h^2}}{g-dh}.
\end{align}

Moreover, since 
\begin{equation}
    g^2 - k^2h^2 \geqslant g^2 - d^2h^2 \geqslant 0,
\end{equation}
it follows that
\begin{equation}
    g^2 - k^2h^2 \geqslant h^2.
\end{equation}
Consequently,
\begin{equation}
    \frac{h - \sqrt{g^2 - k^2h^2}}{g-dh} < 0 \ .
\end{equation}
Because $e^\theta >0$, only one positive root is admissible. 
With a positive leading coefficient, the maximum occurs at the critical point:
\begin{align}
    e^{\theta^*}
    &= \frac{h + \sqrt{g^2 - k^2h^2}}{g-dh} \ .
\end{align}
This behavior is summarized in Fig.~\ref{fig:landscape}; the asymptotic limits follow directly from the properties of hyperbolic functions.
\begin{align}
    \lim_{\theta\to \infty}f(\theta) 
    = 
    h 
    = -\lim_{\theta\to -\infty}f(\theta).
\end{align}
\end{proof}

Building on Lem.~\ref{lem:landscape}, the maximum of $f(\theta)$ with respect to $\theta$ is characterized in the following lemma.

\begin{mylem}{Maximum of $f(\theta)$}{opt_point}
The function $f(\theta)$ defined in Eq.~\eqref{eq:def_f} achieves its maximum at $\theta=\theta^*$ given in Eq.~\eqref{eq:e_theta}.
The resulting maximal value is
    \begin{equation}
        f(\theta^*) =  \frac{d^2(t_1+t_2) - 2d \sqrt{t_1t_2}}{d^2 - 1},
    \end{equation}
where $t_1$ and $t_2$ are introduced in Eqs.~\eqref{eq:t_1} and~\eqref{eq:t_2}. 
\end{mylem}

\begin{proof}
From the definition of $f(\theta)$, determining the exact value of $\theta^*$ is unnecessary. 
It suffices to evaluate the hyperbolic functions $\sinh{\theta}$ and $\cosh{\theta}$. We first compute $\sinh{\theta^*}$, which is given by

    \begin{align}
        \sinh{\theta^*}
        &= \frac{e^\theta - e^{-\theta}}{2} \\
        &= \frac{1}{2}\qty(\frac{h + \sqrt{g^2 - k^2h^2}}{g-dh} - \frac{g-dh}{h + \sqrt{g^2 - k^2h^2}}) \\
        &= \frac{1}{2}\frac{
        h^2 + g^2 - k^2h^2 + 2h\sqrt{g^2- k^2 h^2}
        -g^2 - d^2h^2 + 2dgh
        }{(g-dh)(h + \sqrt{g^2 - k^2h^2})} \\
        &= \frac{1}{2}\frac{
        2h\sqrt{g^2- k^2 h^2} - 2k^2 h^2 + 2dgh
        }{(g-dh)(h + \sqrt{g^2 - k^2h^2})} \\
        &= h\frac{ (h + \sqrt{g^2- k^2 h^2}) + d(g - dh) }{(g-dh)(h + \sqrt{g^2 - k^2h^2})} \\
        &= h\qty(\frac{d}{h + \sqrt{g^2 - k^2h^2}} + \frac{1}{g-dh}) \\
        &= h\qty(\frac{d\sqrt{g^2 - k^2h^2} - dh}{g^2 - k^2h^2 - h^2} + \frac{g + dh}{g^2-d^2h^2}) \\
        &= h\qty(\frac{d\sqrt{g^2 - k^2h^2} - dh + g + dh}{g^2-d^2h^2}) \\
        &= h\qty(\frac{g + d\sqrt{g^2 - k^2h^2}}{g^2-d^2h^2}).
    \end{align}
    
    The corresponding expression for the hyperbolic cosine is
    \begin{align}
        \cosh{\theta^*}
        &= \frac{e^\theta + e^{-\theta}}{2} \\
        &= \frac{1}{2}\qty(\frac{h + \sqrt{g^2 - k^2h^2}}{g-dh} + \frac{g-dh}{h + \sqrt{g^2 - k^2h^2}}) \\
        &= \frac{1}{2}\frac{
        h^2 + g^2 - k^2h^2 + 2h\sqrt{g^2- k^2 h^2}
        +g^2 + d^2h^2 - 2dgh
        }{(g-dh)(h + \sqrt{g^2 - k^2h^2})} \\
        &= \frac{1}{2}\frac{
        2h\sqrt{g^2- k^2 h^2} + 2g^2 + 2h^2 - 2dgh
        }{(g-dh)(h + \sqrt{g^2 - k^2h^2})} \\
        &= \frac{ h\sqrt{g^2- k^2 h^2} + g^2 + h^2 - dgh }{(g-dh)(h + \sqrt{g^2 - k^2h^2})} \\
        &= \frac{ h(\sqrt{g^2- k^2 h^2}+h) + g(g - dh)}{(g-dh)(h + \sqrt{g^2 - k^2h^2})} \\
        &= \frac{h}{g-dh} + \frac{g}{\sqrt{g^2- k^2 h^2}+h} \\
        &= \frac{gh + dh^2}{g^2-d^2h^2} + \frac{g\sqrt{g^2 - k^2h^2} - gh}{g^2- d^2 h^2} \\
        &= \frac{dh^2 + g\sqrt{g^2 - k^2h^2}}{g^2- d^2 h^2}.
    \end{align}

    Inserting $g=d(t_1 + t_2)$ (see Eq.~\eqref{eq:g_explicit}) and $h =d(t_2 - t_1)/k$ (see Eq.~\eqref{eq:h_explicit}) into the expressions for $\sinh{\theta^*}$ and $\cosh{\theta^*}$ yields
    
    \begin{align}
        \sinh{\theta^*}
        &= h\qty(\frac{g + d\sqrt{g^2 - k^2h^2}}{g^2-d^2h^2}) \\
        &= \frac{d}{k}(t_2-t_1) \qty(\frac{d(t_1+t_2) + 2d^2\sqrt{t_1 t_2}}{d^2(t_1+t_2)^2 - \frac{d^4}{k^2}(t_2-t_1)^2}) \\
        &= (t_2-t_1) \qty(\frac{(t_1+t_2) + 2d\sqrt{t_1 t_2}}{k(t_1+t_2)^2 - \frac{d^2}{k}(t_2-t_1)^2}) \\
        &= k (t_2-t_1) \qty(\frac{2d\sqrt{t_1 t_2} + (t_1+t_2)}{4d^2t_1 t_2 - (t_2+t_1)^2}) \\
        &= \frac{k (t_2-t_1) }{2d \sqrt{t_1 t_2} - (t_2+t_1)},
    \end{align}

    and
    
    \begin{align}
        \cosh{\theta^*}
        &= \frac{dh^2 + g\sqrt{g^2 - k^2h^2}}{g^2- d^2 h^2} \\
        &= \frac{\frac{d^3}{k^2}(t_2-t_1)^2 + 2d^2 (t_1+t_2) \sqrt{t_1 t_2}}{d^2(t_1+t_2)^2 - \frac{d^4}{k^2}(t_2-t_1)^2} \\
        &= \frac{d(t_2-t_1)^2 + 2k^2 (t_1+t_2) \sqrt{t_1 t_2}}{k^2(t_1+t_2)^2 - d^2(t_2-t_1)^2} \\
        &= \frac{d(t_2-t_1)^2 + 2k^2 (t_1+t_2) \sqrt{t_1 t_2}}{4d^2t_1 t_2 - (t_2+t_1)^2} \ .
    \end{align}

    Finally, substituting the expressions for $\sinh{\theta^*}$ and $\cosh{\theta^*}$ back into $f(\theta^*)$, one readily finds that
        \begin{align}
        f(\theta^*)
        &= \frac{g+h\sinh{\theta^*}}{d+ \cosh{\theta^*}} \\
        &= \frac{d(t_1+t_2)+d (t_2-t_1)^2 \frac{2d\sqrt{t_1 t_2} + (t_1+t_2)}{4d^2t_1 t_2 - (t_2+t_1)^2}}{d+ \frac{d(t_2-t_1)^2 + 2k^2 (t_1+t_2) \sqrt{t_1 t_2}}{4d^2t_1 t_2 - (t_2+t_1)^2}} \\
        &= \frac{d(t_1+t_2)+d (t_2-t_1)^2 \frac{2d\sqrt{t_1 t_2} + (t_1+t_2)}{4d^2t_1 t_2 - (t_2+t_1)^2}}{\frac{4d^3 t_1 t_2 - 4 d t_1 t_2 + 2k^2 (t_1+t_2) \sqrt{t_1 t_2}}{4d^2t_1 t_2 - (t_2+t_1)^2}} \\
        &= \frac{d(t_1+t_2)\qty(4d^2t_1 t_2 - (t_2+t_1)^2)+d (t_2-t_1)^2 \qty(2d\sqrt{t_1 t_2} + (t_1+t_2))}{4dk^2 t_1 t_2 + 2k^2 (t_1+t_2) \sqrt{t_1 t_2}} \\
        &= \frac{d(t_1+t_2)\qty(4d^2t_1 t_2 - (t_2+t_1)^2 + (t_2-t_1)^2)+2d^2 (t_2-t_1)^2 \sqrt{t_1 t_2}}{4dk^2 t_1 t_2 + 2k^2 (t_1+t_2) \sqrt{t_1 t_2}} \\
        &= \frac{4dk^2(t_1+t_2) t_1 t_2 +2d^2 (t_2-t_1)^2 \sqrt{t_1 t_2}}{4dk^2 t_1 t_2 + 2k^2 (t_1+t_2) \sqrt{t_1 t_2}} \\
        &= \frac{4dk^2(t_1+t_2) t_1 t_2 +2k^2 (t_2-t_1)^2 \sqrt{t_1 t_2} + 2(t_2-t_1)^2 \sqrt{t_1 t_2}}{4dk^2 t_1 t_2 + 2k^2 (t_1+t_2) \sqrt{t_1 t_2}} \\
        &= \frac{4dk^2(t_1+t_2) t_1 t_2 +2k^2 (t_2+t_1)^2 \sqrt{t_1 t_2} - 8 k^2t_1t_2\sqrt{t_1 t_2} + 2(t_2-t_1)^2 \sqrt{t_1 t_2}}{4dk^2 t_1 t_2 + 2k^2 (t_1+t_2) \sqrt{t_1 t_2}} \\
        &= (t_1 +t_2) + \frac{\sqrt{t_1 t_2}}{k^2} \frac{(t_2-t_1)^2 - 4 k^2t_1t_2}{2d t_1 t_2 + (t_1+t_2) \sqrt{t_1 t_2}} \\
        &= (t_1 +t_2) + \frac{1}{k^2} \frac{(t_2+t_1)^2 - 4 d^2t_1t_2}{2d \sqrt{t_1 t_2} + (t_1+t_2)} \\
        &= (t_1 +t_2) + \frac{1}{k^2} \frac{\qty((t_2+t_1) - 2 d\sqrt{t_1t_2})\qty((t_1+t_2) + 2 d\sqrt{t_1t_2})}{2d \sqrt{t_1 t_2} + (t_1+t_2)} \\
        &= (t_1 +t_2) + \frac{(t_1+t_2) - 2 d\sqrt{t_1t_2}}{k^2} \\
        &= \frac{d^2(t_1+t_2) - 2 d\sqrt{t_1t_2}}{k^2} \\
        &= \frac{d^2(t_1+t_2) - 2d \sqrt{t_1t_2}}{d^2 - 1} \\
        &= \frac{d(t_1+t_2)}{d+1} + \frac{d (\sqrt{t_1} - \sqrt{t_2})^2}{d^2 - 1}.
    \end{align}
    Thus the claim follows.
\end{proof}

It is straightforward to verify that strong duality holds for the optimization problem considered here. 
Consequently, the dual formulation in Eq.~\eqref{eq:AVB_SOCP} is equivalent to the primal one in Eq.~\eqref{eq:SDP_Approximate}, and the optimal value of the sample complexity overhead is characterized by the following theorem.

\begin{mythm}{Analytical Solution for Approximate Virtual Broadcasting}{Analytical_Solution}
For the approximate virtual broadcasting task defined in Eq.~\eqref{eq:SDP_Approximate}, with receiver errors $\epsilon_1$ and $\epsilon_2$, the the minimal sample complexity overhead takes the form
\begin{align}
    \max_{\theta} \left\{ 2f(\theta) - 1, 1\right\}
    =
    \max\left\{\frac{d(3-2(\epsilon_1+\epsilon_2))-1}{d+1} + \frac{2d\qty(\sqrt{1-\epsilon_1} - \sqrt{1-\epsilon_2})^2}{d^2-1},1\right\}.
\end{align}
\end{mythm}

The general result in Thm.~\ref{thm:Analytical_Solution} also contains exact virtual broadcasting (see Lem.~\ref{lem:SC_EVB}) as a special case, recovering the result of Ref.~\cite{PhysRevA.110.012458}. 
In particular, when the marginal is perfect, i.e., $\epsilon_1=\epsilon_2=0$, the expression $2f(\theta^*) - 1$ simplifies to
\begin{align}
    2f(\theta^*) - 1
    % =\frac{3d^2-4d+1}{d^2-1}
    % =
    % \frac{(3d-1)(d-1)}{(d+1)(d-1)}
    =
    \frac{3d-1}{d+1}.
\end{align}
The quantity $(3d-1)/(d+1)$ is always greater than $1$, since the dimension $d$ of a quantum system is an integer with $d\geqslant2$. 
Thus, in this case, we have
\begin{align}
    \max_{\theta} \left\{ 2f(\theta) - 1, 1\right\}
    =
    \max\left\{\frac{3d-1}{d+1},1\right\}
    =
    \frac{3d-1}{d+1}.
\end{align}
When the errors on the receivers' sides are required to be identical, i.e., $\epsilon_1=\epsilon_2=\epsilon$, the sample complexity $u_2(\epsilon_1, \epsilon_2)$ (see Eq.~\eqref{eq:SDP_Approximate}) simplifies to
\begin{align}
    u_2(\epsilon, \epsilon)=\max\left\{\frac{d(3-4\epsilon)-1}{d+1}, 1\right\}.
\end{align}
Whenever the first term inside the maximum exceeds 1, i.e., $\epsilon\leqslant(d-1)/2d$, the resulting sample complexity depends linearly on the error parameter $\epsilon$.

Before concluding this subsection, we briefly summarize the derivation of the final analytical solution to the optimization problem, as presented in the following chain of equations.
\begin{align}
    u_2(\epsilon_1, \epsilon_2)
    \xeq[\text{Subsec.~\ref{subsec:Marginals}}]{\text{Thm.~\ref{thm:Depolarizing_Marginal}}}
    v_2(p_1, p_2)
    \xeq[\text{Subsec.~\ref{subsec:SDP_Dual}}]{\text{Lem.~\ref{lem:Dual_Formulation_AVB_S_1}}}
    v_2^{\text{Dual}}(p_1, p_2)
    \xeq[\text{Subsec.~\ref{subsec:Analytic_Solution}}]{\text{Thm.~\ref{thm:Analytical_Solution}}}
    \max\left\{\frac{d^2(3-2(\epsilon_1+\epsilon_2))-4d\sqrt{(1-\epsilon_1)(1-\epsilon_2)}+1}{d^2-1},1\right\}.
\end{align}
Here $u_2(\epsilon_1, \epsilon_2)$ (see Eq.~\eqref{eq:SDP_Approximate}) denotes the primal formulation of the sample complexity overhead for approximate virtual broadcasting. 
The quantity $v_2(p_1, p_2)$ (see Eq.~\eqref{eq:SDP_Approximate_v}) is a simplified form obtained by exploiting triple-twirling invariance. 
The term $v_2^{\text{Dual}}(p_1, p_2)$ represents the corresponding simplified dual formulation derived from the same symmetry. 
The final analytical solution follows by transforming the original dual SDP into a SOCP (see Eq.~\eqref{eq:AVB_SOCP}), which admits an explicit solution. 
The parameters $p_i$ and $\epsilon_i$ ($i\in\{1,2\}$) are related through Eq.~\eqref{eq:p_epsilon_i}.

%%%%%%%%%%%%%%%%%%%%%%%%%%%%%%%%%%%%%%%%%%%%%%%%%%%%%%%%%%%%%%%%%%%%%%%%%%%%%%%%%%%%%%%%%%%%%%%%%%%%%%%%%%%%%%%%%%%%%%

\subsection{Sample Efficiency in Approximate Broadcasting}
\label{subsec:SE_AB}

In many areas of physics, the efficiency of a physical process is not determined by its absolute performance alone, but by how much useful output is generated per unit of consumed resource. 
This idea is often captured at a conceptual level by the relation

\begin{align}\label{eq:Concept_Efficiency}
    \text{Efficiency}\sim \frac{\text{Performance}}{\text{Consumed Resource}},
\end{align}
which emphasizes that performance must be assessed relative to the cost required to achieve it. 
A canonical example is power, defined as energy produced per unit time: 
increasing the total energy output without a corresponding reduction in the required time does not constitute an efficiency gain. 
Closely related rate-like quantities appear across physics, including the entropy production rate in nonequilibrium thermodynamics and the information rate in Shannon theory, which quantifies the amount of information transmitted per channel use.
Such quantities provide an intrinsic, resource-normalized characterization of performance, independent of the overall scale at which the process is implemented.

The same consideration applies naturally to quantum broadcasting. 
One may ask how accurately a protocol can broadcast quantum information given a fixed number of input copies, but this absolute figure alone does not reveal how efficiently those copies are utilized. 
To see this, consider two approximate broadcasting strategies. 
Suppose that one protocol achieves a higher average fidelity, but only by consuming significantly more input copies, while another attains a lower fidelity using far fewer resources. 
Judged solely by average broadcasting fidelity, the comparison is ambiguous;
judged by output per copy, their relative merit becomes clear.
These observations suggest that quantum broadcasting should be evaluated through a rate-based lens. 
Rather than focusing on average broadcasting fidelity, it is more informative to quantify how much fidelity is generated per input copy. 
This perspective leads naturally to the notion of a broadcasting rate, which we introduce below and adopt as the central figure of merit for benchmarking the efficiency of approximate quantum broadcasting protocols.

\begin{mydef}
{Broadcasting Rate}
{Broadcasting_Rate}
For an approximate broadcasting protocol $\mE$ characterized by parameters $\epsilon_1$ and $\epsilon_2$ on Bob's and Claire's sides, respectively (see Subsec.~\ref{subsec:AVB}), let $F_{\text{Bcast}}$ denote the resulting average broadcasting fidelity, and let $n(\epsilon_1, \epsilon_2)$ be the number of input samples required to pass the $\epsilon-\delta$ test. 
The broadcasting rate is then defined as 
\begin{align}\label{eq:BC_Rate}
    R(\epsilon_1, \epsilon_2):=\frac{F_{\text{Bcast}}(\mE)}{n(\epsilon_1, \epsilon_2)}.
\end{align}
\end{mydef}

Here, the task is to estimate the expectation values of 
\begin{align}
    \Tr[\id\otimes(\Tr_{C}\circ\,\mE_{A\to BC})(\phi^{+})\cdot\phi^{+}],
\end{align}
and
\begin{align}
    \Tr[\id\otimes(\Tr_{B}\circ\,\mE_{A\to BC})(\phi^{+})\cdot\phi^{+}],
\end{align}
which correspond to the marginal broadcasting fidelities on the receivers' sides (see Subsec.~\ref{subsec:Broadcasting_Fidelity}). 

It is worth noting that the broadcasting fidelity (see Def.~\ref{def:BF_B}) can be interpreted as the expectation value of an observable, in this case the projector onto the maximally entangled state. 
This interpretation allows us to employ Hoeffding's inequality (see Lem.~\ref{lem:Hoeffding}) to quantify the number of samples required to achieve a desired level of statistical confidence. 
By contrast, we do not adopt the diamond norm as the figure of merit.
The reason is that the diamond norm does not naturally admit an interpretation as the expectation value of experimentally accessible data, and therefore the corresponding sample complexity overhead cannot be straightforwardly characterized within our framework.

Without loss of generality, assume that Bob and Claire require $n_1$ and $n_2$ copies of the maximally entangled state $\phi^{+}$ (see Eq.~\eqref{eq:MES}) to pass their respective $\epsilon-\delta$ test (see Def.~\ref{def:epsilon-delta}) when operating under physical quantum channels.
Denoting the corresponding sample requirement by $n_Q$, 
\begin{align}
    n_Q:=\max\{n_1, n_2\}.
\end{align}
An approximate broadcasting protocol implemented via virtual operations requires instead $u^2_2(\epsilon_1, \epsilon_2) n_Q$ input copies. 
The prefactor $u^2_2(\epsilon_1, \epsilon_2)$ quantifies the sampling overhead introduced by the virtual protocol, as discussed in Cor.~\ref{cor:Virtual Sample Cost}.
For the approximate virtual broadcasting protocol analyzed in Eq.~\eqref{eq:SDP_Approximate}, the marginal channels are known to reduce to depolarizing channels $\mD_{p}$ (see Subsec.~\ref{subsec:Marginals}). 
This observation allows the numerator of Eq.~\eqref{eq:BC_Rate} to be evaluated as

\begin{align}
    F_{\text{Bcast}}(\mE)
    =
    &1+\frac{1-d^2}{2d^2}(p_1+p_2)\\
    =
    &1-\frac{\epsilon_1+\epsilon_2}{2},
\end{align}
where the second equality follows directly from Eq.~\eqref{eq:p_epsilon_i} of Thm.~\ref{thm:Depolarizing_Marginal}.
In this case, one finds
\begin{align}\label{eq:Broadcasting_Rate_2}
    R(\epsilon_1, \epsilon_2)=\frac{1-\frac{\epsilon_1+\epsilon_2}{2}}{u^2_2(\epsilon_1, \epsilon_2) n_Q}.
\end{align}

Alternatively, the task can be accomplished trivially by distributing $n_Q$ copies to Bob and another $n_Q$ copies to Claire, yielding perfect performance on both sides and the corresponding broadcasting rate
\begin{align}
    R_{\text{naive}}=\frac{\frac{1}{2}(1+1)}{n_Q+n_Q}=\frac{1}{2n_Q}.
\end{align}

Accordingly, approximate broadcasting is sample efficient (SE) if and only if there exists a region of parameters $(\epsilon_1,\epsilon_2)$ such that 
\begin{align}
    R(\epsilon_1, \epsilon_2)\geqslant
    R_{\text{naive}}.
\end{align}
Because $n_Q>0$, this condition is equivalent to
\begin{align}
    \frac{1-\frac{\epsilon_1+\epsilon_2}{2}}{u^2_2(\epsilon_1, \epsilon_2)}
    \geqslant
    \frac{1}{2}.
\end{align}
We now state the following theorem concerning the practicality of approximate virtual broadcasting.

\begin{mythm}
{Practical Approximate Virtual Broadcasting}
{Practical_AVB}
An approximate broadcasting protocol, characterized by error parameters $\epsilon_1$ and $\epsilon_2$ on Bob's and Claire's sides, is sample efficient (SE) only if $\epsilon_1$ and $\epsilon_2$ satisfy
\begin{align}\label{eq:1-to-2_PAVB}
    u_2(\epsilon_1, \epsilon_2)\leqslant\sqrt{2-(\epsilon_1+\epsilon_2)},
\end{align}
where $u_2(\epsilon_1, \epsilon_2)$ is given by Thm.~\ref{thm:Analytical_Solution}.
\end{mythm}

Numerical results for 1-to-2 approximate virtual broadcasting are shown in Fig.~\ref{fig:1to2_dims}, illustrating both the sample complexity overhead and the regimes in which Eq.~\eqref{eq:1-to-2_PAVB} is satisfied, namely, where the broadcasting rate $R(\epsilon_1, \epsilon_2)$ (see Eq.~\ref{eq:BC_Rate}) exceeds that of the naive protocol.

\begin{figure}
    \centering
    \includegraphics[width=0.9\linewidth]{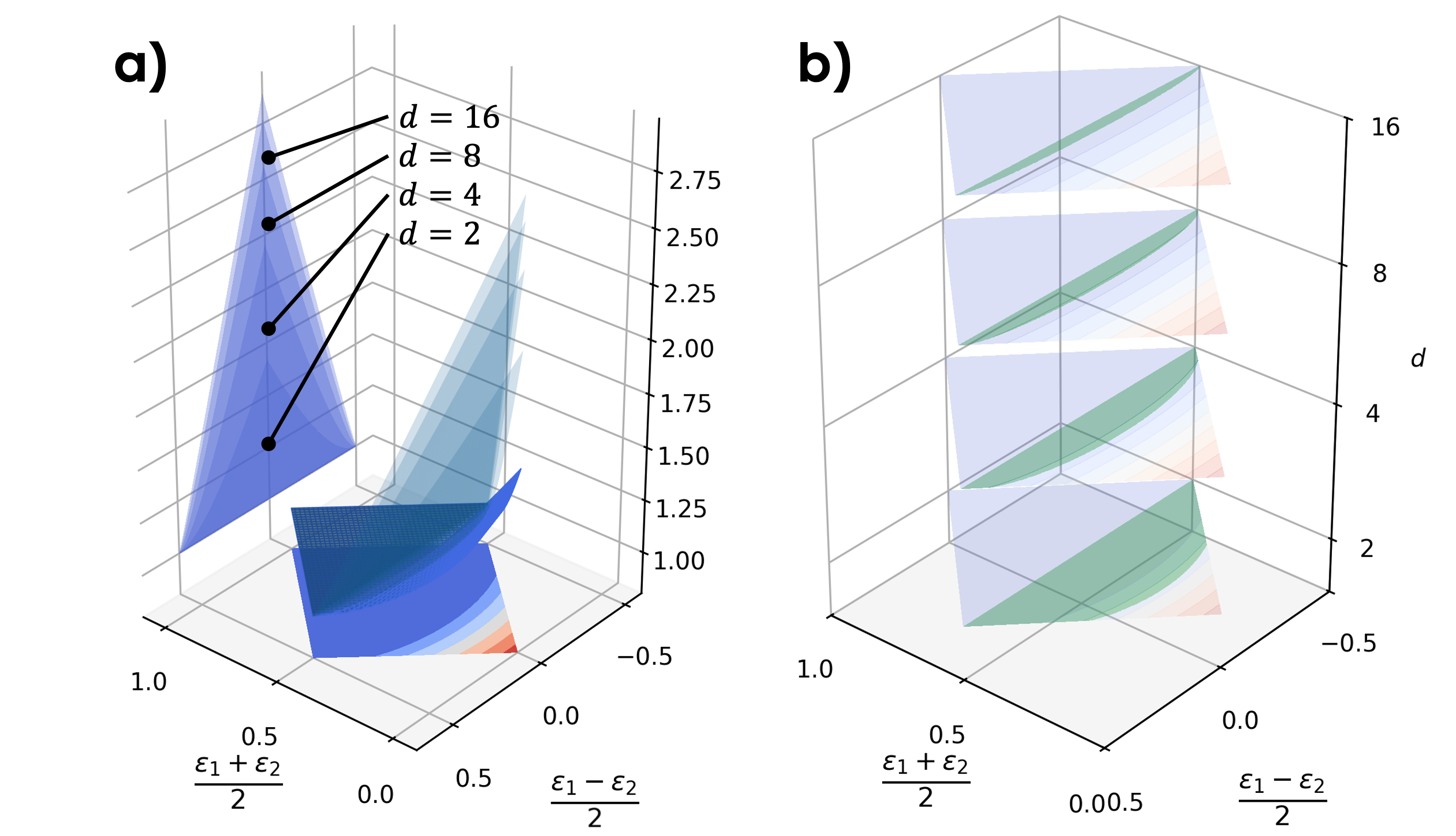}
    \caption{(Color online) \textbf{Sample Complexity Landscape for Approximate Virtual Broadcasting}.
    For 1-to-2 approximate virtual broadcasting, the sample complexity overhead $u_2(\epsilon_1, \epsilon_2)$ (see Eq.~\eqref{eq:SDP_Approximate}) is shown in (a). The parameter region in which Eq.~\eqref{eq:1-to-2_PAVB} is satisfied is highlighted in (b) (green), indicating where the broadcasting rate surpasses that of the naive prepare-and-distribute protocol.
    }
    \label{fig:1to2_dims}
\end{figure}

%%%%%%%%%%%%%%%%%%%%%%%%%%%%%%%%%%%%%%%%%%%%%%%%%%%%%%%%%%%%%%%%%%%%%%%%%%%%%%%%%%%%%%%%%%%%%%%%%%%%%%%%%%%%%%%%%%%%%%

\subsection{Learning Quantum State Properties via Approximate Broadcasting}
\label{subsec:Learning_AQB}

Assume that multiple copies of a quantum state $\rho$ are available to Alice, who subsequently distributes them to the receivers Bob and Claire via an approximate virtual broadcasting map. 
Consider the scenario in which Bob uses the state he receives to estimate the expectation value of an observable $\mO$. 
In the standard setting, $n=\mO\left((c^2\ln{1/\delta}\right)/\epsilon^2)$ suffices to determine the expectation value to accuracy $\epsilon$ with success probability at least $1-\delta$ (see Cor.~\ref{cor:Sample Cost}), giving rise to the familiar $\epsilon-\delta$ test (see Def.~\ref{def:epsilon-delta}). 
A natural question is how this guarantee is modified when the available copies are first processed through the approximate virtual broadcasting protocol that maximizes the performance characterized in Eq.~\eqref{eq:SDP_Approximate}, as introduced in Subsec.~\ref{subsec:AVB}. 
This subsection analyzes this scenario and derives the corresponding $\epsilon-\delta$ test behavior in the presence of approximate virtual broadcasting.

As discussed in Subsecs.~\ref{subsec:AVB},~\ref{subsec:Symmetry_Simplifications} and~\ref{subsec:Marginals}, the approximate broadcasting map can, without loss of generality, be taken to be $\mE$ with the decomposition of Eq.~\eqref{eq:Virtual_Decomposition}, i.e., $\mE=a\mE_1-b\mE_2$, where the channels $\mE_1$ and $\mE_2$ arise from the triple-twirling $\mT$ (see Eq.~\eqref{eq:Twirling}). 
Consequently, when Bob estimates the observable $\mO$, let $\overline{X}$ denote the empirical average obtained from the measurement outcomes. 
Suppose that $\overline{X}$ satisfies the $\epsilon-\delta$ test for the quantity $\Tr[\Tr_{C}\circ\, \mE(\rho)\cdot\mO]$. 
The problem then is whether the same experimental data can be used to infer the true value of $\Tr[\rho\mO]$. 
To address this point, it is instructive to first analyze the deviation between the experimental estimate $\overline{X}$ and the expectation value $\Tr[\rho\mO]$, which is given by
\begin{align}
    \left|\overline{X}-\Tr[\rho\mO]\right|
    =
    &\left|\overline{X}-
    \Tr[\Tr_{C}\circ\, \mE(\rho)\cdot\mO]+
    \Tr[\Tr_{C}\circ\, \mE(\rho)\cdot\mO]-
    \Tr[\rho\mO]\right|\\
    \leqslant
    &\left|\overline{X}-
    \Tr[\Tr_{C}\circ\, \mE(\rho)\cdot\mO]\right|+
    \left|\Tr[\Tr_{C}\circ\, \mE(\rho)\cdot\mO]-
    \Tr[\rho\mO]\right|\label{eq:AVB_Statistic_2}\\
    =
    &\epsilon+\left|\frac{d^2\epsilon_1}{d(d^2-1)}\Tr[\mO]-\frac{d^2\epsilon_1}{d^2-1}\Tr[\rho\mO]\right|.\label{eq:AVB_Statistic_3}
\end{align}
The step from Eq.~\eqref{eq:AVB_Statistic_2} to Eq.~\eqref{eq:AVB_Statistic_3} follows from the fact that the marginals of approximate virtual broadcasting reduce to depolarizing channel, as established in Eq.~\eqref{eq:Approx_B_v} and Thm.~\ref{thm:Depolarizing_Marginal}; that is 
\begin{align}
    \Tr_{C}\circ\, \mE_{A\to BC}(\cdot)=\mD_{p_1}(\cdot)
    =
    p_1\left(\frac{1}{d}\1_{B}\right)\otimes\Tr_{A}[\cdot]
    +
    (1-p_1)\id_{A\to B}(\cdot).
\end{align}
Without loss of generality, the observable of interest may be taken to be traceless, as in the case of Pauli strings. 
Denoting by $o_{\max}$ the maximal eigenvalue of $\mO$, the above deviation can be written as
\begin{align}
    \left|\overline{X}-\Tr[\rho\mO]\right|
    \leqslant
    \epsilon+\frac{d^2\epsilon_1o_{\max}}{d^2-1}.
\end{align}
The resulting performance guarantee for estimating $\mO$ is stated in the following theorem.

\begin{mythm}
{Statistical Guarantees for Approximate Broadcasting}
{Statistical_AVB} 
For a traceless observable $\mO$ with maximal eigenvalue $o_{\max}$, consider approximate virtual broadcasting implemented by a virtual operation that distributes the input state from Alice to Bob and Claire. 
If Bob passes an $\epsilon-\delta$ test, then the same data can be used to estimate the exact value of $\Tr[\rho\mO]$ with an $(\epsilon+(d^2\epsilon_1o_{\max}/(d^2-1)))-\delta$ test, with the required number of copies given by
\begin{align}
    n=\mO\left(\frac{u_2(\epsilon_1, \epsilon_2)^2c^2}
    {\left(\epsilon+\frac{d^2\epsilon_1o_{\max}}{d^2-1}\right)^2}\ln{\frac{1}{\delta}}\right),
\end{align}
where $u_2(\epsilon_1, \epsilon_2)$ is defined in Eq.~\eqref{eq:SDP_Approximate}, and $c$, defined in Eq.~\eqref{eq:c}, denotes the maximal difference between possible measurement outcomes.
\end{mythm}

%%%%%%%%%%%%%%%%%%%%%%%%%%%%%%%%%%%%%%%%%%%%%%%%%%%%%%%%%%%%%%%%%%%%%%%%%%%%%%%%%%%%%%%%%%%%%%%%%%%%%%%%%%%%%%%%%%%%%%

\section{Probabilistic Quantum Broadcasting}
\label{sec:PQB}

While the previous analysis focused on deterministic virtual broadcasting in the exact and approximate regimes, these scenarios do not exhaust all possible strategies. 
We therefore turn to the possibility of probabilistic quantum broadcasting.
We first show in Subsec.~\ref{subsec:No_Probabilistic_QB} that broadcasting remains impossible under quantum operations alone, even when arbitrarily small success probabilities are permitted. 
We then extend the operational framework to probabilistic virtual operations, enabling probabilistic virtual broadcasting at the cost of an additional sample complexity overhead, which can be systematically characterized through a semidefinite programming, as formulated in Subsec.~\ref{subsec:PVB}. 
The corresponding analytic solution is presented in Subsec.~\ref{subsec:PVB_Analytical_Solution}. 
Finally, in Subsec.~\ref{subsec:SE_PB}, by combining this analytic result with a comparison to naive protocols, we establish a new no-go theorem for probabilistic virtual broadcasting.

%%%%%%%%%%%%%%%%%%%%%%%%%%%%%%%%%%%%%%%%%%%%%%%%%%%%%%%%%%%%%%%%%%%%%%%%%%%%%%%%%%%%%%%%%%%%%%%%%%%%%%%%%%%%%%%%%%%%%%

\subsection{No Probabilistic Quantum Broadcasting}
\label{subsec:No_Probabilistic_QB}

Deterministic broadcasting that exactly reproduces the initial state has been proven impossible, giving rise to the {\it no-broadcasting theorem} in quantum theory~\cite{PhysRevLett.76.2818}. 
Recent results further show that even when arbitrary linear maps are considered and sample complexity is taken into account, broadcasting quantum information remains impossible, a limitation referred to as {\it no practical quantum broadcasting}~\cite{z2pr-zbwl,8g6j-w7ld}. 
However, the possibility of constructing probabilistic broadcasting protocols has not yet been ruled out. 
In this subsection, this possibility is examined within standard quantum theory, restricting attention to quantum channels, that is, completely positive trace-preserving (CPTP) linear maps. 
The analysis ultimately shows that broadcasting cannot be achieved not only deterministically, but also probabilistically.

At first sight, probabilistic broadcasting protocols may appear highly intricate. 
For instance, when the sender Alice attempts to broadcast an initial state to the receivers Bob and Claire, one might imagine that the success probability depends on the input state. 
Locally, it could occur that different states are successfully delivered to Bob with different probabilities. 
Globally, it might also be conceivable that Bob and Claire receive the broadcasted state with unequal success probabilities.

The analysis below rules out all such possibilities. 
If broadcasting were to occur probabilistically, the success probability would necessarily have to be fixed across all systems and for all input states. 
The result shows that even under this probabilistic scenario, quantum broadcasting remains impossible: quantum information cannot be broadcast even with arbitrarily small success probability.

We begin by analyzing the local success probability of quantum broadcasting, as formalized in the following lemma.

\begin{mylem}
{Local Success Probability}
{Local_Success_Probability} 
For probabilistic broadcasting implemented by a quantum subchannel $\mE:A\to BC$, the success probability on each receiver's side is independent of the input state.
\end{mylem}

\begin{proof}
We sketch the proof for a pair of specific input states; the general case follows from the same argument. 
Without loss of generality, consider the situation in which Alice attempts to broadcast the states $\ket{0}$ and $\ket{1}$ using subchannel $\mE:A\to BC$, and Bob receives them locally with success probabilities $p_0$ and $p_1$, respectively. 
\begin{align}
    \Tr_{C}[\mE(\ketbra{0}{0})]&=p_0\ketbra{0}{0},\\
    \Tr_{C}[\mE(\ketbra{1}{1})]&=p_1\ketbra{1}{1}.
\end{align}
Applying $\mE$ to the input state $(\ket{0}+\ket{1})/2$ yields
\begin{align}
    \frac{p_0\ketbra{0}{0}+p_1\ketbra{1}{1}}{2}
    =\Tr_{C}\qty[\mE\left(\frac{\ketbra{0}{0}+\ketbra{1}{1}}{2}\right)]
    =\frac{p(\ketbra{0}{0}+\ketbra{1}{1})}{2},
\end{align}
with some probability $p$.
A comparison of the expressions immediately implies
\begin{align}
    p=p_0=p_1,
\end{align}
thereby completing the proof.
\end{proof}

We next turn to the global success probability of quantum broadcasting, stated in the lemma below.

\begin{mylem}
{Global Success Probability}
{Global_Success_Probability} 
For probabilistic broadcasting implemented by a quantum subchannel $\mE:A\to BC$, the success probabilities on the two receivers' sides must be identical.
\end{mylem}

\begin{proof}
Thanks to Lem.~\ref{lem:Local_Success_Probability}, the marginals on Bob's and Claire's sides can be assumed to take the form
\begin{align}
    \Tr_{C}[J^{\mE}]&=p_{\text{Bob}}\Gamma_{AB},\\
    \Tr_{B}[J^{\mE}]&=p_{\text{Claire}}\Gamma_{AC}.
\end{align}
Tracing out system $B$ in $\Tr_{C}[J^{\mE}]$ and system $C$ in $\Tr_{C}[J^{\mE}]$ gives
\begin{align}
    p_{\text{Bob}}\1_{A}=p_{\text{Claire}}\1_{A},
\end{align}
from which 
\begin{align}
    p_{\text{Bob}}=p_{\text{Claire}}
\end{align}
follows, thereby proving the statement.
\end{proof}

We are now ready to state the result for probabilistic broadcasting.

\begin{mythm}
{No Probabilistic Quantum Broadcasting}
{No_Probabilistic_Quantum_Broadcasting} 
Quantum information cannot be broadcast through quantum channels, even with arbitrarily small success probability.
\end{mythm}

\begin{proof}
Using Lems.~\ref{lem:Local_Success_Probability} and~\ref{lem:Global_Success_Probability}, one may assume that  
\begin{align}
    \Tr_{C}\circ\,\mE=\Tr_{B}\circ\,\mE=p\,\id,
\end{align}
for some nonzero success probability $p$.
This would imply the existence of a broadcasting channel $\mB$, 
\begin{align}
    \mB:=\frac{1}{p}\mE,
\end{align}
contradicting the no-broadcasting theorem~\cite{PhysRevLett.76.2818}.
Hence $p=0$, showing that quantum information cannot be broadcast via quantum channels, even probabilistically.
\end{proof}

%%%%%%%%%%%%%%%%%%%%%%%%%%%%%%%%%%%%%%%%%%%%%%%%%%%%%%%%%%%%%%%%%%%%%%%%%%%%%%%%%%%%%%%%%%%%%%%%%%%%%%%%%%%%%%%%%%%%%%

\subsection{Probabilistic Virtual Broadcasting}
\label{subsec:PVB}

As established in the previous subsection, quantum information cannot be broadcast using standard quantum operations, namely quantum channels (CPTP maps), even when probabilistic protocols are allowed. 
This raises the question of whether the situation changes if the operational framework is broadened. 
In particular, one may ask whether broadcasting becomes possible when the allowed transformations are extended from purely quantum channels to virtual operations, which combine quantum processes with classical post-processing. 
The present subsection addresses this question by examining the possibility of probabilistic virtual broadcasting. 

We begin by introducing a precise formulation of this task.
Assume that the probabilistic virtual broadcasting map $\mE$ is realized by combining two subchannels (CPTNI maps) so as to simulate a Hermiticity-preserving (HP) and trace-non-increasing (TNI) linear map. 
In particular, $\mE$ takes the form
\begin{align}
    \mE=a\mE_1-b\mE_2,
\end{align}
where $a$ and $b$ are non-negative real coefficients subject to the condition
\begin{align}
    a-b\leqslant1,
\end{align}
ensuring that the resulting virtual map $\mE$ remains trace non-increasing.
For a given success probability $p$, the minimal sample complexity overhead is characterized by the following semidefinite programming (SDP):

\begin{align}\label{eq:PVB_SDP}
    s_{2}(p)
    =
    \min \quad 
    & 
    a+b
    \\
    \text{s.t.} \quad 
    &\1_{C}\star (J_1-J_2)=p\Gamma_{AB},\\
    &\1_{B}\star (J_1-J_2)=p\Gamma_{AC},\\
    &J_1\geqslant0,\,\, \1_{BC}\star J_1\leqslant a\1_{A},\\
    &J_2\geqslant0,\,\, \1_{BC}\star J_2\leqslant b\1_{A},\\
    &a-b\leqslant1.
\end{align}
Remark that the proofs of Lems.~\ref{lem:Local_Success_Probability} and~\ref{lem:Global_Success_Probability} apply not only to subchannels but also to virtual subchannels. 
Therefore, in Eq.~\eqref{eq:PVB_SDP} we may assume a fixed success probability $p$ for both marginals.

%%%%%%%%%%%%%%%%%%%%%%%%%%%%%%%%%%%%%%%%%%%%%%%%%%%%%%%%%%%%%%%%%%%%%%%%%%%%%%%%%%%%%%%%%%%%%%%%%%%%%%%%%%%%%%%%%%%%%%

\subsection{Analytical Solution for Probabilistic Virtual Broadcasting}
\label{subsec:PVB_Analytical_Solution}

To find the exact analytical solution to the Semidefinite Programming (SDP) in Eq.~\eqref{eq:PVB_SDP}, we analyze its dual formulation. 
By exploiting the underlying unitary symmetries of the problem, the operator constraints can be reduced to a set of scalar conditions, which in turn leads to a much simpler linear programming.
The associated Lagrangian of this optimization problem is given by

\begin{align}
    \mathcal{L} :=& (a+b) + (a-b-1)c \\
    & - \Tr_{AC}\qty[(\Tr_B[J_1 - J_2]) - p \Gamma_{AC}) \cdot X_{AC}] \\
    & - \Tr_{AB}\qty[(\Tr_C[J_1 - J_2]) - p \Gamma_{AB}) \cdot Y_{AB}] \\
    &
    - \Tr_{ABC}[J_1\cdot M_{ABC}]
    - \Tr_{ABC}[J_2\cdot N_{ABC}] \\
    & + \Tr_{A}\qty[ \left(\Tr_{BC}[J_1]-a\1_{A} \right) \cdot P_A] +
    \Tr_{A}\qty[ \qty(\Tr_{BC}[J_2]-b\1_{A}) \cdot Q_A ],
\end{align}
where $X_{AC}, Y_{AB}$ are Hermitian operators, $M_{ABC}, N_{ABC}, P_A$ and $Q_A$ are positive semidefinite, and real number $c \geqslant 0$.
Optimizing the Lagrangian over the primal variable gives its dual form:

\begin{align}
    s_{2}^{\text{Dual}}(p)
    =
    \max \quad 
    & 
    p \Tr[X_{AC} \cdot \Gamma_{AC}] + p \Tr[Y_{AB} \cdot \Gamma_{AB}] - c 
    \\ \label{eq:Probabilistic_Dual_From}
    \text{s.t.} \quad 
    & \Tr[P_A] = 1 + c, \\
    & \Tr[Q_A] = 1 - c, \\
    & - Q_A \otimes \mathbb{1}_{BC} \leqslant Y_{AB} \otimes \mathbb{1}_C + \mathbb{1}_B \otimes X_{AC} \leqslant P_A \otimes \mathbb{1}_{BC}, \\
    & P_A \geqslant 0, \quad Q_A \geqslant 0, \\
    & c \geqslant 0.
\end{align}

\begin{mylem}
{Dual Formulation Triple-Twirling Invariance}
{Probabilistic_Dual_Formulation_TTI} 
The feasible set of the dual SDP $s_{2}^{\text{Dual}}(p)$ in Eq.~\eqref{eq:Probabilistic_Dual_From} is invariant under triple twirling: 
if
\begin{align}
    \{P_{A}, Q_{A}, X_{AC}, Y_{AB}\}
\end{align}
is feasible, then so is 
\begin{align}
    \{\mT(P_{A}), \mT(Q_{A}), \mT(X_{AC}), \mT(Y_{AB})\},
\end{align}
where $\mT$ is the triple-twirling defined in Eq.~\eqref{eq:Twirling}.
\end{mylem}

\begin{proof}
    Similar to our Lem.~\ref{lem:Dual_Formulation_TTI}, as $\Gamma$ is invariant under isotropic twirling, such that
    \begin{align}
        \left(U^*\otimes U\right)\cdot \Gamma_{AB} \cdot \left(U^{\T}\otimes U^{\dagger}\right)
        = \Gamma_{AB}.
    \end{align}

    If $\{P_{A}, Q_{A}, X_{AC}, Y_{AB}\}$ is a feasible solution, its unitary conjugate under any $U$, namely 
    \begin{align}
        \{
        U^* P_{A} U^{\T}, 
        U^* Q_{A} U^{\T}, 
        \left(U^*\otimes U\right)\cdot X_{AC} \cdot \left(U^{\T}\otimes U^{\dagger}\right),
        \left(U^*\otimes U\right)\cdot Y_{AB} \cdot \left(U^{\T}\otimes U^{\dagger}\right)
        \}
    \end{align}
    is also feasible. 
\end{proof}

As a result, the operators can be replaced by their twirled counterparts, which take the form
\begin{align}
    \mathcal{T}(P_A) &= \frac{1+c}{d} \mathbb{1}_A,\\
    \mathcal{T}(Q_A) &= \frac{1-c}{d} \mathbb{1}_A,\\
    \mathcal{T}(X_{AC}) &= x_1 \mathbb{1}_A + x_2 \Gamma,\\
    \mathcal{T}(Y_{AB}) &= y_1 \mathbb{1}_A + y_2 \Gamma.
\end{align}
The minimal sample complexity is then obtained from the following equivalent formulation.
\begin{align}
    s_{2}^{\text{Dual}}(p)
    =
    \max \quad 
    & 
    p x_1 d + px_2d^2 + p y_1 d + py_2 d^2 - c 
    \\
    \text{s.t.} \quad 
    & 1-c\geqslant 0,\\
    & 1+c\geqslant 0,\\
    & c\geqslant 0,\\
    & - \frac{1-c}{d}\mathbb{1}_{ABC} \leqslant (x_1 + y_1) \mathbb{1}_{ABC} + y_2 \Gamma_{AB} \otimes \mathbb{1}_C + x_2 \mathbb{1}_B \otimes \Gamma_{AC} \leqslant \frac{1+c}{d}\mathbb{1}_{ABC}.
\end{align}

In addition, the problem is symmetric under exchanging the $B$ and $C$ systems, which implies that the objective function remains invariant under the swap $x_2$ with $y_2$.
Consequently, without loss of generality, we may introduce the variables
\begin{align}
    x &:= \frac{x_2+y_2}{2},\\
    \alpha &:= x_1 + y_1,
\end{align}
and operator $M(x,x)$ 
\begin{align}
    M(x,x) = x \qty(\Gamma_{AB} \otimes \mathbb{1}_C + \mathbb{1}_B \otimes \Gamma_{AC})
\end{align}
as defined in Eq.~\eqref{eq:M_xy}.
With these definitions, the dual problem can be simplified to the following form.

\begin{mylem}
{Simplified Dual Formulation of Probabilistic Broadcasting}
{PVB_Simplified_Dual_Formulation} 
The sample complexity $s_{2}^{\text{Dual}}(p)$ associated with probabilistic virtual broadcasting in Eq.~\eqref{eq:pvb_sdp_dual} admits the following simplified dual formulation
\begin{align}\label{eq:pvb_sdp_dual}
    s_{2}^{\text{Dual}}(p)
    = \max \quad 
    &  p(\alpha d + 2x d^2) - c \\
    \text{s.t.} \quad 
    &\qty(\frac{c - 1}{d} - \alpha)\1_{abc}
    \leqslant
    M(x,x)
    \leqslant
    \qty(\frac{c + 1}{d} - \alpha)\1_{abc}, \\ \label{ln:M_bound}
    &0 \leqslant c \leqslant 1.
\end{align}
\end{mylem}

The constraints in Eq.~\eqref{ln:M_bound} is equivalent to
\begin{equation}
    \frac{c-1}{d}-\alpha \leqslant \lambda_{\min}(M(x,x)), \quad\text{and}\quad \lambda_{\max}(M(x,x)) \leqslant \frac{c+1}{d}-\alpha.
\end{equation}
where $\lambda_{\max}(M(x,x))$ and $\lambda_{\min}(M(x,x))$ are the maximal and minimal eigenvalues of $M(x,x)$ respectively.
As discussed in Subsec.~\ref{subsec:M_xy}, $M(x,x)$ has $3$ eigenvalues: $\lambda_-$, $\lambda_+$ and $0$, where $\lambda_\pm$ are given in Eq.~\eqref{eq:lambda_pm} that
\begin{equation}
   \lambda_{\pm}=
    \frac{(x+y)d\pm\sqrt{(x-y)^2d^2+4xy}}{2},
\end{equation}
where $y = x$ in this case.
That is to say $\lambda_\pm = xd\pm |x|$.

In case $x \geqslant 0$, the constraints reduce to
\begin{equation}
    \frac{c-1}{d}-\alpha \leqslant 0, \quad\text{and}\quad x(d+1) \leqslant \frac{c+1}{d}-\alpha.
\end{equation}
In case $x \leqslant 0$, the constraints instead take the form
\begin{equation}
    \frac{c-1}{d}-\alpha \leqslant x(d+1), \quad\text{and}\quad 0 \leqslant \frac{c+1}{d}-\alpha.
\end{equation}
Combining the two cases above, the problem reduces to the following Linear Programming (LP).

\begin{mylem}
{Probabilistic Broadcasting via LP}
{PB_LP} 
The sample complexity $s_{2}^{\text{Dual}}(p)$ associated with probabilistic virtual broadcasting in Eq.~\eqref{eq:PVB_SDP} is determined by the larger of $s_{2,1}^{\text{Dual}}(p)$ and $s_{2,2}^{\text{Dual}}(p)$; that is
\begin{align}
    s_{2}^{\text{Dual}}(p):=\max\{s_{2,1}^{\text{Dual}}(p),s_{2,2}^{\text{Dual}}(p)\}
\end{align}
where $s_{2,1}^{\text{Dual}}(p)$ and $s_{2,2}^{\text{Dual}}(p)$ denote the optimal values of the following two linear programmings 
\begin{equation}
\begin{aligned}
    s_{2,1}^{\text{Dual}}(p):=\max \quad 
    &  \alpha pd + 2pd^2x - c \\
    \text{s.t.} \quad 
    &\frac{c - 1}{d} \leqslant \alpha \\
    &\frac{c + 1}{d} - \alpha \geqslant x (d+1) \geqslant 0 \\
    &0 \leqslant c \leqslant 1.
\end{aligned}
\quad\quad\text{and}\quad\quad
\begin{aligned}
    s_{2,2}^{\text{Dual}}(p):=\max \quad 
    &  \alpha pd + 2pd^2x - c \\
    \text{s.t.} \quad 
    &\frac{c + 1}{d} \geqslant \alpha \\
    &\frac{c - 1}{d} - \alpha \leqslant x (d+1) \leqslant 0 \\
    &0 \leqslant c \leqslant 1.
\end{aligned}
\end{equation}
\end{mylem}

Since the problem reduces to a linear program, the optimum is attained at the boundary of the feasible region.
For $s_{2,1}^{\text{Dual}}(p)$, maximizing over $x$ gives
\begin{align}
    \max \quad 
    &  \alpha pd - \frac{2pd^2 \alpha}{d+1} + \frac{2pd (c+1)}{d+1} - c \\
    \text{s.t.} \quad 
    &\frac{c - 1}{d} \leqslant \alpha, \\
    &0 \leqslant c \leqslant 1.
\end{align}
The coefficient of $\alpha$ in the objective function is $pd(1-d)/(d+1)\leqslant0$. 
Consequently, the optimal value is attained at the minimal feasible $\alpha$, namely $\alpha=(c-1)/d$.
Substituting this relation simplifies the optimization to
\begin{align}
    \max \quad 
    &  p(c-1) + \frac{4pd}{d+1} - c \\
    \text{s.t.} \quad 
    &0 \leqslant c \leqslant 1.
\end{align}
A final maximization over $c$ yields the optimal value 
\begin{align}
    s_{2,1}^{\text{Dual}}(p)=p\left(\frac{3d-1}{d+1}\right).
\end{align}

For $s_{2,2}^{\text{Dual}}(p)$, maximizing over $x$ leads to
\begin{align}
    \max \quad 
    &  \alpha pd - c, \\
    \text{s.t.} \quad 
    &\alpha \leqslant \frac{c+1}{d}, \\
    &0 \leqslant c \leqslant 1.
\end{align}
Since the objective function increases with $\alpha$, with the optimum is attained at the largest feasible value, i.e., $\alpha=(c+1)/d$. 
Substituting this relation reduces the problem to
\begin{align}
    \max \quad 
    &  -(1-p)c+p, \\
    \text{s.t.} \quad 
    &0 \leqslant c \leqslant 1.
\end{align}
The optimal value is obtained at $c=0$, yielding
\begin{align}
    s_{2,2}^{\text{Dual}}(p)=p.
\end{align}

Now, by Lem.~\ref{lem:PB_LP}, the sample complexity $s_{2}^{\text{Dual}}(p)$ is given by
\begin{align}
    s_{2}^{\text{Dual}}(p)=\max\left\{
    p\left(\frac{3d-1}{d+1}\right), p
    \right\}.
\end{align}
Here $d$ denotes the system dimension. 
For quantum systems one always has $d\geqslant2$, which implies that $(3d-1)/(d+1)>1$. 
Thus, $s_{2,1}^{\text{Dual}}(p)$ is strictly larger than $s_{2,2}^{\text{Dual}}(p)$, i.e., $s_{2,1}^{\text{Dual}}(p)>s_{2,2}^{\text{Dual}}(p)$, leading to the following theorem.

\begin{mythm}{Analytical Solution for Probabilistic Virtual Broadcasting}{Analytical_Solution_PVB}
    The minimal sample complexity $s_{2}^{\text{Dual}}(p)$ for probabilistic virtual broadcasting, characterized by the SDP in Eq.~\eqref{eq:PVB_SDP}, with success probability $p$, takes the closed form
    \begin{equation}\label{eq:Probabilistic_SDP_Opt}
        s_{2}^{\text{Dual}}(p) = p\left(\frac{3d-1}{d+1}\right).
    \end{equation}
    The optimal value therefore scales linearly with the success probability $p$.
\end{mythm}

The strong duality between the primal formulation $s_{2}(p)$ of probabilistic virtual broadcasting in Eq.~\eqref{eq:PVB_SDP} and the dual formulation $s_{2}^{\text{Dual}}(p)$ in Eq.~\eqref{eq:pvb_sdp_dual} can be verified directly and is therefore omitted. 
Theorem~\ref{thm:Analytical_Solution_PVB} also provides an alternative route to recover the main result of~\cite{PhysRevA.110.012458} as a special case. 
In particular, when the success probability is set to $p=1$, the probabilistic protocol reduces to deterministic virtual broadcasting without error. 
The optimal sample complexity overhead then reduces to $(3d-1)/(d+1)$.

%%%%%%%%%%%%%%%%%%%%%%%%%%%%%%%%%%%%%%%%%%%%%%%%%%%%%%%%%%%%%%%%%%%%%%%%%%%%%%%%%%%%%%%%%%%%%%%%%%%%%%%%%%%%%%%%%%%%%%

\subsection{No Practical Virtual Broadcasting: Even Probabilistically}
\label{subsec:SE_PB}

When virtual protocols are considered, the accompanying increase in sample complexity must also be taken into account. 
Without loss of generality, consider a scenario in which the sender Alice distributes quantum information to two receivers, Bob and Claire. 
Each receiver subsequently performs a local estimation task, characterized by observables $\mO_{B}$ and $\mO_{C}$, respectively. 
Suppose that Bob and Claire are required to satisfy their own error thresholds, more precisely, an $\epsilon_1-\delta_1$ test for Bob and an 
$\epsilon_2-\delta_2$ test for Claire. 
In this case, the total number of copies that Alice must prepare is
\begin{align}
    n_Q:=\max\{n_1, n_2\},
\end{align}
where the individual sample requirements $n_1$ and $n_2$ are given by
\begin{align}\label{eq:n_i}
    n_i=\mO\left(\frac{c^2}{\epsilon_i^2}
    \ln{\frac{1}{\delta_i}}\right),\quad i\in\{1,2\}.
\end{align}
A detailed derivation of these quantities is presented in Subsec.~\ref{subsec:SAnalysis_SComplexity}.

A natural benchmark is provided by the most direct strategy for distributing quantum information. 
In this approach, Alice simply prepares a sufficient number of identical copies of the quantum state and sends them independently to Bob and Claire. 
Ensuring that both receivers pass their respective $\epsilon-\delta$ tests then requires at least $n_1+n_2$ copies in total. 
For clarity in the discussion below, this naive strategy will be represented by the simplified cost $2n_Q$, which serves as the reference sample complexity against which the virtual protocol will be compared.

We now return to the probabilistic protocol considered here. 
The analysis developed for the deterministic case cannot be applied directly here, and a modification of the results in Subsec.~\ref{subsec:SAnalysis_SComplexity} is required.
To illustrate the idea, consider Bob's side. 
In the case of exact broadcasting, given an initial state $\rho$, Bob's task is to estimate $\Tr[\rho\mO]$. 
In the present setting, however, the quantity $p\Tr[\rho\mO]$ is estimated through a virtual subchannel (see Eq.~\eqref{eq:PVB_SDP}). 
If Bob is still able to pass the $\epsilon-\delta$ test, the event that the estimation error exceeds $\epsilon$ can be written as
\begin{align}
    |\overline{X}-p\Tr[\rho\mO]|\geqslant\epsilon,
\end{align}
where $\overline{X}$ denotes the statistical average of the measurement outcomes obtained from measuring $\Tr[\Tr_C[\mE_{A\to BC}(\rho)]\mO]$ on Bob's subsystem. 
When $p\neq0$, above inequality can be recast as
\begin{align}
    |\frac{1}{p}\overline{X}-\Tr[\rho\mO]|\geqslant\frac{\epsilon}{p}.
\end{align}
If we wish to maintain the error tolerance at $\epsilon$, we can replace the original threshold $\epsilon/p$ with a new parameter $\epsilon$. 
This modification increases the required sample size from
\begin{align}
    \mO\left(\frac{c^2}{\epsilon^2}
    \ln{\frac{1}{\delta}}\right)
\end{align}
to
\begin{align}
    \frac{1}{p^2}\mO\left(\frac{c^2}{\epsilon^2}
    \ln{\frac{1}{\delta}}\right).
\end{align}
On the other hand, $\overline{X}/p$ arises from the following virtual operations.
\begin{align}
    \frac{1}{p}(a\mE_1-b\mE_2).
\end{align}
The subchannels $\mE_1$ and $\mE_2$ can always be extended to the corresponding quantum channels $\mF_1$ and $\mF_2$ by assigning the failure outcome the label 0. 
This procedure leaves the statistical average obtained in the experiment unchanged. 
The measurement process can therefore be regarded as arising from the following operations.
\begin{align}
    \frac{1}{p}(a\mF_1-b\mF_2)=\frac{a+b}{p}\left(\frac{a}{a+b}\mF_1-\frac{b}{a+b}\mF_1\right).
\end{align}
The maximal difference between measurement outcomes is therefore updated from $c$ to $c(a+b)/p$, which in turn yields the following sample complexity required required on Bob's side for the probabilistic protocol.
\begin{align}
    \frac{(a+b)^2}{p^4}\mO\left(\frac{c^2}{\epsilon^2}
    \ln{\frac{1}{\delta}}\right). 
\end{align}
The same argument applies to Claire's side. 
We are now ready to present the following lemma, which characterizes the sample complexity associated with probabilistic virtual broadcasting.

\begin{mylem}
{Sample Complexity of Probabilistic Virtual Protocols}
{SC_PVP}
If Bob and Claire require $n_1$ and $n_2$ copies, respectively, to pass their $\epsilon-\delta$ test, and we denote their maximum by $n_Q$, then under the probabilistic protocol the minimal sample complexity $n_{\text{prob}}(p)$ becomes
\begin{align}
    n_{\text{prob}}(p)=\frac{s_{2}^2(p)}{p^4}n_Q,
\end{align}
where $p$ denotes the success probability and $s_{2}$ is given in Eq.~\eqref{eq:PVB_SDP}.
\end{mylem}

Whenever the following condition 
\begin{align}
    n_{\text{prob}}(p)<2n_Q
\end{align}
is satisfied, the corresponding probabilistic virtual broadcasting protocol will be referred to as sample efficient (SE).

Theorem~\ref{thm:Analytical_Solution_PVB} provides a closed-form expression for the optimal sample-complexity overhead of probabilistic virtual broadcasting, which allows $s_{2}^2(p)/p^4$ to be evaluated explicitly. 
Its minimal value is attained at $p=1$, namely the case of deterministic virtual broadcasting. 
Even in this most favorable case, however, the overhead $(3d-1)^2/(d+1)^2$ remains strictly larger than 2, as discussed in Eq.~\eqref{eq:Upper_Bound_No_VBroadcasting}. 
Probabilistic virtual broadcasting therefore never satisfies the SE condition.
The situation becomes even less favorable when the success probability is reduced, as the required overhead increases further. 
This observation leads naturally to the following no Practical probabilistic virtual broadcasting theorem.

\begin{mythm}
{No Practical Virtual Broadcasting: Even Probabilistically}
{No_Practical_Probabilistic_Virtual_Broadcasting} 
1-to-2 probabilistic virtual broadcasting cannot satisfy the sample efficiency (SE) condition and is therefore inherently impractical.
\end{mythm}

%%%%%%%%%%%%%%%%%%%%%%%%%%%%%%%%%%%%%%%%%%%%%%%%%%%%%%%%%%%%%%%%%%%%%%%%%%%%%%%%%%%%%%%%%%%%%%%%%%%%%%%%%%%%%%%%%%%%%%

\section{General 1-to-$N$ Quantum Broadcasting}
\label{sec:N_QB}

So far, most of the analysis has focused on 1-to-$2$ virtual broadcasting, in which a single sender distributes quantum information to two receivers. 
This section extends the framework to the general 1-to-$N$ setting, where the sender broadcasts to $N$ receivers, and investigates the corresponding sample complexity overhead in both the approximate and probabilistic regimes. 
In particular, the framework of 1-to-$N$ approximate virtual broadcasting is developed in Subsec.~\ref{subsec:N_AVB}, while 1-to-$N$ probabilistic virtual broadcasting is analyzed in Subsec.~\ref{subsec:N_PVB}.
Some concluding remarks and discussions on virtual broadcasting are presented in Subsec.~\ref{subsec:Closing_Remarks}.

%%%%%%%%%%%%%%%%%%%%%%%%%%%%%%%%%%%%%%%%%%%%%%%%%%%%%%%%%%%%%%%%%%%%%%%%%%%%%%%%%%%%%%%%%%%%%%%%%%%%%%%%%%%%%%%%%%%%%%

\subsection{General 1-to-$N$ Approximate Virtual Broadcasting}
\label{subsec:N_AVB}

In this subsection, we extend the analysis of approximate virtual broadcasting from two receivers to the general setting of $N$ receivers.
We derive the associated sample complexity overhead and compare its performance with the naive strategy of preparing sufficiently many copies and distributing them directly to all receivers.

More precisely, consider a sender Alice who broadcasts her initial quantum state to $N$ receivers, denoted $B_1$, $\ldots$, $B_N$, via a virtual broadcasting map $\mE:A\to B_1\cdots B_N$, where each receiver is allowed an error $\epsilon_i$ with $i\in\{1,\ldots,N\}$. 
Without loss of generality, we assume that the marginal channel from Alice to receiver $B_i$ is a depolarizing channel $\mD_{p_i}$ (see Subsec.~\ref{subsec:Channel_Twirling}). 
\begin{align}
    \mD_{p_i}(\cdot)=p_{i}\frac{\1_{B_1}}{d}\otimes\Tr_{A}[\cdot]+(1-p_i)\id_{A\to B_i}(\cdot),
\end{align}
with
\begin{align}\label{eq:p_i_N}
    p_i=\frac{d^2 \epsilon_i}{d^2-1}, \quad
    \forall\, i\in\{1, \cdots, N\}.
\end{align}
Under this assumption, the overall sample complexity overhead is characterized by

\begin{align}\label{eq:SDP_Approximate_v_N}
    v_N(p_1, \cdots, p_N):=
    \min \quad 
    & a+b\\
    \text{s.t.} \quad 
    &\1_{B_1\cdots B_{i-1}B_{i+1}\cdots B_{N}}\star(J_1-J_2)=
    \frac{p_i}{d}\,\1_{AB_{i}}
    +
    (1-p_i) \Gamma_{AB_{i}},
    \quad i\in\{1,\ldots,N\}
    \label{eq:Approx_C_v_i}\\
    & J_1\geqslant0,\,\, 
    \Tr_{BC}[J_1]=a\,\1_A,\label{eq:J_1_v_N}\\
    &J_2\geqslant0,\,\,
    \Tr_{BC}[J_2]=b\,\1_A,\label{eq:J_2_v_N}\\
    &a-b=1.\label{eq:VB_TP_Approximate_v_N}
\end{align}
Here $\star$ denotes the link product (see Def.~\ref{def:Link_Product}).
Throughout, all systems are taken to have the same dimension $d$, namely, $d=\dim A=\dim B_1=\cdots=\dim B_N$.

Remark that the sample complexity in Eq.~\eqref{eq:SDP_Approximate} can in principle be extended directly to the general 1-to-$N$ broadcasting scenario. 
However, by employing the same symmetry arguments used in Subsec.~\ref{subsec:Marginals}, twirling can be applied to simplify the analysis, replacing the symmetric group $\mathfrak{S}_3$ with $\mathfrak{S}_{N+1}$.
We therefore begin with this simplified formulation of Eq.~\eqref{eq:SDP_Approximate_v_N}, in which the approximation condition is already captured by depolarizing marginals, rather than by the broadcasting fidelity defined in Eq~\eqref{eq:Approx_C}.

We now examine whether practical virtual quantum broadcasting can emerge in the approximate protocol for 1-to-$N$ broadcasting by analyzing the corresponding broadcasting rate. 
Following Subsec.~\ref{subsec:SE_AB}, let $n_i$ denote the number of copies required for receiver $B_i$ to pass the $\epsilon-\delta$ test under quantum operations, and define 
\begin{align}\label{eq:n_Q_N}
    n_Q:=\max_i\{n_i\}.
\end{align}
The broadcasting rate for 1-to-$N$ broadcasting is then given by
\begin{align}
    R(\epsilon_1,\cdots, \epsilon_N):=\frac{F_{\text{Bcast}}(\mE)}{n(\epsilon_1, \cdots, \epsilon_N)}
    =
    \frac{\frac{1}{N}\left((1-\frac{d^2-1}{d^2}p_1)+\cdots+(1-\frac{d^2-1}{d^2}p_N)\right)}{v^2_N(p_1, \cdots, p_N)n_Q}
    =
    \frac{1-\frac{\sum_i \epsilon_i}{N}}{v^2_N(p_1, \cdots, p_N)n_Q}.
\end{align}

\begin{figure}
    \centering
    \includegraphics[width=1\linewidth]{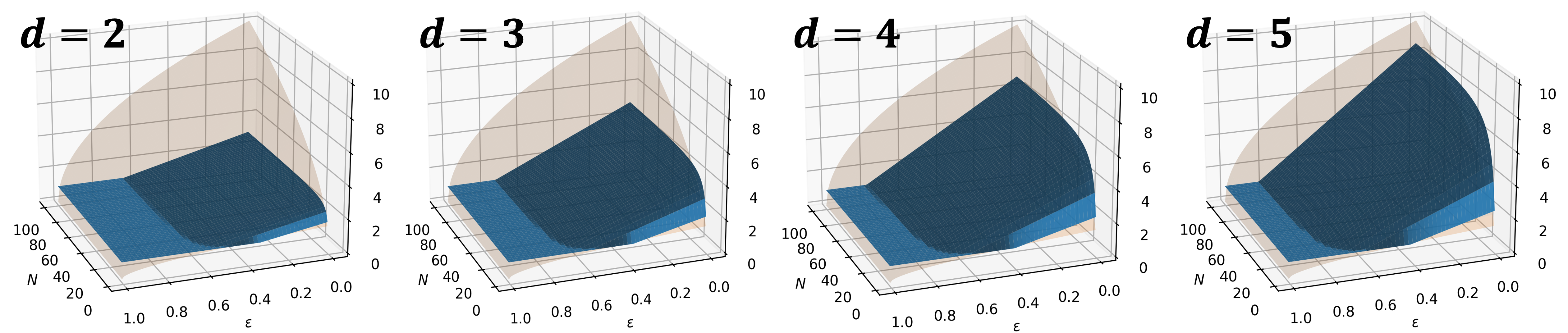}
    \caption{(Color online) \textbf{Sample Efficient Approximate Virtual Broadcasting}.
    For varying receiver number $N$ and uniform error $\epsilon$ across all receivers, we plot the corresponding minimal sample complexity overhead $v_N(p_1, \cdots, p_N)$, given by the left-hand side of Eq.~\eqref{eq:ABC_Numerical_N}, with $p_i=d^2\epsilon/(d^2-1)$ as defined in Eq.~\eqref{eq:p_i_N}, shown in blue, alongside the sample complexity of the naive prepare-and-distribute protocol, shown in grey. The regimes in which sample efficiency is achieved are highlighted in dark blue, for system dimensions $d=2$ to $4$.
    }
    \label{fig:4copy}
\end{figure}

By contrast, the naive protocol consists of Alice preparing $Nn_Q$ copies of the initial state and distributing $n_Q$ copies to each receiver. 
The broadcasting rate $R_{\text{naive}}$ can be expressed as
\begin{align}
    R_{\text{naive}}=\frac{\frac{1}{N}(1+\cdots+1)}{n_Q+\cdots+n_Q}=\frac{1}{Nn_Q}.
\end{align}
Since $n_Q>0$, the existence of a practical advantage in sample efficiency is equivalent to verifying the following inequality
\begin{align}
    \frac{1-\frac{\sum_i \epsilon_i}{N}}{v^2_N(p_1, \cdots, p_N)}
    \geqslant
    \frac{1}{N},
\end{align}
equivalently,
\begin{align}\label{eq:ABC_Numerical_N}
    v_N(p_1, \cdots, p_N)\leqslant
    \sqrt{N-\sum_i \epsilon_i}.
\end{align}
Numerical results for general 1-to-$N$ approximate virtual broadcasting are shown in Fig.~\ref{fig:4copy}.
To facilitate comparison and clearly illustrate Eq.~\eqref{eq:ABC_Numerical_N}, we assume identical error tolerances across all receivers.
We then vary $N$ to determine the regimes in which sample efficiency is achieved.

%%%%%%%%%%%%%%%%%%%%%%%%%%%%%%%%%%%%%%%%%%%%%%%%%%%%%%%%%%%%%%%%%%%%%%%%%%%%%%%%%%%%%%%%%%%%%%%%%%%%%%%%%%%%%%%%%%%%%%

\subsection{General 1-to-$N$ Probabilistic Virtual Broadcasting}
\label{subsec:N_PVB}

We now consider probabilistic virtual broadcasting in the 1-to-$N$ setting. 
Following the approach developed in Sec.~\ref{sec:PQB}, we fix a success probability and analyze the resulting sample complexity overhead. 
More precisely, for a given success probability $p$, this overhead can be characterized by the following SDP

\begin{align}\label{eq:PVB_SDP_N}
    s_{N}(p)
    =
    \min \quad 
    & 
    a+b
    \\
    \text{s.t.} \quad 
    &\1_{B_1\cdots B_{i-1}B_{i+1}\cdots B_{N}}\star (J_1-J_2)=p\Gamma_{AB_i},\\
    &J_1\geqslant0,\,\, \1_{B_1\cdots B_N}\star J_1\leqslant a\1_{A},\\
    &J_2\geqslant0,\,\, \1_{B_1\cdots B_N}\star J_2\leqslant b\1_{A},\\
    &a-b\leqslant1.
\end{align}
The constraint $a-b\leqslant1$ guarantees that the overall virtual operation is trace-nonincreasing (TNI), and hence describes a probabilistic operation.

The technique used in Subsec.~\ref{subsec:PVB_Analytical_Solution} extends directly to the 1-to-$N$ probabilistic virtual broadcasting.
The associated Lagrangian for Eq.~\eqref{eq:PVB_SDP_N} takes the form

\begin{align}
    \mathcal{L} :=& (a+b) + (a-b-1)c \\
    & - \sum_i \Tr_{AB_i}\qty[(\Tr_{B_i\cdots B_{i-1}B_{i+1}\cdots B_N}[J_1 - J_2]) - p \Gamma_{AB_i}) \cdot X_{AB_i}^{(i)}]
    - \Tr[J_1\cdot E]
    - \Tr[J_2\cdot F] \\
    & + \Tr_{A}\qty[ \left(\Tr_{B_1\cdots B_N}[J_1]-a\1_{A} \right) \cdot P_A] +
    \Tr_{A}\qty[ \qty(\Tr_{B_1 \cdots B_N}[J_2]-b\1_{A}) \cdot Q_A ],
\end{align}
where $X_{AB_i}^{(i)}$ are Hermitian operators, $E_{AB_1\cdots B_N}, F_{AB_1\cdots B_N}, P_A$ and $Q_A$ are positive semidefinite, and real number $c \geqslant 0$.
Optimizing the Lagrangian over the primal variables yields the corresponding dual programming:

\begin{align}
    s_{N}^{\text{Dual}}(p)
    =
    \max \quad 
    & 
    p \sum_i \Tr[X_{AB_i} \cdot \Gamma_{AB_i}] - c 
    \\
    \text{s.t.} \quad 
    & \Tr[P_A] = 1 + c, \\
    & \Tr[Q_A] = 1 - c, \\
    & - Q_A \otimes \mathbb{1}_{B_1 \cdots B_N} \leqslant \sum_i \Gamma_{AB_i} \otimes \mathbb{1}_{B_1\cdots B_{i-1}B_{i+1}\cdots B_N} \leqslant P_{A} \otimes \mathbb{1}_{B_1 \cdots B_N}, \\
    & P_A \geqslant 0, \quad Q_A \geqslant 0, \\
    & c \geqslant 0.
\end{align}

Moreover, following an argument similar to that in Lem.~\ref{lem:Probabilistic_Dual_Formulation_TTI}, the dual problem is invariant under unitary twirling.
The operators $P_A, Q_A$ and $X^{(i)}_{AB_i}$ may therefore be replaced by their twirled counterparts without affecting the optimal value.
\begin{align}
    \mathcal{T}(P_A) &= \frac{1+c}{d} \mathbb{1}_A,\\
    \mathcal{T}(Q_A) &= \frac{1-c}{d} \mathbb{1}_A,\\
    \mathcal{T}(X_{AB_i}^{(i)}) &= x_i \mathbb{1}_A + y_i \Gamma_{AB_i}.
\end{align}
Owing to the symmetry of the problem under permutations of the systems $B_i$ and $B_j$, the variables $x_i$ and $y_i$ can be replaced by their averages
\begin{align}
    x &:= \frac{\sum_i x_i}{N},\\
    y &:= \frac{\sum_i y_i}{N}.
\end{align}
The dual problem then reduces to the following simplified form: 

\begin{mylem}
{Simplified Dual Formulation of 1-to-$N$ Probabilistic Virtual Broadcasting}
{VPB_SDF_N} 
The sample complexity $s_{N}^{\text{Dual}}(p)$ 
associated with 1-to-$N$ probabilistic virtual broadcasting, defined by the SDP in Eq.~\eqref{eq:PVB_SDP_N}, admits the following simplified dual formulation:
\begin{align}\label{eq:pvb_sdp_dual_N}
    s_{N}^{\text{Dual}}(p)
    = \max \quad 
    &  p(Ndx + Nd^2y) - c \\
    \text{s.t.} \quad 
    &\qty(\frac{c - 1}{d} - Nx)\1_{abc}
    \leqslant
    yZ
    \leqslant
    \qty(\frac{c + 1}{d} - Nx)\1_{abc}, \\ \label{ln:N_spectral_constraint}
    &0 \leqslant c \leqslant 1,
\end{align}
where the operator $Z$ is defined as
\begin{align}\label{eq:Z}
    Z := \sum_i \Gamma_{AB_i} \otimes \mathbb{1}_{B_1\cdots B_{i-1}B_{i+1}\cdots B_N}.
\end{align}
Here $\Gamma_{AB_i}$ is the unnormalized maximally entangled state (see Eq.~\eqref{eq:UMES}) shared between systems $A$ and $B_i$, $p$ represents the success probability, and $d$ denotes the system dimension, namely $d=\dim A=\dim B_i$ for all $i\in\{1, \cdots, N\}$.
\end{mylem}

It is straightforward to verify the strong duality between the primal and dual formulations of the SDP for 1-to-$N$ probabilistic virtual broadcasting.
The constraints in Eq.~\eqref{ln:N_spectral_constraint} are equivalent to
\begin{align}
    &\frac{c - 1}{d} - Nx \leqslant \min\left\{y \lambda_{\min}(Z), y\lambda_{\max}(Z)\right\}, \\
    &\frac{c + 1}{d} - Nx \geqslant \max\left\{y \lambda_{\max}(Z), y\lambda_{\min}(Z)\right\},
\end{align}
where $\lambda_{\min}(Z)$ and $\lambda_{\max}(Z)$ denote the minimum and maximum eigenvalues of the operator $Z$.
Since $Z$, is a sum of positive semidefinite operators, it follows that
\begin{align}
    \lambda_{\min}(Z)\geqslant0.
\end{align}
Furthermore,
\begin{align}
Z\left(\ket{0}_{A}\ket{1}_{B_1}\cdots\ket{1}_{B_N}\right)=\mathbf{0},
\end{align}
as each term $\Gamma_{AB_i} \otimes \mathbb{1}_{B_1\cdots B_{i-1}B_{i+1}\cdots B_N}$ annihilates this state.
This implies $\lambda_{\min}(Z)=0$. 
Here $\mathbf{0}$ denotes the zero vector.
The dual problem therefore reduces to

\begin{mylem}
{1-to-$N$ Virtual Probabilistic Broadcasting via LP}
{PB_LP_N} 
The sample complexity $s_{N}^{\text{Dual}}(p)$ associated with 1-to-$N$ probabilistic virtual broadcasting in Eq.~\eqref{eq:PVB_SDP_N} is determined by the larger of two contributions, $s_{N,1}^{\text{Dual}}(p)$ and $s_{N,2}^{\text{Dual}}(p)$, namely
\begin{align}
    s_{N}^{\text{Dual}}(p):=\max\left\{s_{N,1}^{\text{Dual}}(p),s_{N,2}^{\text{Dual}}(p)\right\}.
\end{align}
Here $s_{N,1}^{\text{Dual}}(p)$ and $s_{N,2}^{\text{Dual}}(p)$ are the optimal values of the following two linear programmings (LPs): 
\begin{equation}
\begin{aligned}
    s_{N,1}^{\text{Dual}}(p):=\max \quad 
    &  pN(dx + d^2y) - c \\
    \text{s.t.} \quad 
    &\frac{c - 1}{d} \leqslant Nx, \\
    &\frac{c + 1}{d} - Nx \geqslant y \lambda_{\max}(Z) \geqslant 0, \\
    &0 \leqslant c \leqslant 1.
\end{aligned}
\quad\quad\text{and}\quad\quad
\begin{aligned}
    s_{N,2}^{\text{Dual}}(p):=\max \quad 
    &  pN(dx + d^2y) - c \\
    \text{s.t.} \quad 
    &\frac{c + 1}{d} \geqslant Nx, \\
    &\frac{c - 1}{d} - Nx \leqslant y \lambda_{\max}(Z) \leqslant 0, \\
    &0 \leqslant c \leqslant 1.
\end{aligned}
\end{equation}
Here $p$ is the success probability, and $d$ the system dimension, with $d=\dim A=\dim B_i$ for all $i\in\{1, \cdots, N\}$.
\end{mylem}

Since the problem reduces to a linear programming, the optimum is attained at the boundary of the feasible region.
For $s_{N,1}^{\text{Dual}}(p)$, maximizing over $y$ yields
\begin{align}
    \max \quad 
    &  pN\qty(dx + d^2 \qty( \frac{c+1}{d\lambda_{\max}(Z)} - x\frac{N}{\lambda_{\max}(Z)})) - c \\
    \text{s.t.} \quad 
    &\frac{c - 1}{d} \leqslant Nx, \\
    &0 \leqslant c \leqslant 1.
\end{align}
The coefficient of $x$ in the objective function is $pdN(1-dN)/\lambda_{\max} \leqslant0$. 
Hence the maximum is attained at the smallest feasible value of $x$, namely $x=(c-1)/{dN}$.
Substituting this relation into the objective gives
\begin{align}
    pN\qty(dx + d^2 \qty( \frac{c+1}{d\lambda_{\max}(Z)} - x\frac{N}{\lambda_{\max}(Z)})) - c
    &= pdN\frac{c-1}{dN} + pd^2N \qty( \frac{c+1}{d\lambda_{\max}(Z)} - \frac{c-1}{dN}\frac{N}{\lambda_{\max}(Z)}) - c \\
    &= p(c-1) + pd^2N \qty( \frac{c+1}{d\lambda_{\max}(Z)} - \frac{c-1}{d\lambda_{\max}(Z)}) - c \\
    &= p(c-1) + \frac{2pdN}{\lambda_{\max}(Z)} - c.
\end{align}
Therefore the optimization problem is equivalent to
\begin{align}
    \max \quad 
    &  c(p-1) - p + \frac{2pdN}{\lambda_{\max}(Z)} \\
    \text{s.t.} \quad 
    &0 \leqslant c \leqslant 1.
\end{align}
A final maximization over $c$ gives the optimal value 
\begin{align}
    s_{N,1}^{\text{Dual}}(p)=p\left(\frac{2dN-\lambda_{\max}(Z)}{\lambda_{\max}(Z)}\right).
\end{align}

For $s_{2,2}^{\text{Dual}}(p)$, maximizing over $y$ leads to
\begin{align}
    \max \quad 
    &  pdNx - c, \\
    \text{s.t.} \quad 
    &Nx \leqslant \frac{c+1}{d}, \\
    &0 \leqslant c \leqslant 1.
\end{align}
Since the objective function increases monotonically with $x$, the optimum is attained at the largest feasible value, i.e., $x=(c+1)/dN$. 
Substituting this relation reduces the problem to
\begin{align}
    \max \quad 
    &  -(1-p)c+p, \\
    \text{s.t.} \quad 
    &0 \leqslant c \leqslant 1.
\end{align}
The maximum is achieved at $c=0$, yielding
\begin{align}
    s_{N,2}^{\text{Dual}}(p)=p.
\end{align}

Combining these results with Lem.~\ref{lem:PB_LP_N}, the sample complexity $s_{N}^{\text{Dual}}(p)$ is therefore given by
\begin{align}
    s_{N}^{\text{Dual}}(p)=\max\left\{
    p\left(\frac{2dN-\lambda_{\max}(Z)}{\lambda_{\max}(Z)}\right), p
    \right\}.
\end{align}
Here $p$ is the success probability, and $d$ denotes the system dimension.

It worth notice that, each term in the sum $Z = \sum_i \Gamma_{AB_i} \otimes \mathbb{1}_{B_1\cdots B_{i-1}B_{i+1}\cdots B_N}$ is positive semidefinite with largest eigenvalue $d$.
As a result, the maximum eigenvalue $\lambda_{\max}(Z)$ of $Z$ is bounded above by $Nd$.
It follows that $(2dN-\lambda_{\max}(Z))/\lambda_{\max}(Z)\geqslant1$, and
\begin{align}
    p\left(\frac{2dN-\lambda_{\max}(Z)}{\lambda_{\max}(Z)}\right)\geqslant p.
\end{align}
Therefore, the minimal sample complexity overhead $s_{N}^{\text{Dual}}(p)$ is obtained as
\begin{align}
    s_{N}^{\text{Dual}}(p)
    =
    \max\left\{s_{N,1}^{\text{Dual}}(p),s_{N,2}^{\text{Dual}}(p)\right\}
    =
    p\left(\frac{2dN-\lambda_{\max}(Z)}{\lambda_{\max}(Z)}\right).
\end{align}

Several questions remain unresolved for 1-to-$N$ probabilistic virtual broadcasting.
In particular, the maximal eigenvalue of the operator $Z$ defined in Eq.~\eqref{eq:Z} is not known, and therefore a closed-form expression for $s_{N}^{\text{Dual}}(p)$ is still lacking.
To make further progress, it is instructive to revisit the analysis in Subsec.~\ref{subsec:PVB_Analytical_Solution}.
There we observed that, for probabilistic quantum broadcasting implemented solely with quantum channels, the case $p=1$, which reduces to the deterministic setting, reveals that the sample complexity of exact broadcasting appears as the coefficient multiplying $p$. 
This observation provides an important clue for resolving the present problem.
Fortunately, a recent work~\cite{yao2025quantifyingunextendibilityvirtualstate} establishes a closed-form solution for general 1-to-$N$ virtual broadcasting in the exact regime. 
This result is summarized in the following lemma.

\begin{mylem}
{Sample Complexity of 1-to-$N$ Exact Virtual Broadcasting~\cite{yao2025quantifyingunextendibilityvirtualstate}}
{SC_EVB_N}
For systems of dimension $d$, with $\dim A=\dim B_i=d$ for all $i\in\{1, \cdots, N\}$, the sample complexity of 1-to-$N$ exact virtual broadcasting defined as $v_N(0, \cdots, 0)$ (see Eq.~\eqref{eq:SDP_Approximate_v_N}), whose optimal value is given by
\begin{align}
    v_N(0, \cdots, 0)=\frac{2dN}{N+d-1}-1.
\end{align}
\end{mylem}

By Lem.~\ref{lem:SC_EVB_N}, setting $p=1$ recovers exact virtual broadcasting, and $s_{N}^{\text{Dual}}(p)$ satisfies
\begin{align}
    s_{N}^{\text{Dual}}(1)=v_N(0, \cdots, 0),
\end{align}
which is equivalent to
\begin{align}
    \frac{2dN-\lambda_{\max}(Z)}{\lambda_{\max}(Z)}
    =
    \frac{2dN}{N+d-1}-1.
\end{align}
This relation implies
\begin{align}
    \lambda_{\max}(Z)=N+d-1,
\end{align}
thus resolving the open question identified above.
We can now obtain the analytic solution for 1-to-$N$ probabilistic virtual broadcasting, which is characterized by the following theorem.

\begin{mythm}{Analytical Solution for 1-to-$N$ Probabilistic Virtual Broadcasting}{Analytical_Solution_PVB_N}
    The minimal sample complexity $s_{N}^{\text{Dual}}(p)$ (see Eq.~\eqref{eq:pvb_sdp_dual_N}) for 1-to-$N$ probabilistic virtual broadcasting with success probability $p$ is given by
    \begin{equation}\label{eq:Probabilistic_SDP_Opt_N}
        s_{N}^{\text{Dual}}(p) = p\left(\frac{2dN}{N+d-1}-1\right).
    \end{equation}
Here $d$ denotes the system dimension, namely $d=\dim A=\dim B_i$ for all $i\in\{1, \cdots, N\}$.
\end{mythm}

\begin{figure}
    \centering
    \includegraphics[width=1\linewidth]{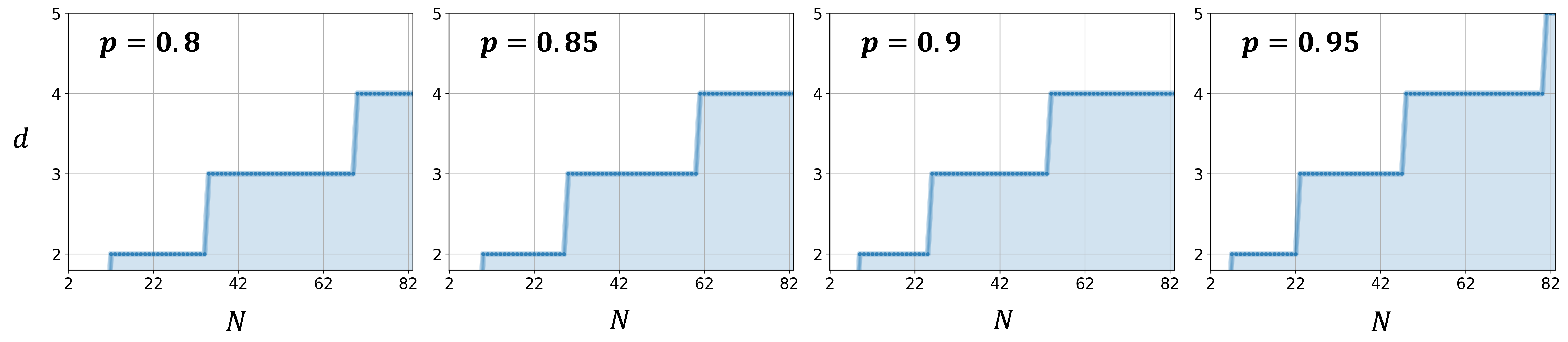}
    \caption{(Color online) \textbf{Sample Efficient Probabilistic Virtual Broadcasting}.
    For a given success probability $p$ from 0.8 to 0.95, varying the system dimension $d$ leads to different minimal scales $N$ required for 1-to-$N$ probabilistic virtual broadcasting to become practical, that is, to satisfy the sample efficiency condition of Eq.~\eqref{eq:n_prob_N}.
    }
    \label{fig:prob4copy}
\end{figure}

We now examine whether 1-to-$N$ probabilistic virtual broadcasting can satisfy the condition of sample efficiency (SE). 
To this end, we compare the required sample complexity with that of the naive strategy.
Suppose that receiver $B_i$ requires $n_i$ copies of the initial state in order to pass their $\epsilon-\delta$ test, and denote the maximum among them by $n_Q$ (see Eq.~\eqref{eq:n_Q_N}).
In the naive protocol, Alice prepares $Nn_Q$ copies of the initial state and distributes $n_Q$ copies to each receiver, ensuring that all receivers pass their $\epsilon-\delta$ tests. 
The discussion then turns to the probabilistic virtual broadcasting protocol.
The modification of the measurement outcomes rescales the coefficient $c$ in Eq.~\eqref{eq:c} by the factor $s^2_{N}(p)/p^2$, where $s^2_{N}(p)$ corresponds to the optimal value of the SDP in Eq.~\eqref{eq:PVB_SDP_N}. 
Accounting for the success probability $p$ of the protocol introduces an additional overhead $1/p^2$, so that the overall sample complexity reads
\begin{align}
    n_{\text{prob}}(p):=\frac{s^2_{N}(p)n_Q}{p^4}.
\end{align}
A probabilistic protocol is regarded as more efficient than the naive protocol whenever the number of samples required for its implementation is smaller than that required by the naive strategy, namely,
\begin{align}\label{eq:n_prob_N}
    n_{\text{prob}}(p)<Nn_Q.
\end{align}
Since $n_Q>0$, the problem reduces to determining when $N-(s^2_{N}(p)/p^4)>0$. 
Using Thm.~\ref{thm:Analytical_Solution_PVB_N}, the coefficient of $n_{\text{prob}}(p)$ can be computed directly as
\begin{align}
    \frac{s^2_{N}(p)}{p^4}
    =
    \frac{1}{p^2}\left(\frac{(2N-1)d-(N-1)}{d+(N-1)}\right)^2
    .
\end{align}
Here we consider only the case of nonzero success probability.
The SE condition for probabilistic virtual broadcasting, namely that this quantity be smaller than $N$, can thus be written as
\begin{align}\label{eq:prob_SE_constraint}
    -\frac{(N-1)(p\sqrt{N}-1)}{2N+p\sqrt{N}-1}< 
    d<
    \frac{(N-1)(p\sqrt{N}+1)}{2N-p\sqrt{N}-1}.
\end{align}
Remark that we implicitly assume $N\geqslant2$, as the broadcasting task involves at least two receivers.
Equation~\eqref{eq:prob_SE_constraint} leads to the following corollary.

\begin{mycor}
{Practical 1-to-$N$ Probabilistic Virtual Broadcasting}
{Practical_PVB_N}
A 1-to-$N$ probabilistic virtual broadcasting protocol with success probability $p$ can be sample efficient (SE) only if the system dimension satisfies
\begin{align}
    -\frac{(N-1)(p\sqrt{N}-1)}{2N+p\sqrt{N}-1}< 
    d<
    \frac{(N-1)(p\sqrt{N}+1)}{2N-p\sqrt{N}-1}.
\end{align}
\end{mycor}

The above corollary highlights a key distinction between conventional quantum broadcasting and the practical broadcasting framework considered here, where sample efficiency (SE) plays a central role.
In standard quantum information theory, broadcasting is typically studied in the single-shot setting and restricted to quantum channels. In this regime, the impossibility of 1-to-2 broadcasting immediately implies the impossibility of 1-to-$N$ broadcasting for any $N\geqslant2$.
Indeed, if a 1-to-$N$ broadcasting map existed, one could simply trace out the additional output systems to obtain a 1-to-2 broadcasting map. 
This elementary fact can be summarized as 
\begin{align}\label{eq:Implies_Quantum_BC}
    \nexists\,\,\text{1-to-2 Quantum Broadcasting}
    \implies
    \nexists\,\,\text{1-to-$N$ Quantum Broadcasting}.
\end{align}
This conclusion no longer holds once sample complexity requirements, captured by the notion of sample efficiency (SE), are taken into account. 
In Thm.~\ref{thm:No_Practical_Probabilistic_Virtual_Broadcasting}, we proved that no 1-to-2 probabilistic virtual broadcasting protocol exists for any success probability, including the deterministic case $p=1$. 
However, as will be shown later through numerical experiments of Cor.~\ref{cor:Practical_PVB_N},
probabilistic virtual broadcasting can exist when $N$ is sufficiently large.
Figure~\ref{fig:prob4copy} illustrates the minimal $N$ required to guarantee sample efficiency across different success probabilities and system dimensions.

\begin{figure}[t]
    \centering
    \includegraphics[width=1\linewidth]{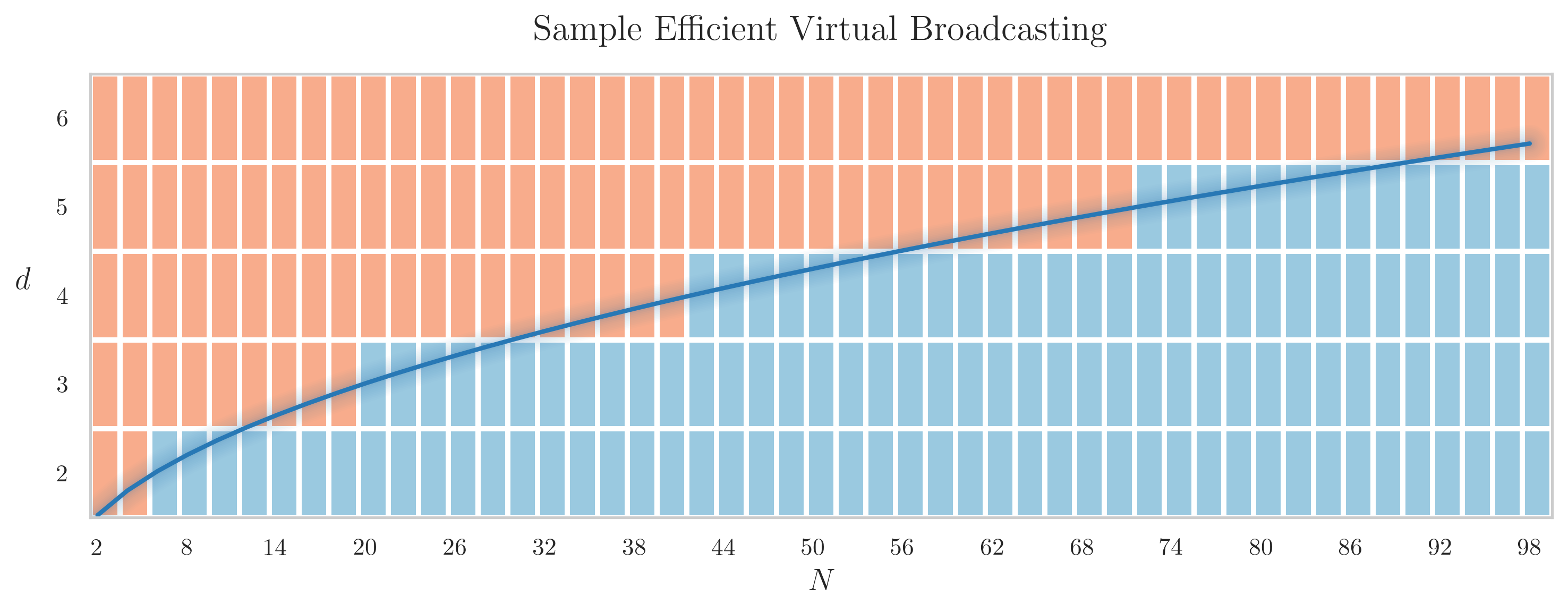}
    \caption{(Color online) \textbf{Sample Efficient Virtual Broadcasting}.
    Numerical illustration of the relation between $d$ and $N$ in Eq.~\eqref{eq:Exact_VBC_SE}.
    The red region corresponds to parameter regimes where 1-to-$N$ virtual broadcasting with system dimension $d$ violates the sample efficiency (SE) condition, while the blue region corresponds to regimes where the SE condition holds.}
    \label{fig:efficient_virtual_broadcasting}
\end{figure}

To simplify the discussion, we focus on the deterministic case $p=1$. 
In this setting, the SE condition (see Eq.~\eqref{eq:n_prob_N}) becomes
\begin{align}\label{eq:N_d_p_1}
    \left(\frac{(2N-1)d-(N-1)}{d+(N-1)}\right)^2<N.
\end{align}
This condition is equivalent to requiring that
\begin{align}
    -\frac{(\sqrt{N}-1)^2}{2\sqrt{N}-1}
    <d<
    \frac{(\sqrt{N}+1)^2}{2\sqrt{N}+1}.
\end{align}
Since the left-hand side is smaller than 1, the condition becomes trivial, yielding the following simplified form.
\begin{align}\label{eq:Exact_VBC_SE}
d<\frac{(\sqrt{N}+1)^2}{2\sqrt{N}+1}.
\end{align}
Note that the right-hand side of the above inequality increases monotonically with $N$. 
For any fixed system dimension $d$, a sufficiently large $N$ always allows virtual broadcasting while satisfying the sample efficiency (SE) condition. 
For example, in the qubit case, i.e., $d=2$, Eq.~\eqref{eq:Exact_VBC_SE} implies the condition $N>3+2\sqrt{2}$. 
Since $N$ represents the number of receivers in virtual broadcasting and must be an integer, the smallest admissible value is $N=6$.
Consequently, the SE condition rules out 1-to-2, 1-To-3, 1-To-4, and 1-To-5
broadcasting for qubit systems. 
By contrast, virtual broadcasting becomes feasible for any 1-to-$N$ scenario with $N\geqslant6$. 
This illustrates a striking feature of the SE constraint: small-scale broadcasting is excluded, whereas sufficiently large-scale virtual broadcasting remains permitted.
A similar behavior arises in higher dimensions. 
For $d=3$, the constraint becomes $N>10+4\sqrt{6}$, yielding $N\geqslant20$.
For $d=4$, one obtains $N>21+12\sqrt{3}$, corresponding to $N\geqslant42$. 
Additional examples can be obtained directly from Eq.~\eqref{eq:Exact_VBC_SE} and are illustrated in Fig.~\ref{fig:efficient_virtual_broadcasting}.
We summarize this observation
\begin{align}\label{eq:Implies_Practical_BC}
    \nexists\,\,\text{1-to-2 Practical Broadcasting}
    \nRightarrow
    \nexists\,\,\text{1-to-$N$ Practical Broadcasting}.
\end{align}
Here 1-to-$N$ practical broadcasting does not generally imply 1-to-2 practical broadcasting, since tracing out additional systems does not reduce the number of samples required to realize the virtual operation.
A numerical illustration of the relation between $d$ and $N$ in Eq.~\eqref{eq:Exact_VBC_SE} is shown in Fig.~\ref{fig:efficient_virtual_broadcasting}.

%On the other hand, for a given success probability $p$ and system dimension $d$, the minimal number of receivers $N$ required for PBC to become practical, namely, to satisfy SE, can be determined from Eq.~\eqref{eq:n_prob_N}. 
%To this end, we first reformulate Eq.~\eqref{eq:n_prob_N}.
%\begin{align}\label{eq:n_prob_N_re}
%    \frac{1}{p^2}
%    \left(\frac{(2d-1)N-(d-1)}{N+(d-1)}\right)^2<N.
%\end{align}
%Observe that, in the above expression, both the numerator and the denominator within the parentheses are strictly positive. 
%It follows that

%%%%%%%%%%%%%%%%%%%%%%%%%%%%%%%%%%%%%%%%%%%%%%%%%%%%%%%%%%%%%%%%%%%%%%%%%%%%%%%%%%%%%%%%%%%%%%%%%%%%%%%%%%%%%%%%%%%%%%

\subsection{Closing Remarks}
\label{subsec:Closing_Remarks}

We conclude by examining which properties can be preserved in the construction of practical virtual broadcasting and which cannot, thereby uncovering the fundamental connections among them.

\begin{figure}
    \centering
    \includegraphics[width=1\linewidth]{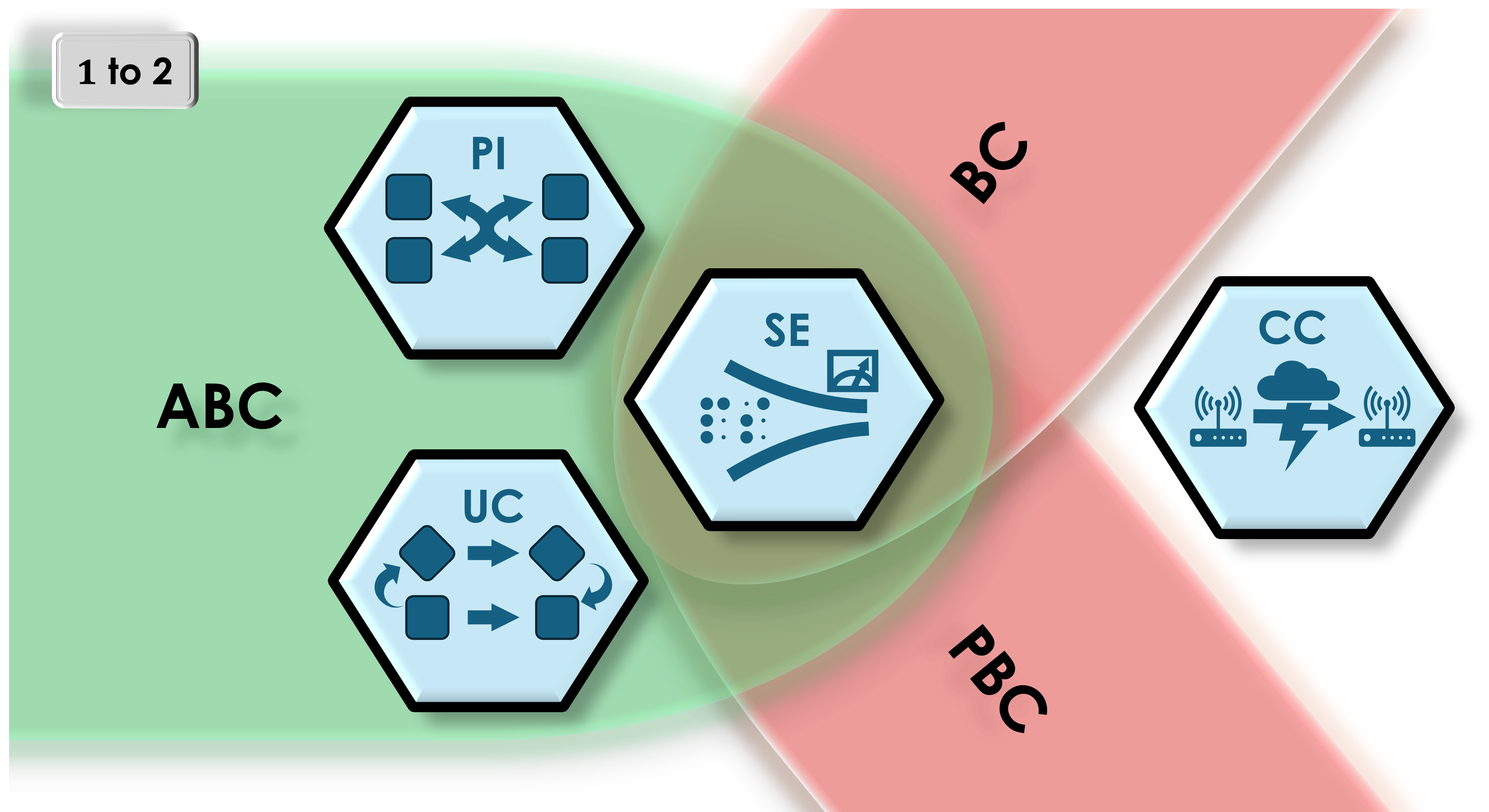}
    \caption{(Color online) \textbf{Virtual Broadcasting and Associated Properties (1-to-2 Case)}.
    For 1-to-2 virtual broadcasting, no practical protocol exists for deterministic cases, namely, no protocol satisfies SE (see Def.~\ref{def:Sample_Efficiency}), for either exact broadcasting (BC) or probabilistic broadcasting (PBC). 
    By contrast, when errors are allowed, approximate virtual broadcasting (ABC) can satisfy SE. Although BC lies at the intersection of ABC and PBC, the diagram is not intended as a Venn diagram and BC is therefore not depicted within either region. 
    Red indicates incompatibility between properties, whereas green denotes compatibility. 
    Beyond SE, ABC can additionally be compatible with UC (see Def.~\ref{def:Unitary_Covariance}) and PI (see Def.~\ref{def:Permutation_Invariance}).}
    \label{fig:Set_1to2}
\end{figure}

So far, we have considered three types of broadcasting maps: 
standard broadcasting, approximate broadcasting, and probabilistic broadcasting, which we denote by BC, ABC, and PBC, respectively. 
Clearly, BC constitutes the common core of both ABC and PBC, namely
\begin{align}
    \text{BC}\subset\text{ABC},
    \quad\text{and}\quad
    \text{BC}\subset\text{PBC}.
\end{align}
Indeed, ABC reduces to BC when the error parameters vanish, while PBC recovers BC when the success probability is set to $p=1$.

The four conditions typically considered are SE (see Def.~\ref{def:Sample_Efficiency}), UC (see Def.~\ref{def:Unitary_Covariance}), PI (see Def.~\ref{def:Permutation_Invariance}), and CC (see Def.~\ref{def:Classical_Consistency});
their precise definitions are given in Subsec.~\ref{subsec:Conditions}. 
In previous work~\cite{z2pr-zbwl,8g6j-w7ld} by some of the present authors, it was shown that for general linear maps no transformation can simultaneously satisfy all four conditions (see Thm.~\ref{thm:no_Qbroadcasting}). 
This indicates that the full set of constraints is overly restrictive.
One must therefore either relax some of them or focus on a smaller subset that captures the essential operational requirements.

For broadcasting tasks, BC itself is indispensable. 
At the same time, without the SE condition, broadcasting cannot be regarded as operationally meaningful.
It is therefore natural to begin with the minimal set consisting of SE and BC. 
However, even under these two conditions, we find that 1-to-2 virtual broadcasting already fails to be practical (see Thm.~\ref{thm:no_Vbroadcasting}). 
This observation motivates the exploration of two alternative routes, namely ABC (see Sec.~\ref{sec:AQB}) and PBC (see Sec.~\ref{sec:PQB}).

Turning to ABC, we find that it remains compatible with the SE constraint even in the 1-to-2 setting. 
Remarkably, the optimal approximate virtual broadcasting map also satisfies the UC condition (see Cor.~\ref{cor:UC_Optimality}).
Furthermore, when the error tolerances on the two receivers are chosen to be identical, the protocol automatically fulfills the PI condition as well. 
Thus, for 1-to-2 virtual broadcasting there exist quantum dynamics that simultaneously satisfy ABC, SE, UC, and PI, as demonstrated in Fig.~\ref{fig:Set_1to2}. 
The analysis can be extended to more general 1-to-$N$ configurations, which we investigate in Subsec.~\ref{subsec:N_AVB}.

\begin{figure}
    \centering
    \includegraphics[width=1\linewidth]{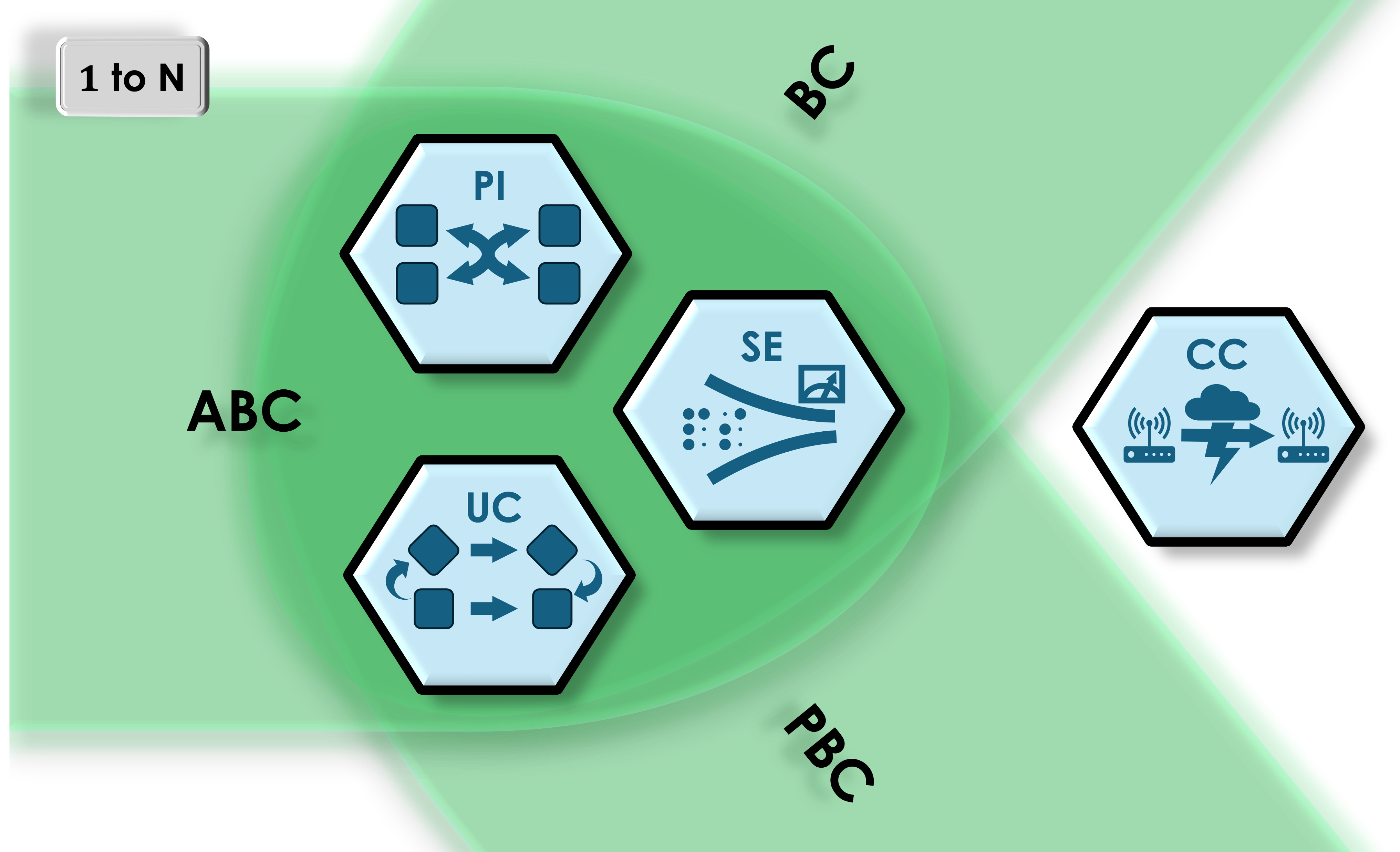}
    \caption{(Color online) \textbf{Virtual Broadcasting and Associated Properties (1-to-N Case)}.
    For 1-to-$N$ virtual broadcasting, once the number of receivers $N$ exceeds the threshold determined by Eq.~1, all classes of protocols, including exact (BC), approximate (ABC), and probabilistic (PBC) broadcasting, become compatible with SE (see Def.~\ref{def:Sample_Efficiency}), UC (see Def.~\ref{def:Unitary_Covariance}), and PI (see Def.~\ref{def:Permutation_Invariance}). CC (see Def.~\ref{def:Classical_Consistency}), however, remains incompatible.
    }
    \label{fig:Set_1toN}
\end{figure}

A different picture emerges for PBC, where the structure becomes more constrained.
We have shown that if a probabilistic broadcasting protocol exists, the success probability must be independent of the input state of the local system (see Lem.~\ref{lem:Local_Success_Probability}) and must also be identical for all receivers (see Lem.~\ref{lem:Global_Success_Probability}).
Focusing on the 1-to-2 case, we prove that no practical PBC protocol can be compatible with the SE condition (see Thm.~\ref{thm:No_Practical_Probabilistic_Virtual_Broadcasting}). 
Strikingly, unlike many other tasks in quantum information processing, introducing probabilistic protocols does not improve practicality; 
instead, it further restricts the practicality of virtual broadcasting (see Lem.~\ref{lem:SC_PVP}).
The situation changes, however, once the number of receivers becomes sufficiently large (see Eq.~\ref{eq:Exact_VBC_SE}). 
In that regime probabilistic broadcasting can become feasible, revealing a counterintuitive feature when sample complexity is taken into account. 
Although 1-to-2 broadcasting is impossible for any system dimension and any success probability, larger configurations can overcome this limitation. For example, in the qubit case even deterministic protocols become feasible once the number of receivers reaches $N=6$. 
This behavior suggests that the fundamental constraints of quantum theory suppress small-scale distribution of quantum information while allowing sufficiently large-scale virtual broadcasting.
Note that, the optimal PBC protocol can, without loss of generality, be taken to satisfy the UC (see Lem.~\ref{lem:Probabilistic_Dual_Formulation_TTI}) and PI (see Lem.~\ref{lem:Global_Success_Probability}) conditions, as illustrated in Fig.~\ref{fig:Set_1toN}.

%%%%%%%%%%%%%%%%%%%%%%%%%%%%%%%%%%%%%%%%%%%%%%%%%%%%%%%%%%%%%%%%%%%%%%%%%%%%%%%%%%%%%%%%%%%%%%%%%%%%%%%%%%%%%%%%%%%%%%

\end{document}